\documentclass{article}
\usepackage{lmodern}
\usepackage{geometry}
\usepackage[utf8]{inputenc} % allow utf-8 input
\usepackage[T1]{fontenc}    % use 8-bit T1 fonts
\usepackage{url}            % simple URL typesetting
\usepackage{booktabs}       % professional-quality tables
\usepackage{amsfonts}       % blackboard math symbols
\usepackage{nicefrac}       % compact symbols for 1/2, etc.
\usepackage{microtype}      % microtypography
\usepackage{natbib}

\usepackage{mathtools}
\usepackage{amsmath}
\usepackage{amsthm,bbm}
\usepackage{color}
\usepackage{graphicx}
\usepackage{setspace}
\usepackage{esint}
\usepackage{soul}

\mathtoolsset{showonlyrefs,showmanualtags}

\usepackage{caption}
\usepackage{subcaption}

\usepackage[unicode=true,bookmarks=true,bookmarksnumbered=false,bookmarksopen=false,breaklinks=false,pdfborder={0 0 0},pdfborderstyle={},backref=page,colorlinks=true]{hyperref}
\hypersetup{pdftitle={Coupled Half t},pdfauthor={Biswas Bhattacharya Jacob Johndrow},linkcolor=RoyalBlue,citecolor=RoyalBlue}
\usepackage[dvipsnames,svgnames,x11names,hyperref]{xcolor}

\newtheorem{Th}{Theorem}[section]
\newtheorem{Lemma}[Th]{Lemma}
\newtheorem{Cor}[Th]{Corollary}
\newtheorem{Prop}[Th]{Proposition}

\newtheorem{Remark}[Th]{Remark}

\newcommand{\Uniform}{\mathrm{Uniform}}
\newcommand{\InvGamma}{\mathrm{InvGamma}}

\usepackage[ruled]{algorithm2e}
\usepackage{graphicx}
\usepackage{graphics}

\title{Coupling-based convergence assessment 
  of some Gibbs samplers 
  for high-dimensional Bayesian regression with shrinkage priors}
\author{Niloy Biswas \thanks{Department of Statistics, Harvard University. Email: niloy\_biswas@g.harvard.edu} 
\and Anirban Bhattacharya \thanks{Department of Statistics, Texas A\&M University. Email: anirbanb@stat.tamu.edu}
\and Pierre E. Jacob \thanks{Department of Statistics, Harvard University. Email: pjacob@fas.harvard.edu}
\and James E. Johndrow \thanks{Department of Statistics, The Wharton School, University of Pennsylvania. Email: johndrow@wharton.upenn.edu}}

\begin{document}

\maketitle

\begin{abstract}
We consider Markov chain Monte Carlo (MCMC) algorithms for Bayesian high-dimensional regression with 
continuous shrinkage priors. A common challenge with these
algorithms is the choice of the number of iterations to perform. This is
critical when each iteration is expensive, as is the
case when dealing with modern data sets, such as genome-wide association
studies with thousands of rows and up to hundred of thousands of columns.  We
develop coupling techniques tailored to the setting of high-dimensional
regression with shrinkage priors, which enable practical, non-asymptotic
diagnostics of convergence without relying on traceplots or long-run
asymptotics. By establishing geometric drift and minorization
conditions for the algorithm under consideration, we prove that 
the proposed couplings have finite expected meeting time.  
Focusing on a class of shrinkage priors which includes the
``Horseshoe'', we empirically demonstrate the scalability of the proposed
couplings. A highlight of our findings is that less
than 1000 iterations can be enough for a Gibbs sampler to reach stationarity
in a regression on $100,000$ covariates.  The numerical results also 
illustrate the
impact of the prior on the computational efficiency of the coupling,
and suggest the use of priors where the local precisions are
Half-t distributed with degree of freedom larger than one.
\end{abstract}

\section{Introduction} \label{section:introduction}

\subsection{Iterative computation in high dimensions}

We consider the setting of high-dimensional regression where the
number of observations $n$ is smaller than the number of covariates $p$ and 
the true signal is sparse. This problem formulation is ubiquitous in modern
applications ranging from genomics to the social
sciences. Optimization-based methods, such as the LASSO
\citep{tibshirani1994regression} or 
%adaptive LASSO \citep{zou2006theadaptive},
Elastic Net \citep{zou2005regularization}, allow sparse point estimates to be obtained even when the number of covariates
is on the order of hundreds of thousands. More specifically, iterative
optimization procedures to obtain these estimates are practical 
because the following conditions are met: 1) the cost per iteration scales favorably with
the size of the input ($n$ and $p$), 2)  the number of iterations to
convergence also scales favorably, and 3) there are
reliable stopping criteria to detect convergence.

In the Bayesian paradigm, Markov chain Monte Carlo (MCMC) methods are commonly
used to sample from the posterior distribution. 
%A na\"{i}ve application of MCMC
%algorithms to high-dimensional settings, 
%for example using general-purpose software
%such as Stan \citep{carpenter2017stan},
%can lead to high computational costs,
%both per iteration and in terms of the number of iterations required 
%to reach convergence.
Default MCMC implementations, for example using general-purpose software
such as Stan \citep{carpenter2017stan}, can lead to high computational costs in high-dimensional settings,
both per iteration and in terms of the number of iterations required 
to reach convergence. In the setting of high-dimensional regressions, 
tailored algorithms 
can provide substantial improvements on both fronts
\citep[e.g.][]{bhattacharya2016fastBIOMETRIKA,nishimura2018prior,johndrow2020scalableJMLR,abjjwiley}.
Comparatively less attention has been put on the design of reliable stopping criteria.
Stopping criteria for MCMC, such as
univariate effective sample size  \citep[ESS,
e.g][]{vats2020multivariateBIOMETRIKA}, or the $\hat{R}$
%{gelman1992inferenceSS, brooks1998generalJCGS, vehtari2020rankBA,
%vats2020revisitingSS}
convergence diagnostic \citep[e.g.][]{vats2020revisitingSS}, 
rely on the asymptotic behavior of the chains as
time goes to infinity, and thus effectively ignore the non-asymptotic
``burn-in'' bias, which vanishes as time progresses.
This is acceptable in situations where we have solid a
priori estimates of the burn-in period; otherwise the lack of
non-asymptotic stopping criteria poses an important practical problem.
With Bayesian
analysis of high-dimensional data sets, MCMC algorithms can
require seconds or minutes per iteration, and thus any 
non-asymptotic insight on the number of iterations to perform
is highly valuable.
%have run-times of minutes per 
%iteration when the number of covariates is on the 
%order of hundreds of thousands. 

Diagnostics of convergence are particularly 
useful when considering the number of factors
that affect the performance 
of MCMC algorithms for Bayesian regressions.  
%We note that the MCMC run-time per iteration may still be significant in such large-scale 
%settings. For this GWAS example on a 2015 Macbook Pro,
%each iteration of the coupled chain takes approximately $1$ minute, and running 
%one coupled chain until meeting can take one to two days. This run-time is dominated 
%by the calculation of the weighted matrix cross-product $X Diag(\eta)^{-1} X^T$, 
%which has $\mathcal{O}(n^2p)$ cost. In other cases with smaller $n$ such 
%that $n^2 \leq p$, the run-time may be dominated by other parts of the 
%algorithm. 
Beyond the size of the data set, 
the specification of the
prior and the (much less explicit) signal-to-noise ratio in the data all
have an impact on the convergence of MCMC algorithms. 
Across applications, the number of iterations required for the chain to converge
varies by multiple orders of magnitude. 
This could lead users to either waste computational resources
by running overly long chains, 
or, worrisomely, to base their analysis on chains that have not converged.
This manuscript proposes a concrete method to avoid 
these pitfalls in the case of a Gibbs sampler for Bayesian regression
with heavy-tailed shrinkage priors. Specifically, we follow
the approach of \citet{glynn2014exact,jacob2020unbiasedJRSSB,biswas2019estimating} and use
coupled lagged chains to monitor convergence.

\subsection{Bayesian shrinkage regression with Half-$t(\nu)$ priors} \label{subsection:problemformulation}

Consider Gaussian linear regression with $n$ observations and $p$ covariates,
with a focus on the high-dimensional setting where $n \ll p$. The likelihood is given by
\begin{equation}
  \label{eq:linearregression}
  L(\beta, \sigma^2 ; X, y) = \frac{1}{(2 \pi \sigma^2)^{n/2}} \exp \Big( - \frac{1}{2 \sigma^2} \| y - X \beta \|^2 \Big),
\end{equation}
where $\| \cdot \|$ denotes the $L_2$ norm, $X \in \mathbb{R}^{n\times p}$ is the observed design matrix, $y \in \mathbb{R}^n$ is the observed response vector, $\sigma^2 > 0$ is the unknown Gaussian noise variance, and $\beta \in \mathbb{R}^p$ is the unknown signal vector that is assumed to be sparse. We study hierarchical Gaussian scale-mixture priors on $(\beta, \sigma^2)$ given by 
\begin{flalign} \label{local_global_prior}
  \beta_j \mid \sigma^2, \xi, \eta \underset{j=1,\ldots,p}{\overset{ind}{\sim}} \mathcal{N} \Big(0, \frac{\sigma^2}{\xi \eta_j} \Big), \ \xi \sim \pi_\xi(\cdot), \ \eta_j \underset{j=1,\ldots,p}{\overset{\text{i.i.d.}}{\sim}} \pi_\eta(\cdot), \ \sigma^2 \sim \InvGamma \bigg(\frac{a_0}{2}, \frac{b_0}{2} \bigg),
\end{flalign}
where $a_0, b_0 >0$, and $\pi_\xi(\cdot), \pi_\eta(\cdot)$ are continuous
densities on $\mathbb{R}_{>0}$. Such global-local mixture priors induce
approximate sparsity, where the components of $\beta$ can be arbitrarily close to zero,
but never exactly zero. This is in contrast to point-mass mixture priors
\citep[e.g][]{johnson2012bayesianJASA}, where some components of $\beta$
can be exactly zero a posteriori. The global precision parameter $\xi$
%in \eqref{local_global_prior} 
relates to the number of signals, and we use 
$\xi^{-1/2} \sim \text{Cauchy}_+(0,1)$ throughout.
The local precision parameters $\eta_j$ determine which components 
of $\beta$ are null. 

We focus on the Half-$t(\nu)$ prior family for the local scale parameter, $\eta_j^{-1/2} \sim  t_+(\nu)$ for $\nu \ge 1$, where a $t_+(\nu)$ distribution has density proportional to $(1+x^2/\nu)^{-(\nu+1)/2} \, \mathbbm{1}_{(0, \infty)}(x)$. The induced prior on the local precision $\eta_j$ has density
\begin{equation}\label{eq:half_t_eta}
\pi_{\eta}(\eta_j) \propto \frac{1}{\eta_j^{\frac{2-\nu}{2}}(1+\nu \eta_j)^{\frac{\nu + 1}{2}}} \, \mathbbm{1}_{(0, \infty)}(\eta_j). 
\end{equation}
The case $\nu = 1$ corresponds to the Horseshoe prior
\citep{carvalho2009handlingPMLR,carvalho2010theBIOMETRIKA} and has received
overwhelming attention in the literature among this prior class; see
\citet{bhadra2019lasso} for a recent overview. The use of Half-$t(\nu)$
priors on scale parameters in hierarchical Gaussian models was popularized by 
\citet{gelman2006priorBA}. However, the subsequent
literature has largely gravitated towards the default choice
\citep{polson2012half} of $\nu = 1$, in part because a convincing argument for
preferring a different value of $\nu$ has been absent to date. Here we give 
empirical evidence that Half-$t(\nu)$ priors yield statistical estimation
properties similar to the Horseshoe prior, whilst leading to significant
computational gains when $\nu > 1$. We note that
\citet{ghosh2017asymptoticBA} established optimal posterior concentration in
the Normal means problem for a broad class of priors on the local scale that
includes the Half-$t(\nu)$ priors, extending earlier work by
\citet{vanderpas2014theEJS} for the Horseshoe, providing theoretical support
for the comparable performance we observe with $\nu=1$ and $\nu>1$ in simulations. 

Our stopping criterion relies on couplings. Couplings have
long been used to analyze the convergence of MCMC algorithms, while
methodological implementations have been rare. The work of
\citet{valen_johnson_1996} pioneered the use of
couplings for practical diagnostics of convergence, 
but relied on the assumption of some similarity  
between the initial distribution and the target,
which can be hard to check. 
Here we instead follow the approach of 
\citet{glynn2014exact,jacob2020unbiasedJRSSB} based on couplings
of lagged chains, applied to the question of convergence in
\citet{biswas2019estimating}. That approach makes no assumptions
on the closeness of the initial distribution to the target,
and in our experiments we initialize all chains independently from the prior distribution.

The approach requires the ability to sample pairs of Markov chains such that 1)
each chain marginally follows a prescribed MCMC algorithm and 2) the two chains
``meet'' after a random --but almost surely finite-- number of iterations,
called the meeting time. When the chains meet, i.e. when all of their
components become identical, the user can stop the procedure.  From the output
the user can obtain unbiased estimates of expectations with respect to the
target distribution, which can be generated in parallel and averaged, as well
as estimates of the finite-time bias of the underlying MCMC algorithm.
Crucially, this information is retrieved from the distribution of the meeting
time, which is an integer-valued random variable irrespective of the dimension
of the problem, and thus provides a convenient summary of the performance of
the algorithm.  A difficulty inherent to the approach is the design of an
effective coupling strategy for the algorithm under consideration.
High-dimensional parameter spaces add a substantial layer of complication as
some of the simpler coupling strategies, that are solely based on maximal
couplings, lead to prohibitively long meeting times.

\subsection{Our contribution} \label{subsection:ourcontribution}

%\textcolor{red}{We might want to highlight again the main 
%  message here,
%with something like ``In this paper
%we provide substantial elements to argue that
%couplings offer a practical way of assessing the number
%of iterations required by MCMC algorithms
%in high-dimensional regression settings. 
%Specifically, our contributions are\dots }
%\textcolor{blue}{updated line added below.}

In this paper we argue that couplings offer a
practical way of assessing the number of iterations required by MCMC algorithms
in high-dimensional regression settings.  Our specific contributions are
summarized below. 

%We develop couplings of Gibbs samplers for Bayesian linear
%regression 
%\textcolor{blue}{\st{with Half-$t(\nu)$ local shrinkage priors}},
%and make several contributions.
%\begin{enumerate}
% \item \textbf{A geometrically ergodic MCMC algorithm.} 

We introduce blocked Gibbs samplers (Section 
\ref{section:single_kernel_MCMC}) for Bayesian linear regression 
focusing on Half-$t(\nu)$ local shrinkage priors, extending
the algorithm of \citet{johndrow2020scalableJMLR} for the Horseshoe.
Our algorithm has an overall computation cost of $\mathcal{O}(n^2p)$ per iteration,
which is state-of-the-art for $n \ll p$. 
We then design a \textit{one-scale} coupling strategy
for these Gibbs samplers (Section
\ref{subsec:max_coupling}) and show that it 
results in meeting times that have exponential tails, which in turn 
validates their use in convergence diagnostics \citep{biswas2019estimating}. 
Our proof of exponentially-tailed meeting times is based on identifying 
a geometric drift condition, as a corollary to which we also 
establish geometric ergodicity of the marginal Gibbs sampler.
Despite a decade of widespread use, MCMC algorithms for the Horseshoe and 
Half-$t(\nu)$ priors generally have not previously been shown to be geometrically ergodic,
in contrast to MCMC algorithms
for global-local models with exponentially-tailed local scale priors 
\citep{khare2013geometricEJS, pal2014geometricEJS}.
Our proofs utilize a uniform bound on 
generalized ridge regression estimates, which enables us to avoid a technical 
assumption of \citet{johndrow2020scalableJMLR}, and
auxiliary results on moments of distributions related to the confluent
hypergeometric function of the second kind. 
Shortly after an earlier version of this manuscript was posted, \citet{bhattacharya2021geo} have also
established geometric ergodicity of related Gibbs samplers for the Horseshoe prior.

The meeting times resulting from the one-scale coupling 
increase exponentially as dimension increases. To address that issue
we design improved coupling strategies 
(Section \ref{section:coupled_MCMC_2}).
We develop a \textit{two-scale} strategy, based
on a carefully chosen metric, which informs the choice between synchronous and
maximal couplings. In numerical experiments in high dimensions 
this strategy leads to orders of magnitude shorter meeting
times compared to the na\"ive approach. 
We establish some preliminary results in a stylized setting,
providing a partial understanding of
why the two-scale strategy leads to the observed drastic reduction in meeting times. 
We also describe other couplings that do not require explicit 
calculation of a metric between the chains,
have exponentially-tailed meeting times, and empirically result in similar meeting 
times as the two-scale coupling.
For our two-scale coupling, we further note 
that priors with stronger shrinkage towards zero 
can give significantly shorter meeting times in 
high dimensions. This motivates usage of Half-$t(\nu)$ priors with 
greater degrees of freedom $\nu > 1$, which can give similar statistical 
performance to the Horseshoe ($\nu = 1$) whilst enjoying orders of 
magnitude computational improvements.

Finally, we demonstrate (Section \ref{section:gwas}) that the proposed 
coupled MCMC algorithm is applicable to big data settings. 
Figure \ref{fig:gwas_highlight} displays 
an application of our coupling strategy to
monitor the convergence of the Gibbs sampler for Bayesian regression
on a genome-wide association study (GWAS) dataset 
with $(n,p) \approx (2000, 100000)$. 
Our method suggests that a burn-in of just $700$ iterations
can suffice even for such a large problem, and 
confirms the applicability of coupled MCMC algorithms 
for practical high-dimensional inference.

\begin{figure}[!ht]
\captionsetup[subfigure]{font=normalsize,labelfont=normalsize}
    \centering
    \includegraphics[width=0.43\textwidth]{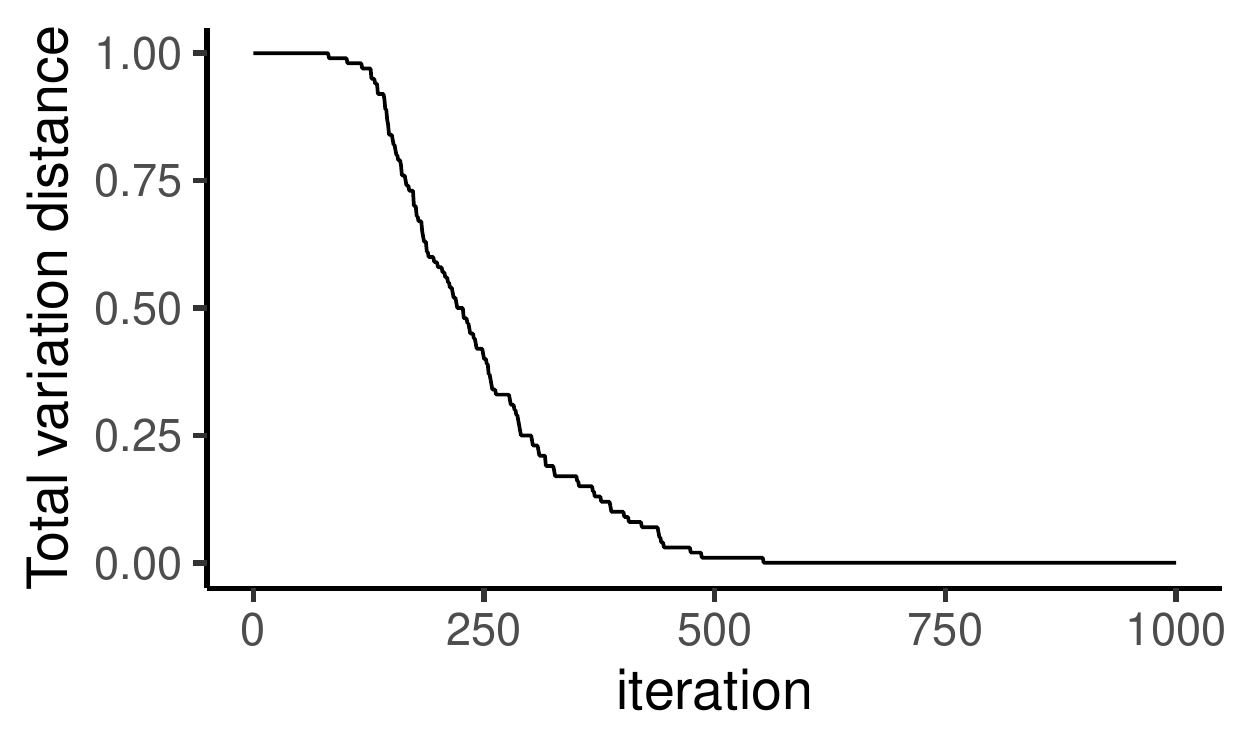}
  \caption{Upper bounds on the total variation between the marginal distributions of the chain
  and the target distribution for Bayesian regression on a GWAS dataset with 
  $(n,p) \approx (2000, 100000)$.}
  \label{fig:gwas_highlight}
\end{figure}

Scripts in R \citep{Rsoftware} are available from 
\url{https://github.com/niloyb/CoupledHalfT} to reproduce the figures of the article.

\section{Coupled MCMC for regression with Half-$t(\nu)$ priors} \label{section:single_kernel_MCMC}

We develop MCMC techniques for Bayesian shrinkage regression with
Half-$t(\nu)$ priors. The model is given by \eqref{eq:linearregression}--\eqref{local_global_prior} 
and $\eta^{-1/2}\sim t(\nu)$ for some $\nu \geq 1$ as in \eqref{eq:half_t_eta}. 
Posterior distributions resulting from such heavy-tailed shrinkage priors 
have desirable statistical properties, but their features pose challenges to 
generic-purpose MCMC algorithms. Characterization of the marginal prior and 
posterior densities of $\beta$ on $\mathbb{R}^p$ given $\xi$ and $\sigma^2$, and 
further comments on these challenges 
are presented in Appendix \ref{appendices:mcmc_challenges}.
Specifically, the resulting posterior distributions 
present 1) multimodality, 2) heavy tails and 3) infinite density at zero. 
This hints at a trade-off between statistical accuracy and computational feasibility, since the very
features that present computational difficulties are crucial for optimal
statistical performance across sparsity levels and signal strengths.

% In this section, we introduce Gibbs samplers for such targets. 
\subsection{A blocked Gibbs sampler for Half-$t(\nu)$ priors} \label{subsec:blocked_gibbs_mcmc}

Blocked Gibbs samplers are popularly used for Bayesian regression with
global-local shrinkage priors, due to the analytical tractability of full
conditional distributions
\citep{park2008theJASA,carvalho2010theBIOMETRIKA,polson2014theJRSSB,bhattacharya2015dirichletJASA,makalic2016simple}.
Several convenient blocking and
marginalization strategies are possible, leading to conditionals that are easy
to sample from.
For the case of the Horseshoe prior ($\nu=1$), \citet{johndrow2020scalableJMLR} have
developed exact and approximate MCMC algorithms for
high-dimensional settings, 
building on the algorithms of \citet{polson2014theJRSSB} and
\citet{bhattacharya2016fastBIOMETRIKA}. We extend that sampler
to general Half-$t(\nu)$ priors,
and summarize it in \eqref{eq:blocked_gibbs_ergodicity_version}. 
Given some initial state $(\beta_0, \eta_0, \sigma_0^2, \xi_0) 
\in \mathbb{R}^p \times \mathbb{R}^p_{> 0} \times \mathbb{R}_{> 0} \times \mathbb{R}_{> 0}$, 
it generates a Markov chain $(\beta_t, \eta_t, \sigma^2_t, \xi_t)_{t \geq 0}$ which targets the 
posterior corresponding to Half-$t(\nu)$ priors. 

\begin{flalign} 
\begin{split} \label{eq:blocked_gibbs_ergodicity_version}
&\pi(\eta_{t+1} | \beta_{t}, \sigma_t^2, \xi_t) 
%= \prod_{j=1}^p \pi(\eta_{t+1,j} | \beta_{t,j}, \sigma_t^2, \xi_t) 
\propto \prod_{j=1}^p \frac{e^{-m_{t,j} \eta_{t+1,j}}}{\eta_{t+1,j}^{\frac{1-\nu}{2}} (1+ \nu \eta_{t+1,j})^{\frac{\nu+1}{2}}}  \text{ for } m_{t,j} = \frac{\xi_t \beta_{t,j}^2}{2 \sigma_t^2}, \\
&\pi(\xi_{t+1} | \eta_{t+1}) \propto L(y | \xi_{t+1}, \eta_{t+1} ) \pi_{\xi}(\xi_{t+1}) \ 
%\text{ for } \ L(y | \xi_{t}, \eta_{t+1}) \text{ from Equation } \eqref{eq:Half_t_Exact_algo_ll}
, \\
&\sigma_{t+1}^2 | \eta_{t+1}, \xi_{t+1} \sim \InvGamma \bigg( \frac{a_0+n}{2}, \frac{y^T M^{-1} y + b_0}{2} \bigg) 
%\text{ for } M_{\xi, \eta} \text{ from Equation } \eqref{eq:Half_t_Exact_algo_ll}
%\text{ for } M_{\xi, \eta} = I_n + \xi^{-1} X Diag(\eta^{-1}) X^T, \\
, \\
&\beta_{t+1} | \eta_{t+1}, \xi_{t+1}, \sigma_{t+1}^2 \sim \mathcal{N}(\Sigma^{-1} X^Ty, \sigma_{t+1}^2 \Sigma^{-1}) 
%\text{ for } \Sigma^{-1} = X^TX + \xi Diag(\eta_{t+1} )
.
\end{split}
\end{flalign}

Above  $L(y | \xi_{t+1}, \eta_{t+1})$
represents the marginal likelihood of $y$ given $\xi_{t+1}$ and $\eta_{t+1}$, 
and we use the notation $M = I_n + \xi_{t+1}^{-1} X\, \text{Diag}(\eta_{t+1}^{-1})\, X^T$,
and $\Sigma = X^TX + \xi_{t+1} \text{Diag}(\eta_{t+1})$. We sample 
$\eta_{t+1} | \beta_{t}, \sigma_t^2, \xi_t$ component-wise independently using slice sampling
and sample $\xi_{t+1} | \eta_{t+1}$ using Metropolis--Rosenbluth--Teller--Hastings with 
step-size $\sigma_{\text{MRTH}}$  \citep{Hastings1970MonteCS}.
%and sample $\beta_{t+1} | \eta_{t+1}, \xi_{t+1}, \sigma_{t+1}^2$ using the algorithm 
%of \citep{bhattacharya2016fastBIOMETRIKA}. 
Full details of our sampler are in Algorithms \ref{algo:Half_t_Exact} and \ref{algo:slice_sampling_1}
of Appendix \ref{appendices:algos}, and derivations are in Appendix 
\ref{appendices:algo_derivations}. 

% \paragraph{Computation Cost.} 
Sampling $\eta_{t+1} | \beta_{t}, \sigma_t^2, \xi_t$ component-wise independently 
is an $\mathcal{O}(p)$ cost operation. Sampling $\xi_{t+1} | \eta_{t+1}$
requires evaluating $L(y | \xi_{t+1}, \eta_{t+1})$, which  
involves an $\mathcal{O}(n^2p)$ cost operation (calculating the weighted matrix
cross-product $X \text{Diag}(\eta_{t+1})^{-1} X^T$) and an $\mathcal{O}(n^3)$ cost
operation (calculating $y^T M^{-1} y$ and $|M|$ using Cholesky decomposition). Sampling 
$\sigma_{t+1}^2 | \eta_{t+1}, \xi_{t+1}$ is only an 
$\mathcal{O}(1)$ cost operation, as $y^T M^{-1} y$ 
is pre-calculated. Similarly, as $M^{-1}$ is pre-calculated, 
sampling $\beta_{t+1} | \eta_{t+1}, \xi_{t+1}, \sigma_{t+1}^2$ using the 
algorithm of \citet{bhattacharya2016fastBIOMETRIKA} for 
structured multivariate Gaussians 
%Gaussian scale mixture priors 
only involves an
$\mathcal{O}(np)$ and an $\mathcal{O}(n^2)$ cost operation. 
In high-dimensional settings with $p \gg n$, which is our focus, the 
weighted matrix cross-product calculation
is the most expensive and overall 
the sampler described in \eqref{eq:blocked_gibbs_ergodicity_version}
has $\mathcal{O}(n^2p)$ computation cost.

\subsection{Coupling the proposed Gibbs sampler} \label{subsection:coupled_MCMC}

We develop couplings for the blocked Gibbs sampler presented in Section
\ref{subsec:blocked_gibbs_mcmc}. We will use these couplings
to generate coupled chains with 
a time lag $L\geq 1$, in order to implement the diagnostics of convergence
proposed in \citet{biswas2019estimating}. Consider an $L$-lag coupled chain
$(C_t, \tilde{C}_{t-L})_{t \geq L}$ with meeting time $\tau:= \inf \{ t
\geq L : C_{t}=\tilde{C}_{t-L}\}$. Here $C_t$ and $\tilde{C}_t$ denote the full
states $(\beta_t, \eta_t, \sigma_t^2,\xi_t)$ and 
$(\tilde{\beta}_t, \tilde{\eta}_t, \tilde{\sigma}_t^2, \tilde{\xi}_t)$ respectively.
For the coupling-based methods 
proposed in \citet{biswas2019estimating,jacob2020unbiasedJRSSB} 
to be most practical, the meeting times should occur as early as possible,
under the constraint that both marginal chains $(C_t)_{t\geq 0}$ and
$(\tilde{C}_t)_{t\geq 0}$ start from the same initial distribution, 
and that they both marginally evolve according to the blocked Gibbs sampler in
\eqref{eq:blocked_gibbs_ergodicity_version}.

Briefly, we recall how $L$-lag coupled chains enable the estimation of upper
bounds on $\text{TV}(\pi_t,\pi)$. Since $C_{t}=\tilde{C}_{t-L}$ for all $t\geq
\tau$, then the indicator $\mathbbm{1}(C_t \neq \tilde{C}_{t-L})$ can be
evaluated for all $t$ as soon as we observe $\tau$. Such indicators are
estimators of $\mathbb{P}(C_t \neq \tilde{C}_{t-L})$, which themselves are
upper bounds on $\text{TV}(\pi_t, \pi_{t-L})$, since $C_t \sim \pi_t$ and
$\tilde{C}_{t-L}\sim \pi_{t-L}$. We can then obtain upper bounds on
$\text{TV}(\pi_t,\pi)$ by the triangle inequality: $\text{TV}(\pi_t,\pi)\leq
\mathbb{E}[\sum_{j\geq 0} \mathbbm{1}(C_{t+(j+1)L} \neq \tilde{C}_{t+jL})]$.
This can be justified more formally if the meeting times are
exponentially-tailed \citep{biswas2019estimating}. Such tail condition also
ensures that the expected computation time of unbiased estimators as in
\citet{jacob2020unbiasedJRSSB}, generated independently in parallel, scales at
a logarithmic rate in the number of processors.
Thus we will be particularly interested in couplings that result
in exponentially-tailed meeting times.

In Section \ref{subsec:max_coupling}, we consider an algorithm
based on maximal couplings. We show that the ensuing meeting times have 
exponential tails, and discuss how our analysis directly implies that the 
marginal chain introduced in Section \ref{subsec:blocked_gibbs_mcmc} is geometrically
ergodic. 
We then illustrate difficulties encountered by such scheme in high
dimensions, which will motivate alternate strategies
described in Section \ref{section:coupled_MCMC_2}. For simplicity, we mostly omit the lag $L$ from the notation.
%In Sections \ref{subsec:two_scale_coupling} and
%\ref{subsec:coupling_variants}, we develop a
%\textit{two-scale} coupling and other related coupling strategies
%that deliver better empirical performance.}}
%In Section \ref{subsec:degree_of_freedom}, we investigate the
%impact of the degree of freedom $\nu$ for Half-$t(\nu)$ priors on our two-scale
%coupling. We find that Half-$t(\nu)$ priors with higher degrees of freedom
%$\nu > 1$ can give similar statistical performance to the Horseshoe ($\nu=1$),
%whilst enjoying orders of magnitude computational improvements through shorter
%coupling times. 
%For simplicity, we mostly omit the lag $L$ from the notation.
%and specify it only when necessary.
%we denote $L$-lag coupled chains
%$(C^{(1)}_t, C^{(2)}_{t-L})_{t \geq L}$ as $(C^{(1)}_t, C^{(2)}_{t})_{t \geq
%0}$. We will specify the lag $L$ when necessary. 

\subsection{One-scale coupling} \label{subsec:max_coupling}
For the blocked Gibbs sampler in Section \ref{subsec:blocked_gibbs_mcmc}, we first
consider a coupled MCMC algorithm that attempts exact meetings at every step.
We will apply a maximal coupling algorithm with independent residuals (see
\citet[Chapter 1 Section 4.4]{thorisson2000regeneration} and
\citet{johnson1998coupling}). It is included in Algorithm
\ref{algo:max_coupling} of Appendix \ref{appendices:algos}, and has an expected
computation cost of two units \citep{johnson1998coupling}. 
%We write $\gamma_{max}(P, Q)$ to denote (the distribution of) this coupling between distributions $P$ and $Q$. 

Our initial coupled MCMC kernel is given in Algorithm \ref{algo:Half_t_meetings_every_step}, 
which we refer to as a \textit{one-scale} coupling. 
That is, before the chains have met (when $C_t \neq \tilde{C}_t$), the coupled kernel 
on Steps $(1)$ and $(2)(a-c)$ does not explicitly depend 
on the distance between states $C_t$ and $\tilde{C}_t$.
After meeting, the coupled chains remain together by construction, such that
$C_t=\tilde{C}_t$ implies $C_{t+1}=\tilde{C}_{t+1}$. When $C_t\neq \tilde{C}_t$, 
Step $(1)$ uses the coupled slice sampler of Algorithm
\ref{algo:slice_sampling_exact_meeting_coupling} component-wise for
$(\eta_{t+1}, \tilde{\eta}_{t+1})$, which allows each pair of components
$(\eta_{t+1,j}, \tilde{\eta}_{t+1,j})$ to meet exactly with positive
probability. In Algorithm \ref{algo:slice_sampling_exact_meeting_coupling}, we
use a common random numbers or a ``synchronous'' coupling of the auxiliary 
random variables $(U_{j,*}, \tilde{U}_{j,*})$, and alternative couplings could be
considered. Steps $(2)(a-b)$ are maximal couplings of the conditional sampling steps for
$(\xi, \tilde{\xi})$ and $(\sigma^2, \tilde{\sigma}^2)$, such that
$\big(  \sigma_{t+1}, \xi_{t+1}, \eta_{t+1} \big) =
\big( \tilde{\sigma}_{t+1}, \tilde{\xi}_{t+1}, \tilde{\eta}_{t+1} \big)$ occurs
with positive probability for all $t \geq 0$. Step $(2)(c)$ uses common random 
numbers, such that $\big( \sigma^2_{t+1}, \xi_{t+1},
\eta_{t+1} \big) = \big( \tilde{\sigma}^2_{t+1}, \tilde{\xi}_{t+1},
\tilde{\eta}_{t+1} \big)$ implies $\beta_{t+1} = \tilde{\beta}_{t+1}$.
This allows the full chains to meet exactly with positive probability
at every step.

\begin{algorithm}[!htb] 
\DontPrintSemicolon
\KwIn{$ C_t := (\beta_t, \eta_t, \sigma^2_t, \xi_t) $ and $\tilde{C}_t := (\tilde{\beta}_t, \tilde{\eta}_t, \tilde{\sigma}^2_t, \tilde{\xi}_t)$.} 
%\KwResult{Sample $\big( C_{t+1}^{(1)}, C_{t+1}^{(2)} \big) | C_t^{(1)}, C_t^{(2)}  $ jointly such that
%each $C_{t+1}^{(i)} \big| C_t^{(i)}$ are marginally distributed from the Markov kernel in Algorithm \ref{algo:Half_t_Exact}.} 
\lIf{$C_t = \tilde{C}_t$}{
Sample $C_{t+1} | C_{t}$ as in \eqref{eq:blocked_gibbs_ergodicity_version} and set $\tilde{C}_{t+1} = C_{t+1}$.
}
\uElse{
 \begin{enumerate}
 \item Sample $(\eta_{t+1}, \tilde{\eta}_{t+1}) \big| C_t, \tilde{C}_t$ component-wise independently using Algorithm \ref{algo:slice_sampling_exact_meeting_coupling}. 
 \item Sample $( (\xi_{t+1}, \sigma_{t+1}^2, \beta_{t+1} ), (\tilde{\xi}_{t+1}, \tilde{\sigma}_{t+1}, \tilde{\beta}_{t+1} ))$ given $\eta_{t+1}, \tilde{\eta}_{t+1}$ as follows:
 \begin{enumerate}
 \item Sample $(\xi_{t+1}, \tilde{\xi}_{t+1}) \big| \eta_{t+1}, \tilde{\eta}_{t+1}, \xi_{t}, \tilde{\xi}_{t}$ using coupled Metropolis--Rosenbluth--Teller--Hastings (Algorithm \ref{algo:xi_MH_exact_meeting_coupling} in Appendix \ref{appendices:algos}). 
 \item Sample $( \sigma_{t+1}^2, \tilde{\sigma}_{t+1}^2 ) | \xi_{t+1}, \eta_{t+1}, \tilde{\xi}_{t+1}, \tilde{\eta}_{t+1}$ from a maximal coupling of two Inverse Gamma distributions (Algorithm \ref{algo:max_coupling} in Appendix \ref{appendices:algos}).
 \item Sample $( \beta_{t+1}, \tilde{\beta}_{t+1} ) | \sigma_{t+1}^{2}, \xi_{t+1}, \eta_{t+1}, \tilde{\sigma}_{t+1}^{2}, \tilde{\xi}_{t+1}, \tilde{\eta}_{t+1}$ using common random numbers for Gaussian scale mixture priors (Algorithm \ref{algo:fast_mvn_bhattacharya_crn_coupling} in Appendix \ref{appendices:algos}).
 \end{enumerate}
 \end{enumerate}
 }
 \Return $(C_{t+1} , \tilde{C}_{t+1})$.
  \caption{A one-scale coupled MCMC kernel for Half-$t(\nu)$ priors.}
 \label{algo:Half_t_meetings_every_step}
\end{algorithm}
\begin{algorithm}[!htb]
\SetAlgoLined
\DontPrintSemicolon
\KwResult{A coupling of slice samplers marginally targeting 
$p( \cdot | m_j)$ and $p( \cdot | \tilde{m}_j)$, where
$p(\eta_j | m) \propto (\eta_j^{\frac{1-\nu}{2}}(1+\nu \eta_j)^{\frac{\nu + 1}{2}})^{-1} e^{-m \eta_j}$ on $(0,\infty)$.
}
\KwIn{$\eta_{t,j}, \tilde{\eta}_{t,j}, m_{t,j} := \xi_{t} \beta_{t,j}^2 / ( 2 \sigma_{t}^2 ), \tilde{m}_{t,j} := \tilde{\xi}_{t} \tilde{\beta}_{t,j}^2 / (2 \tilde{\sigma}_{t}^2) >0 $.}
\begin{enumerate}
\item Sample 
$U^{crn}_{j,*} \sim \Uniform(0,1)$, and set $U_{j,*} := U^{crn}_{j,*} (1+\nu \eta_{t,j})^{-\frac{\nu + 1}{2}} $ and 
$\tilde{U}_{j,*} := U^{crn}_{j,*} (1+\nu \tilde{\eta}_{t,j})^{-\frac{\nu + 1}{2}} $. 
\item Sample 
$(\eta_{t+1,j}, \tilde{\eta}_{t+1,j})  | U_{j,*}, \tilde{U}_{j,*}$ from a maximal coupling of distributions
$P_j$ and $\tilde{P}_j$ using Algorithm \ref{algo:max_coupling} in Appendix \ref{appendices:algos}, 
%where $P_{j}$ and $\tilde{P}_j$ are as in the second step of  Algorithm \ref{algo:slice_sampling_1}. 
where $P_{j}$ and $\tilde{P}_{j}$ have unnormalized densities $\eta \mapsto \eta^{s-1} e^{ - m_j \eta}$ 
and $\eta \mapsto \eta^{s-1} e^{ - \tilde{m}_j \eta}$ on $(0, T_{j})$
and $(0, \tilde{T}_{j})$ respectively, where $T_{j}=(U_{j,*}^{-2/(1+\nu)}-1)/\nu$, 
$\tilde{T}_{j}=(\tilde{U}_{j,*}^{-2/(1+\nu)}-1)/\nu$ and $s=(1+\nu)/2$. 

%$(\eta^{(1)}_{t+1,j}, \eta^{(2)}_{t+1,j})  | U^{(1)}_{j,*}, U^{(2)}_{j,*} \sim \gamma_{max} \Big( P_j^{(1)}, P_j^{(2)} \Big)$ 
%using Algorithm \ref{algo:max_coupling} in Appendix \ref{appendices:algos},
%where each marginal distribution $P_{j}^{(i)}$ corresponds to distribution $P_j$ of the slice sampler of Algorithm \ref{algo:slice_sampling_1}. 
\end{enumerate}
\Return $(\eta_{t+1,j}, \tilde{\eta}_{t+1,j})$.
\caption{Coupled Slice Sampler for Half-$t(\nu)$ priors using maximal couplings.} 
\label{algo:slice_sampling_exact_meeting_coupling} 
\end{algorithm}

%\paragraph{Computation Cost.} 
%We compare the computational cost of each step of Algorithm \ref{algo:Half_t_meetings_every_step} to the single chain kernel of Algorithm \ref{algo:Half_t_Exact}. When $C_t^{(1)}=C_t^{(2)}$, Algorithm \ref{algo:Half_t_meetings_every_step} has the same $\mathcal{O}(n^2p)$ deterministic cost as Algorithm \ref{algo:Half_t_Exact} by construction. When $C_t^{(1)} \neq C_t^{(2)}$, Steps $(1)$ and $(2)(a-b)$ of Algorithm \ref{algo:Half_t_meetings_every_step} are based on maximal couplings. By Proposition \ref{prop:comp_cost_max_coupling}, this has twice the computation cost in expectation as the corresponding steps of Algorithm \ref{algo:Half_t_Exact}. Step $(2)(c)$ is based on a CRN coupling, which has deterministic cost that is also twice the cost of the corresponding step of Algorithm \ref{algo:Half_t_Exact}. Overall, Algorithm \ref{algo:Half_t_meetings_every_step} has twice the $\mathcal{O}(n^2p)$ cost of the kernel in Algorithm \ref{algo:Half_t_Exact} in expectation. 

%\paragraph{Analysis of meeting times.} 
Under the one-scale coupling, 
we prove that the meetings times have exponential tails and hence
finite expectation. Our analysis is based on identifying a suitable drift
function for a variant of the marginal chain in Section \ref{subsec:blocked_gibbs_mcmc}
and an application of \citet[Proposition 4]{jacob2020unbiasedJRSSB}.
We assume that the global shrinkage prior $\pi_\xi(\cdot)$ has a compact 
support on $(0, \infty)$. Such compactly supported priors 
%or an empirical Bayes estimator 
for $\xi$ have been recommended by \citet{vanderpas2017adaptiveEJS} 
to achieve 
%near 
optimal posterior concentration for the Horseshoe in the Normal means model. 
For simplicity, we also assume that each component $\eta_{t+1,j} | \beta_{t,j}, \xi_t, \sigma_t^2$
and $\xi_{t+1} | \eta_{t+1}$ in \eqref{eq:blocked_gibbs_ergodicity_version}
are sampled perfectly, instead of slice sampling and Metropolis--Rosenbluth--Teller--Hastings
steps respectively.
Perfect sampling algorithms for each component $\eta_{t+1,j} | \beta_{t,j}, \xi_t, \sigma_t^2$
and $\xi_{t+1} | \eta_{t+1}$ are provided in Algorithms \ref{algo:xi_perfect_sampling} 
and \ref{algo:eta_perfect_sampling} of Appendix 
\ref{appendices:algos} for completeness; see also \cite{abjjwiley} for perfect sampling of the components of $\eta_{t+1,j} | \beta_{t,j}, \xi_t, \sigma_t^2$ for the Horseshoe. Perfectly 
sampling $\xi_{t+1} | \eta_{t+1}$ and $\eta_{t+1,j} | \beta_{t,j}, \xi_t, \sigma_t^2$ 
%also gives
retains the  
$\mathcal{O}(n^2p)$ 
computational cost for the full blocked Gibbs sampler, 
%but 
though
in practice this
implementation is more expensive per iteration 
due to the computation of eigenvalue decompositions in lieu of 
Cholesky decompositions, and the computation of inverse of the
confluent hypergeometric function of the second kind. 

%For simplicity, we work with the marginal chain in 
%Equation \eqref{eq:blocked_gibbs_ergodicity_version}.
% and we assume
%couple each $\eta_j | \beta_j, \xi, \sigma^2$ component-wise (Step (1) of
%Algorithm \ref{algo:Half_t_meetings_every_step}) using maximal couplings. 

Proposition \ref{prop:drift} gives a geometric drift condition for this variant of the blocked Gibbs sampler.

\begin{Prop}(Geometric drift). \label{prop:drift} Consider the Markov chain $(\beta_t, \xi_t, \sigma^2_t)_{t\geq 0}$ generated from the blocked Gibbs sampler in \eqref{eq:blocked_gibbs_ergodicity_version}, where 
each component $\eta_{t+1,j} | \beta_{t,j}, \xi_t, \sigma_t^2$ and $\xi_{t+1} | \eta_{t+1}$
are sampled perfectly such that $(\beta_t, \xi_t, \sigma^2_t) \mapsto \eta_{t+1} \mapsto  (\beta_{t+1}, \xi_{t+1}, \sigma^2_{t+1})$ with $\eta_{t+1}$ intermediary for each $t\geq0$. Assume the global shrinkage prior $\pi_\xi(\cdot)$ has a compact support. 
For $m_j = \xi \beta_j^2/(2 \sigma^2)$, $c\in (0,1/2)$ and $d \in (0,1)$, define 
\begin{equation} \label{eq:lyapunov_function}
V(\beta, \xi, \sigma^2) = \sum_{j=1}^p m_j^{-c} + m_j^{d}.
\end{equation}
%where $m_j = \xi \beta_j^2/(2 \sigma^2)$, $c\in (0,1/2)$, and $d \in (0,1)$. 
Then for each $\nu \geq 1$, there exist some 
$c\in (0,1/2), d \in (0,1)$ such that $V$ is a drift function, i.e., 
there exists $\gamma \in (0,1)$ and $K \in (0, \infty)$ such that
for all $\beta_{t} \in \mathbb{R}^p$, $\xi_t >0$ and $\sigma_t^2 >0$,
\begin{flalign} \label{eq:drfit}
%  \forall \beta_{t} \in \mathbb{R}^p, \forall \xi_t \in \mathbb{R}_{>0}, 
%  \forall \sigma_t^2 \in \mathbb{R}_{>0},
%  \quad 
\mathbb{E}[V(\beta_{t+1}, \xi_{t+1}, \sigma_{t+1}^2) | \beta_{t}, \xi_{t}, \sigma_{t}^2] \leq \gamma V(\beta_{t}, \xi_{t}, \sigma_{t}^2) + K.
\end{flalign}
\end{Prop}

The drift (or ``Lyapunov'') function in \eqref{eq:lyapunov_function},
approaches infinity when any $\beta_j$ approaches the 
origin or infinity. This ensures that the corresponding sub-level 
sets, defined by 
$S(R) = \{ (\beta, \xi, \sigma^2) \in \mathbb{R}^p \times \mathbb{R}_{>0} 
\times \mathbb{R}_{>0} : V(\beta, \xi, \sigma) < R \}$ for $R>0$, exclude
the pole at the origin and are bounded. 
Such geometric drift condition, in conjunction with 
\citet[Proposition 4]{jacob2020unbiasedJRSSB}, 
helps verify that the meeting times under the one-scale coupling 
have exponential tails and hence finite expectation. 

\begin{Prop} \label{prop:one_scale_meeting}
Consider the blocked Gibbs sampler in \eqref{eq:blocked_gibbs_ergodicity_version}. 
In the setup of Proposition \ref{prop:drift}, assume that the global shrinkage 
prior $\pi_\xi(\cdot)$ has a compact support. Write $Z_t = (\beta_t, \xi_t, \sigma^2_t)$ and 
$\tilde{Z}_t = ( \tilde{\beta}_t, \tilde{\xi}_t, \tilde{\sigma}^2_t)$.
Consider the one-scale coupling given by 
$(Z_t, \tilde{Z}_t) \mapsto (\eta_{t+1}, \tilde{\eta}_{t+1}) \mapsto (Z_{t+1}, \tilde{Z}_{t+1})$
where $(\eta_{t+1}, \tilde{\eta}_{t+1})$ are maximally coupled component-wise, and 
$(Z_{t+1}, \tilde{Z}_{t+1})$ are coupled using common random numbers. Denote
 the meeting time by $\tau := \inf \{ t \geq 0: Z_t = \tilde{Z}_t \}$. Then
%\begin{flalign}
$\mathbb{P}( \tau > t ) \leq A_0 \kappa_0^t$
%\end{flalign}
for some constants $A_0 \in (0, \infty)$ and $\kappa_0 \in (0,1)$, and for all $t \geq 0$.
\end{Prop}

%\paragraph{Exponentially-tailed meeting times imply 
%geometric ergodicity.}
Exponentially-tailed meeting times immediately imply that the blocked Gibbs sampler in Section 
\ref{subsec:blocked_gibbs_mcmc} is geometrically ergodic. This follows 
from noting that \citet[Proposition 4]{jacob2020unbiasedJRSSB} is closely linked to 
a minorization condition \citep{rosenthal1995minorizationJASA}, and details
are included in Proposition \ref{prop:minorization} of Appendix \ref{appendices:proofs}. 
For a class of Gibbs samplers targeting shrinkage 
priors including the Bayesian LASSO, the Normal-Gamma prior \citep{griffin2010inferenceBA}, and 
Dirichlet-Laplace prior \citep{bhattacharya2015dirichletJASA}, \citet{khare2013geometricEJS} and
\citet{pal2014geometricEJS} have proven geometric ergodicity based on drift and
minorization arguments \citep{meyn1993markov,roberts2004general}. For the Horseshoe prior,
\citet{johndrow2020scalableJMLR} has recently established geometric ergodicity. In
the high-dimensional setting with $p > n$ for the Horseshoe, the 
proof of \citet{johndrow2020scalableJMLR} required 
truncating the prior on each $\eta_j$ below by a small $\ell > 0$ to guarantee the uniform bound 
$\eta_j \geq \ell$. Our proof of geometric ergodicity for Half-$t(\nu)$ priors including 
the Horseshoe ($\nu=1$) works in both low and high-dimensional settings, without any 
modification of the priors on the $\eta_j$. In parallel work, 
\citet{bhattacharya2021geo} have also established geometric ergodicity 
for the Horseshoe prior without requiring such truncation.

\begin{Remark}
Our proofs of exponentially-tailed meeting times and geometric ergodicity
generalize to a larger class of priors satisfying some moment conditions on the full
conditionals of $\eta_j$. 
Consider a compactly supported prior $\pi_\xi$ on $\xi$, and a prior
$\pi_{\eta}$ on each $\eta_j$. Then 
the unnormalized density of the full conditionals of each $\eta_j$ is given by
$\pi \big(\eta_{j} \mid  m_j = \xi \beta_j^2/(2 \sigma^2) \big) \propto \eta_j^{1/2} e^{- m_j \eta_j} \pi_{\eta}(\eta_j)$.
%\begin{flalign}
%\pi \big(\eta_{j} \mid  m_j = \xi \beta_j^2/(2 \sigma^2) \big) \propto \eta_j^{1/2} \exp \big(- m_j \eta_j \big) \pi_{\eta}(\eta_j).
%\end{flalign}
Consider the following assumptions on $\pi_\eta$.
\begin{enumerate}
 \item For some $0<c < 1/2$, there exists $0 < \epsilon < \Gamma(\frac{1}{2}-c)^{-1}\sqrt{\pi}$ and $K^{(1)}_{c, \epsilon} < \infty$ such that
$\mathbb{E}[ \eta_{j}^c | m_j] \leq \epsilon m_j^{-c} + K^{(1)}_{c, \epsilon}$ for all $m_j > 0$.
%\begin{flalign}
%\mathbb{E}[ \eta_{j}^c | m_j] \leq \epsilon m_j^{-c} + K^{(1)}_{c, \epsilon} \quad \text{ for all } m_j > 0.
%\end{flalign}
\item For some $0 < d < 1$, there exists $0 < \epsilon < 2^d $ and $K^{(2)}_{d, \epsilon} < \infty$ such that
$\mathbb{E}[ \eta_{j}^{-d} | m_j] \leq \epsilon m_j^{d} + K^{(2)}_{d, \epsilon}$ for all $m_j > 0$.
%\begin{flalign}
%\mathbb{E}[ \eta_{j}^{-d} | m_j] \leq \epsilon m_j^{d} + K^{(2)}_{d, \epsilon} \quad \text{ for all } m_j > 0.
%\end{flalign}
\end{enumerate}
Then the blocked Gibbs sampler in \eqref{eq:blocked_gibbs_ergodicity_version} is geometrically ergodic. 
\end{Remark}

We now consider the performance of Algorithm
\ref{algo:Half_t_meetings_every_step} on synthetic datasets. In particular, we empirically illustrate that the rate $\kappa_0$ in Proposition \ref{prop:one_scale_meeting} can tend to $1$ exponentially fast as dimension $p$ increases.  
The design matrix $X \in \mathbb{R}^{n \times p}$ 
is generated with $[X]_{i,j}
\overset{\text{i.i.d.}}{\sim} \mathcal{N}(0,1)$, and a response vector as $y \sim
\mathcal{N}(X \beta_{*} , \sigma^2_{*}I_n)$, where $\beta_{*} \in \mathbb{R}^p$
is the true signal and $\sigma_{*} \geq 0$ is the true error standard
deviation. We choose $\beta_{*}$ to be sparse such that given sparsity parameter
$s$, $\beta_{*,j}=2^{(9-j)/4}$ for $1 \leq j \leq s$ and $\beta_{*,j}=0$ for
all $s >  j$. Figure \ref{fig:max_coupling_plot} shows the meeting times $\tau$
of coupled Markov chains (with a lag $L=1$) targeting the Horseshoe posterior
($\nu=1$) with $a_0=b_0=1$ under Algorithm
\ref{algo:Half_t_meetings_every_step}. We consider $n=100, s=10, \sigma_*=0.5$ and vary the dimension $p$. 
For each dimension $p$, we simulate $100$ meeting times 
under the one-scale coupling for independently generated synthetic datasets.
We initialize chains independently from the prior distribution.  
For the $\xi$ updates, we use a proposal step-size of $\sigma_{\text{MRTH}}=0.8$ as in
\citet{johndrow2020scalableJMLR}. 
Figure \ref{fig:max_coupling_plot}, with a vertical axis on a logarithmic scale, shows
that the meeting times increase exponentially and have heavier tails with increasing dimension. 

%Additional figures showing the effect of number of observations $n$, sparsity
%parameter $s$, and true error standard deviation $\sigma_*$ is included in
%Appendix \ref{appendices:figures}. 

%Each sub-figure plots kernel density estimates of $\tau$ against each parameter $n,p,s, \sigma_*$ respectively, based on 100 independent replicate simulations. 
%Figure \ref{fig:vary_n} shows that the meeting times tend to decrease with more observations. 
%Figures \ref{fig:vary_p}, \ref{fig:vary_s} and \ref{fig:vary_sigma} show that the meeting times tend to increase
%exponentially and have heavier tails with increasing dimensions, sparsity parameter and the true error standard deviation respectively. 

%\textcolor{blue}{(JEJ) in general we should not say ``independent iterations'' since this will cause confusion between replicates of the simulation and length of the Markov chain. I have suggested ``replicate simulations'' but feel free to choose a different terminology. Also, I think people may comment on the fact that some of these coupling time distributions are bimodal and ask us to explain why.}

\begin{figure}[!ht]
\captionsetup[subfigure]{font=normalsize,labelfont=normalsize}
    \centering
    \includegraphics[width=0.32\textwidth]{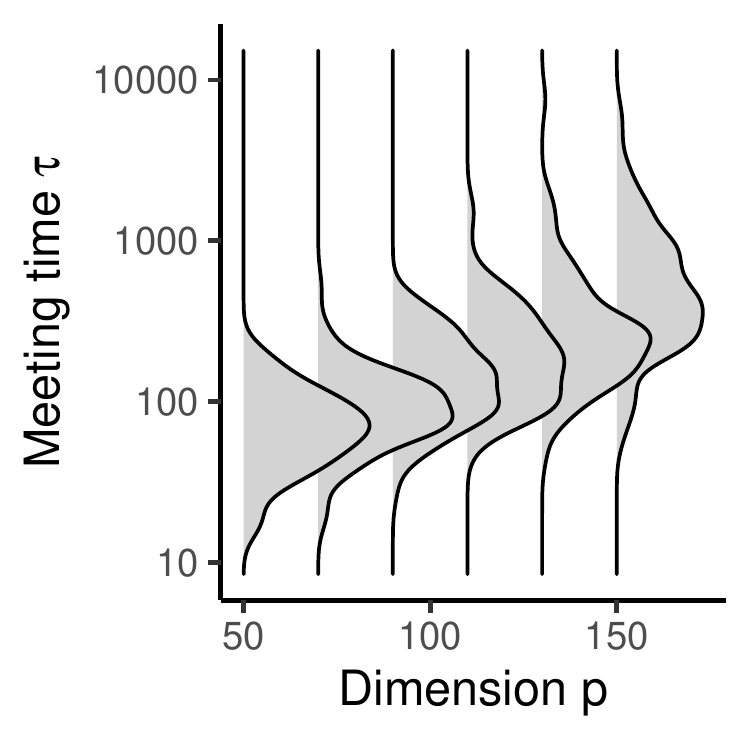}
  \caption{Meeting times against dimension for posteriors associated with a Horseshoe prior under the one-scale coupling of Algorithm \ref{algo:Half_t_meetings_every_step},
  with $n=100, s=10, \sigma_*=0.5$.}
    \label{fig:max_coupling_plot}
\end{figure}

\section{Coupling strategies for high-dimensional settings} \label{section:coupled_MCMC_2}

To address the scaling issues observed with the one-scale coupling in 
Section \ref{subsec:max_coupling}, 
%in Sections \ref{subsec:two_scale_coupling} and \ref{subsec:coupling_variants}
we develop a \textit{two-scale} coupling in Section
\ref{subsec:two_scale_coupling} that delivers better empirical performance.
We also offer some formal results supporting the empirical improvements.
In Section \ref{subsec:coupling_variants} we develop other related coupling strategies 
which have exponentially-tailed meeting times and have similar empirical
performance as the two-scale coupling.
In Section \ref{subsec:degree_of_freedom}, we investigate the
impact of the degree of freedom $\nu$ for Half-$t(\nu)$ priors on our two-scale
coupling. We find that Half-$t(\nu)$ priors with higher degrees of freedom
$\nu > 1$ can give similar statistical performance to the Horseshoe ($\nu=1$),
whilst enjoying orders of magnitude computational improvements through shorter
meeting times. 

\subsection{A two-scale coupling} \label{subsec:two_scale_coupling}
%In Section \ref{subsec:max_coupling}, we highlighted that the coupling in
%Algorithm \ref{algo:Half_t_meetings_every_step}, which has positive probability
%of meeting on every step, has meeting times which sometimes grow exponentially
%with dimension, sparsity parameter and true error standard deviation. 
We develop a \textit{two-scale} coupling algorithm,
which slightly differs from the terminology used in stochastic dynamics  
%\citep{mattingly2003on, hairer2008spectralAOP}
and in analysis of pre-conditioned HMC 
\citep{bou-rabee2020twoscale2020SPDEAC}. 
It is also closely linked to the delayed one-shot coupling 
of \citet{roberts2002oneshotSTOCHPA}, and
refers to coupled kernels that attempt exact meetings when the chains are close, and aim
for a contraction when the chains are far.
%What is meant by ``close'' will be defined carefully below. 
The motivation for this construction is that in Algorithm
\ref{algo:Half_t_meetings_every_step}, the components of
$(\eta_t, \tilde{\eta}_t)$ that fail to meet instead evolve independently. As
dimension grows, the number of components of $(\eta_t, \tilde{\eta}_t)$ that
fail to meet also tends to grow, making it increasingly difficult for all components 
to exactly meet on subsequent iterations. 
%Our hypothesis is that this
%phenomenon is responsible for the exponentially increasing meeting times as a
%function of dimension. 
%As we pay a high price -- independent evolution of
%some of the components of $(\eta^{(1)}_t,\eta^{(2)}_t)$ -- for failing to couple
%at each step, it makes sense to only attempt exact meetings when the associated 
%meeting probability is high enough. 
In the two-scale coupling, we only attempt to obtain exact meetings when the associated 
probability is large enough. 
This is done by constructing a metric $d$ and a corresponding threshold 
parameter $d_0\geq 0$ such that when the current states 
are $d$-close with $d(C_t, \tilde{C}_t) \leq d_0$, we apply Algorithm
\ref{algo:slice_sampling_exact_meeting_coupling} and try to obtain exact meetings.
When $d(C_t, \tilde{C}_t) >d_0$, we instead
employ 
common random numbers to sample $(\eta_t, \tilde{\eta}_t)$.
Our chosen metric $d$ on $\mathbb{R}^p \times \mathbb{R}^p_{> 0} \times \mathbb{R}_{> 0} \times \mathbb{R}_{> 0}$ is 
%defined by 
\begin{equation} \label{eq:distance_metric}
d(C_t, \tilde{C}_t) = \mathbb{P}\big( \eta_{t+1} \neq \tilde{\eta}_{t+1} | C_t, \tilde{C}_t \big),
\end{equation}
where recall that $ C_t = (\beta_t, \eta_t, \sigma^2_t, \xi_t) $, $\tilde{C}_t = (\tilde{\beta}_t, \tilde{\eta}_t, \tilde{\sigma}^2_t, \tilde{\xi}_t)$, and $(\eta_{t+1}, \tilde{\eta}_{t+1}) | C_t, \tilde{C}_t$ is sampled using
the coupled slice sampler in Algorithm \ref{algo:slice_sampling_exact_meeting_coupling}.  As the coupling in Algorithm 
\ref{algo:slice_sampling_exact_meeting_coupling} is independent component-wise, we can obtain a simplified expression involving only univariate probabilities as 
\begin{flalign}
d(C_t, \tilde{C}_t) &=  1 - \prod_{j=1}^p \mathbb{P} \big( \eta_{t+1,j} = \tilde{\eta}_{t+1,j} |  \eta_{t,j}, \tilde{\eta}_{t,j}, m_{t,j}, \tilde{m}_{t,j} \big) \label{eq:distance_metric_estimate2},
\end{flalign}
where $m_{t,j} = {\xi_{t} \beta_{t,j}^2}/(2 \sigma_{t}^2)$ and $m_{t,j} = {\tilde{\xi}_{t} \tilde{\beta}_{t,j}^2}/(2 \tilde{\sigma}_{t}^2)$. Under this metric $d$ and a threshold $d_0 \in [0, 1]$, an exact meeting is attempted when $\mathbb{P}(\eta_{t+1} = \tilde{\eta}_{t+1} | C_t, \tilde{C}_t) \geq 1-d_0$. Once $\eta_{t+1}$ and $\tilde{\eta}_{t+1}$ coincide, and if the scalars $\xi_{t+1}$ and $\tilde{\xi}_{t+1}$
%and $\sigma_{t+1}^2$ and $\tilde{\sigma}_{t+1}^2$ 
also subsequently coincide, then
the entire chains meet. 
%When $d_0= 1$, we always apply the
%maximal coupling slice sampler (Algorithm \ref{algo:slice_sampling_exact_meeting_coupling}). 
%This corresponds to the one-scale coupling of Algorithm \ref{algo:Half_t_meetings_every_step}. 
When $d_0 = 1$, the two-scale coupling reduces to the one-scale coupling of Algorithm \ref{algo:Half_t_meetings_every_step} since the maximal coupling slice sampler (Algorithm \ref{algo:slice_sampling_exact_meeting_coupling}) is invoked at each step. On the other hand, when threshold $d_0= 0$, we always apply the common random numbers. Then the chains
$C_t, \tilde{C}_t$ may come arbitrarily close
but will never exactly meet. Empirically we find that different thresholds $d_0$
values away from $0$ and $1$ give similar meeting times. Simulations 
concerning these choices can be found in Appendix
\ref{appendices:threshold_choice}.

%% The motivation behind this
%% two-scale construction is to only attempt exact meetings when the probability
%% of doing so is high.
%% P: I've removed the above as it was already said a few paragraphs above
%Under this metric $d$, chains are close when
%$\mathbb{P}(\eta_{t+1} = \tilde{\eta}_{t+1} | C_t, \tilde{C}_t) \geq 1-d_0$. 
%In that event, the two-scale coupling strategy 
%involves Algorithm
%\ref{algo:slice_sampling_exact_meeting_coupling} 
%such that the vectors $\eta_{t+1}$ and $\tilde{\eta}_{t+1}$ can
%exactly meet with probability at least $1-d_0$.
%Once $\eta_{t+1}$ and $\tilde{\eta}_{t+1}$
%coincide, and if the scalars $\xi_{t+1}$ and $\tilde{\xi}_{t+1}$, and 
%$\sigma_{t+1}^2$ and $\tilde{\sigma}_{t+1}^2$ also coincide, then
%the entire chains meet. 
We now discuss computation of the metric $d(C_t, \tilde{C}_t)$. Since the probability in \eqref{eq:distance_metric}
is unavailable in closed form, we can resort to various approximations. One option is to evaluate the one-dimensional integrals in \eqref{eq:distance_metric_estimate2} using numerical integration. Alternatively, one can form the Monte Carlo based estimate 
\begin{equation} \label{eq:distance_metric_estimate1}
  \widehat{d_R}^{(1)}(C_t, \tilde{C}_t) = \frac{1}{R} \sum_{r=1}^R \mathbbm{1} \{  \eta^{(r)}_{t+1} \neq \tilde{\eta}^{(r)}_{t+1} \},
\end{equation}
where $(\eta^{(r)}_{t+1}, \tilde{\eta}^{(r)}_{t+1} )$ are 
sampled independently for $r=1,\ldots,R$ using Algorithm
\ref{algo:slice_sampling_exact_meeting_coupling}. We recommend a Rao--Blackwellized estimate by combining Monte Carlo with analytical calculations
of integrals as
\begin{equation} \label{eq:distance_metric_estimate3}
\widehat{d_R}^{(2)}(C_t, \tilde{C}_t) = 1 - \prod_{j=1}^p \Big(\frac{1}{R} \sum_{r=1}^R \mathbb{P}\Big( \eta_{t+1,j} = \tilde{\eta}_{t+1,j} \Big| U^{(r)}_{j}, \tilde{U}^{(r)}_{j}, m_{t,j}, \tilde{m}_{t,j} \Big) \Big),
\end{equation}
where $\big(U^{(r)}_{j}, \tilde{U}^{(r)}_{j} \big)$ are sampled independently for
$r=1,\ldots,R$ using Algorithm \ref{algo:slice_sampling_exact_meeting_coupling}
and each 
$\mathbb{P}( \eta_{t+1,j} = \tilde{\eta}_{t+1,j} \Big| U^{(r)}_{j}, \tilde{U}^{(r)}_{j}, m_{t,j}, \tilde{m}_{t,j})$ 
is calculated analytically based on meeting probabilities from Algorithm \ref{algo:slice_sampling_exact_meeting_coupling}. 
%\begin{enumerate}
%\item We can directly form 
%\item 
%We can calculate $d(C_t, \tilde{C}_t)$ by evaluating the one-dimensional integrals in \eqref{eq:distance_metric_estimate2} using numerical integration or analytically. 
%\item 
%\end{enumerate}
Further metric calculation details are in Appendix \ref{appendices:metric_calculation}. 
For a number of samples $R$, estimates \eqref{eq:distance_metric_estimate1} 
and \eqref{eq:distance_metric_estimate3} both have computation cost
of order $pR$. Compared to the estimate in \eqref{eq:distance_metric_estimate1},
the estimate in \eqref{eq:distance_metric_estimate3}
has lower variance and faster numerical run-times as it
only involves sampling uniformly distributed random numbers. It suffices to choose a small
number of samples $R$, and often we take $R=1$. Indeed it appears unnecessary
to estimate $d(C_t, \tilde{C}_t)$ accurately, as we are only interested in
comparing that distance to a fixed
threshold $d_0\in [0,1]$. Often 
the trajectories $(d(C_t, \tilde{C}_t))_{t \geq 0}$ 
initially take values close to $1$, and then
sharply drop to values close to $0$. This leads to the estimates
in \eqref{eq:distance_metric_estimate1} and \eqref{eq:distance_metric_estimate3}
having low variance. Henceforth, we will use the estimate
in \eqref{eq:distance_metric_estimate3}
in our experiments with two-scale couplings, unless specified otherwise. 

%We consider the computation cost of each step of
%the two-scale coupling. 
Overall the two-scale coupling kernel
has twice the $\mathcal{O}(n^2p)$ cost of the single chain
kernel in \eqref{eq:blocked_gibbs_ergodicity_version}. 
When $C_t \neq \tilde{C}_t$, we calculate a
distance estimate and sample from a coupled slice sampling kernel. Calculating
the distance estimate involves sampling $2pR$ uniforms and has $\mathcal{O}(p)$ cost.
Coupled slice sampling with maximal couplings (Algorithm
\ref{algo:slice_sampling_exact_meeting_coupling}) or common random numbers
both have expected or deterministic cost $\mathcal{O}(p)$. The remaining steps of 
the two-scale coupling match the one-step coupling in Algorithm 
\ref{algo:Half_t_meetings_every_step}. 

We now consider the performance of the two-scale coupling
on synthetic datasets. The setup is identical to that introduced in Section
\ref{subsec:max_coupling}, where for each dimension $p$ 
we simulate $100$ meetings times based on independently generated synthetic datasets.
We target the Horseshoe posterior ($\nu=1$) with
$a_0=b_0=1$. Figure \ref{fig:two_scale_coupling_plot} shows the meeting times
$\tau$ of coupled Markov chains (with a lag $L=1$), for both the one-scale coupling 
and the two-scale coupling
with $R=1$ and threshold $d_0=0.5$. 
Figure \ref{fig:two_scale_coupling_plot} shows that the two-scale coupling 
can lead to orders of magnitude reductions in meeting times, compared to the one-scale 
coupling. 

\begin{figure}[!ht]
\captionsetup[subfigure]{font=normalsize,labelfont=normalsize}
    \centering
    \includegraphics[width=0.43\textwidth]{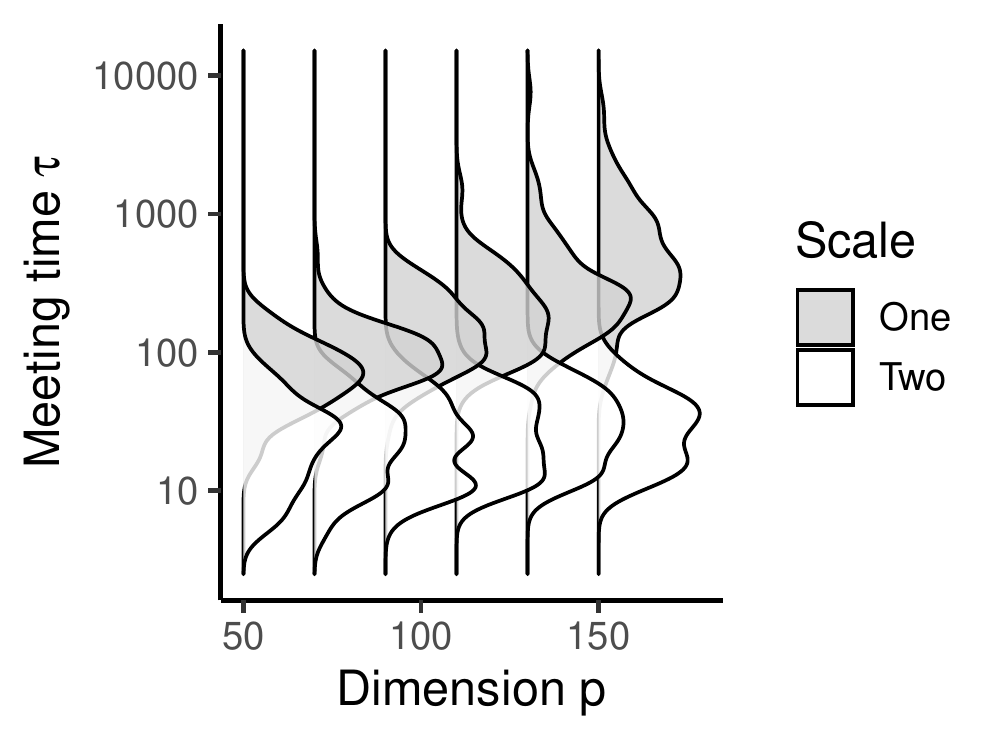}
  \caption{Meeting times against dimension for posteriors associated with a Horseshoe prior under 
    one-scale and two-scale couplings,
    with $n=100,s=10, \sigma_*=0.5$. }
    \label{fig:two_scale_coupling_plot}
\end{figure}
We next present some preliminary results which hint as to why 
the two-scale coupling leads to the drastic reduction in meeting times. 
We focus on the coupling of the $(\eta, \tilde{\eta})$ full conditionals, 
as the one-scale and two-scale algorithms differ exactly at this step. 
First, we show that for any monotone function $h:(0, \infty) \rightarrow \mathbb{R}$
and all components $j=1, \ldots, p$, the expected distance
$\mathbb{E} [ | h(\eta_{t+1,j}) - h(\tilde{\eta}_{t+1,j}) | ] $
is the same under both common random numbers and maximal coupling.

\begin{Prop} \label{prop:crn_max_expectation}
Consider the setup of Proposition \ref{prop:one_scale_meeting}, 
such that the $\eta$ full conditionals are sampled perfectly. Then 
for any monotone function $h:(0, \infty) \rightarrow \mathbb{R}$ and
all components $j=1, \ldots, p$, 
\begin{equation}
\mathbb{E}_{ \text{max} } 
\big[ | h(\eta_{t+1,j}) - h(\tilde{\eta}_{t+1,j}) | \big| m_{t}, \tilde{m}_{t} \big] =
\mathbb{E}_{ \text{crn} } 
\big[ | h(\eta_{t+1,j}) - h(\tilde{\eta}_{t+1,j}) | \big| m_{t}, \tilde{m}_{t} \big],
\end{equation}
where $\mathbb{E}_\text{max}$ and $\mathbb{E}_\text{crn}$ correspond to 
expectations under the maximal coupling and CRN coupling respectively of 
$(\eta_{t+1}, \tilde{\eta}_{t+1})$ given $m_t$ and $\tilde{m}_t$.
\end{Prop}

%Proposition \ref{prop:crn_max_expectation} suggests that such single-step 
%expectations may not be sufficient to distinguish the behavior of CRN from maximal couplings. 
Proposition \ref{prop:crn_max_expectation} implies that such single-step 
expectations alone do not distinguish the behavior of CRN from maximal couplings. This compels us to investigate other distributional features of $| h(\eta_{t+1,j}) - h(\tilde{\eta}_{t+1,j}) |$ under either coupling to possibly tease apart differences in their behavior on high probability events. 
%Focusing on the Horseshoe prior and the function $h(x) = \log(1+x)$, we now consider 
%high probability events in a stylized setup to tease apart this difference between 
%CRN and maximal couplings.
Focusing on the Horseshoe prior and making the choice $h(x) = \log(1+x)$, we uncover a distinction between the tail behavior of the two couplings which can substantially accumulate over $p$ independent coordinates. We offer an illustration in a stylized setting below. 
\begin{Prop} \label{prop:crn_max_high_prob} For the Horseshoe prior $(\nu=1)$,
consider when $m_t = \mu \mathbf{1}_p \in (0,\infty)^p$ and 
$\tilde{m}_t= \tilde{\mu} \mathbf{1}_p \in (0,\infty)^p$ for some positive scalars 
$\mu$ and $\tilde{\mu}$ with $\tilde{\mu} \neq \mu$. Then under a CRN coupling,
\begin{equation} \label{eq:crn}
\mathbb{P} \big( \max_{j=1}^p | \log(1+\eta_{t+1,j}) - \log(1+\tilde{\eta}_{t+1,j}) | > | \log \mu - \log \tilde{\mu} | 
\ \big| \ m_t, \tilde{m}_t \big) = 0
\end{equation}
for all $p \geq 1$. Define a positive constant  $L=\log(1+1/\max \{\mu , \tilde{\mu} \})/2$. Under maximal coupling, for 
any $\alpha \in (0,1)$ and any $\alpha' \in (\alpha,1)$, there
exists a constant $D>0$ which does not depend on $p$ such that for all $p>1$,
\begin{flalign}  \label{eq:max}
\mathbb{P} \big(\max_{j=1}^p | \log(1+\eta_{t+1,j}) - \log(1+\tilde{\eta}_{t+1,j}) | 
> \log ( \alpha L ) + \log \log p  \ \big| \ m_t, \tilde{m}_t \big) > 
1 - e^{-D p^{1-\alpha'}}.
%1 - \exp( -D p^{1-\alpha'}).
\end{flalign}
\end{Prop}

Extensions of Proposition \ref{prop:crn_max_high_prob} which allow
different components of $m_t$ and $\tilde{m}_t$ to take different values
under some limiting assumptions are omitted here for simplicity. 
This stylized setting 
%of Proposition \ref{prop:crn_max_high_prob} 
already captures why the one-scale algorithm may result in much larger
meeting times in high dimensions compared to the CRN-based 
two-scale algorithm.
%, as then the $\mathcal{O}(\log \log p)$ deviations under the one-scale coupling
%can lead to lower coupling probability over multiple steps of the coupling kernel. 
We note that Proposition \ref{prop:crn_max_high_prob} remains informative 
whenever $| \log \mu - \log \tilde{\mu} | \ll \log \log p$, 
even when $\log \log p$ is not large. For example, consider
$\mu = \tilde{\mu}e^{\delta}$ for some small $\delta > 0$ with $\delta \ll \log \log p$. 
Under CRN, $\log (1+\eta_{t+1,j})$ and $\log (1+\tilde{\eta}_{t+1,j})$
remain $\delta$-close almost surely for all components $j$.
In contrast, under the maximal coupling, at least one component will have 
larger than $\log(\alpha L) + \log \log p$ deviations with
high probability. Since $\alpha$ and $L = \log(1+1/\mu)/2$ do not depend on $p$ or $\delta$, $\log(\alpha L) + \log \log p \gg \delta$ for fixed $\mu, \tilde{\mu} >0$. 
These larger deviations between some components of $\eta_{t+1}$
and $\tilde{\eta}_{t+1}$ under maximal couplings push the two chains
much further apart when sampling $(m_{t+1}, \tilde{m}_{t+1}) \big| \eta_{t+1}, \tilde{\eta}_{t+1}$, resulting in meeting probabilities that are much lower over multiple steps of the algorithm. 
%such that the meeting probability will be lower over multiple steps
%of the algorithm. 
Overall, Propositions \ref{prop:crn_max_expectation} and 
\ref{prop:crn_max_high_prob} show that multi-step kernel calculations 
taking into account such high-probability events may be necessary to further
distinguish the relative performances of one-scale and two-scale couplings.
A full analytic understanding of this difference in performance
remains an open problem.

%For example, 
%when $|\log m - \log \tilde{m}| = \delta$ for some small $\delta > 0$ such that $\delta \ll \log \log p$, 
%under CRN $\log (1+\eta_{t+1,j})$ and $\log (1+\tilde{\eta}_{t+1,j})$
%remain $\delta$ close almost surely for all components $j$. 
%In contrast under the maximal coupling at least one component will have 
%$\mathcal{O}(\log \log p)$ deviations with high probability. 
%Even when $\log \log p$ is not large, such deviation is sufficient to push the 
%two chains apart in the $(m_{t+1}, \tilde{m}_{t+1})$ given $\eta_{t+1}$ and $\tilde{\eta}_{t+1}$
%sampling step such that the coupling probability will be lower over multiple steps
%of the one-scale algorithm. 
%Overall, Propositions \ref{prop:crn_max_expectation} and 
%\ref{prop:crn_max_high_prob} highlight that multi-step kernel calculations 
%taking into account such high-probability events may be necessary to further
%tease apart the relative performances of the one-scale and the two-scale coupling 
%algorithms. 

%and an exciting avenue for follow-on work.  
Another way of understanding the benefits of two-scale coupling would 
be to find a distance under which the coupled chains under common random numbers exhibit a contraction. 
The CRN coupling part of Proposition \ref{prop:crn_max_high_prob} hints that 
$\bar d(m_t,\tilde m_t):=\sum_{j=1}^p | \log m_{t,j} - \log \tilde{m}_{t,j} |$ may be a natural  
metric. Appendix \ref{appendices:metric_comparison} contains
discussions and related simulations comparing this metric with 
alternatives. Proposition \ref{prop:tv_metric_upper_bound} 
verifies that such metric bounds the total variation distance between one-step transition kernels. 

\begin{Prop} \label{prop:tv_metric_upper_bound}
Let $\mathcal{P}$ denote the Markov transition kernel associated with the update from $(\beta_t,\xi_t, \sigma_t^2)$ to $(\beta_{t+1},\xi_{t+1}, \sigma_{t+1}^2)$ for the Gibbs sampler in \eqref{eq:blocked_gibbs_ergodicity_version}. 
Let $Z = (\beta, \xi, \sigma^2)$, $\tilde{Z} = (\tilde{\beta}, \tilde{\xi}, \tilde{\sigma}^2)$,
$m_{j} = \xi \beta_{j}^2/(2 \sigma^2)$, and $\tilde{m}_{j} = \tilde{\xi} \tilde{\beta}_{j}^2/(2 \tilde{\sigma}^2)$
for $j=1, \ldots, p$. Then,
\begin{equation} \label{eq:tv_metric_upper_bound}
\text{TV} \big( \mathcal{P}( Z, \cdot), \mathcal{P}( \tilde{Z}, \cdot) \big)^2 
\leq 
% \frac{(1+\nu)(e-1)}{4}  \sum_{j=1}^p | \log(m_j) - \log(\tilde{m}_j) | = 
\frac{(1+\nu)(e-1)}{4} \bar{d}(Z, \tilde{Z}).
\end{equation}
%for the metric $\bar{d}(Z, \tilde{Z}) := \sum_{j=1}^p | \log(m_j) - \log(\tilde{m}_j) |$.
%Consider the blocked Gibbs sampler in Equation \eqref{eq:blocked_gibbs_ergodicity_version}. This corresponds to the update from $(\beta_t,\xi_t, \sigma_t^2)$ to $(\beta_{t+1},\xi_{t+1}, \sigma_{t+1}^2)$, with the $\eta_{t+1}$ sampling step intermediary. 
%Denote $Z = (\beta, \xi, \sigma^2)$ and $\tilde{Z} = (\tilde{\beta}, \tilde{\xi}, \tilde{\sigma}^2)$, and denote
%$m_{j} = \xi \beta_{j}^2/(2 \sigma^2)$ and $\tilde{m}_{j} = \tilde{\xi} \tilde{\beta}_{j}^2/(2 \tilde{\sigma}^2)$
%for $j=1, \ldots, p$. 
%Then, 
%\begin{equation} \label{eq:tv_metric_upper_bound}
%\text{TV}( Z_t, \tilde{Z}_t )
%\leq \frac{(1+\nu)(e-1)}{4}  \sum_{j=1}^p | \log(m_j) - \log(\tilde{m}_j) |
%= \frac{(1+\nu)(e-1)}{4}  \bar{d}(Z, \tilde{Z})
%\end{equation}
%for the metric $\bar{d}(Z, \tilde{Z}) = \sum_{j=1}^p | \log(m_j) - \log(\tilde{m}_j) |$.
\end{Prop}

By Proposition \ref{prop:tv_metric_upper_bound}, it suffices to show that
$\big(\bar{d}(Z_t, \tilde{Z}_t)\big)_{t \geq 0}$ contracts under CRN  to a
sufficiently small value such that the upper bound in
\eqref{eq:tv_metric_upper_bound} is strictly less that $1$. This would ensure
that CRN coupling will bring the marginal chains close enough that the probability of one-step max coupling resulting in a meeting is high. Note that \eqref{eq:tv_metric_upper_bound} corresponds to the total
variation distance of the full chain, so under this setup we would attempt to
maximally couple the full vectors $\eta_{t+1}$ and $\tilde{\eta}_{t+1}$ (rather
than maximally coupling component-wise) when close. Analytical calculations to
establish such a contraction, and to understand how the contraction rate scales with
dimension and other features of the problem such as the choice of
prior, sparsity, and signal to noise ratio 
requires further work. 

\subsection{Alternative coupling strategies} \label{subsec:coupling_variants}

The rapid meeting of the two-scale coupling in simulations recommends its use. However, we were unable to establish finite expected meeting times with this coupling, which prevents its use in constructing unbiased estimates. We now describe several small modifications of the two-scale coupling that have similar empirical performance but for which we can prove that the meeting time distribution has exponential tails. 

% and the conditional independence of
%$(\eta_{t}, \tilde{\eta}_{t})$ motivate other coupling 
%strategies. 

Our basic approach is to 
%for the update of $(\eta_{t}, \tilde{\eta}_{t})$, we can 
maximally couple each component $(\eta_{t,j}, \tilde{\eta}_{t,j})$ until a failed meeting 
$\{ \eta_{t,j} \neq \tilde{\eta}_{t,j} \}$ occurs, after which we 
switch to CRN for the remaining components. 
The ordering of the components $j=1,\ldots,p$ can 
be deterministic or randomized at each iteration. Under this construction, 
only one component of $(\eta_{t}, \tilde{\eta}_{t})$, corresponding to 
the failed meeting, will evolve independently; the other components will either meet 
or evolve under CRN. Therefore, we expect to obtain benefits similar to the two-scale 
coupling in high dimensions. An apparent advantage is that we can bypass the choice
of metric and associated threshold. While the switch-to-CRN coupling has the same $\mathcal{O}(n^2p)$ 
cost per iteration as the two-scale coupling, it can give faster run-times in practice when the 
dimension $p$ is large compared to the number of observations $n$ as it avoids repeated calculation of a metric. 
%
%\paragraph{Computation Cost.} The switch-to-CRN coupling has the same $\mathcal{O}(n^2p)$ 
%cost per iteration as the two-scale coupling. However, the switch-to-CRN coupling does not 
%require calculating a metric, which can give faster run-times in practice when the 
%dimension $p$ is large compared to the number of observations $n$. 

%\paragraph{Simulations.} 
We report the performance of the switch-to-CRN coupling
with the ordering of the components
randomized at each iteration. 
%Under the same simulation setup as in Section \ref{subsec:two_scale_coupling},
The setup is identical to that in Section
\ref{subsec:two_scale_coupling}, where for each dimension $p$ 
we simulate $100$ meetings times based 
on independently generated synthetic datasets.
Figure \ref{fig:switch_to_crn_coupling_plot} shows that the switch-to-CRN coupling 
leads to similar meeting times as the two-scale coupling. 

\begin{figure}[!ht]
\captionsetup[subfigure]{font=normalsize,labelfont=normalsize}
    \centering
    \includegraphics[width=0.53\textwidth]{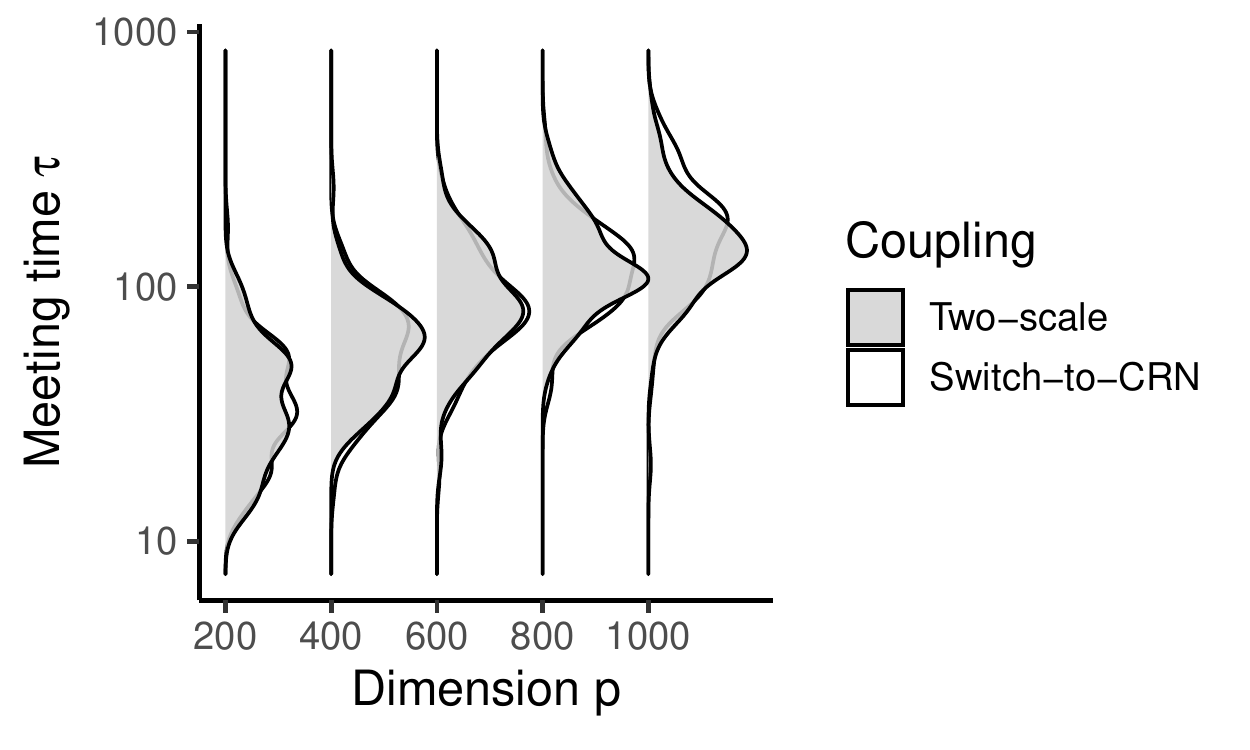}
  \caption{Meeting times for posteriors from the Horseshoe prior under the two-scale and the switch-to-CRN coupling algorithms.
  $n=100,s=10, \sigma_*=0.5$.
  }
    \label{fig:switch_to_crn_coupling_plot}
\end{figure}

%\paragraph{Analysis of meeting times.} 
We can establish that the meeting times under 
this \textit{switch-to-CRN} coupling have exponential tails and hence 
finite expectation. 

\begin{Prop} \label{prop:switch_crn_meeting}
Consider the blocked Gibbs sampler in \eqref{eq:blocked_gibbs_ergodicity_version}. 
We follow the setup in Proposition \ref{prop:drift}, and assume that the global shrinkage 
prior $\pi_\xi(\cdot)$ has a compact support. Denote $Z_t = (\beta_t, \xi_t, \sigma^2_t)$ and 
$\tilde{Z}_t = ( \tilde{\beta}_t, \tilde{\xi}_t, \tilde{\sigma}^2_t)$.
Consider the switch-to-CRN coupling given by 
$(Z_t, \tilde{Z}_t) \mapsto (\eta_{t+1}, \tilde{\eta}_{t+1}) \mapsto (Z_{t+1}, \tilde{Z}_{t+1})$
where $(\eta_{t+1}, \tilde{\eta}_{t+1})$ are maximally coupled component-wise
(under any fixed or random ordering of the components) until the 
first failed meeting, after which common random numbers are employed, and 
$(Z_{t+1}, \tilde{Z}_{t+1})$ are coupled using common random numbers. Denote
the meeting time by 
$\tau := \inf \{ t \geq 0: Z_t = \tilde{Z}_t \}$. Then
%\begin{flalign}
$\mathbb{P}( \tau > t ) \leq A_1 \kappa_1^t$
%\end{flalign}
for some constants $A_1 \in (0, \infty)$ and $\kappa_1 \in (0,1)$, and for all $t \geq 0$.
\end{Prop}

As in Proposition \ref{prop:one_scale_meeting}, our proof of Proposition
\ref{prop:switch_crn_meeting} implies a rate $\kappa_1$ which tends to $1$
exponentially as dimension $p$ increases. To obtain more favorable rates with
respect to the dimension as hinted by Figure
\ref{fig:switch_to_crn_coupling_plot} would require a better
understanding of the CRN coupling, for example establishing a 
contraction in some parts of the state space as discussed above.

%the effect of using common random numbers,
%related to the discussion of meeting times resulting from the two-scale coupling.
%, and is an area of future work.

Overall Figures \ref{fig:two_scale_coupling_plot} and \ref{fig:switch_to_crn_coupling_plot}
show significant empirical improvements in 
using the two-scale coupling and related strategies, compared to the one-scale
coupling of Section \ref{subsec:max_coupling}. Therefore, we recommend the use 
of two-scales couplings and variants such as the switch-to-CRN coupling 
for practical high-dimensional applications. Furthermore, in settings where estimating 
the metric $d$ is computationally expensive or analytically intractable, 
variants such as the switch-to-CRN coupling may offer faster numerical run-times per iteration.
Henceforth, we will use the two-scale coupling of Section
\ref{subsec:two_scale_coupling} in our experiments 
and GWAS dataset applications, unless specified otherwise.

%\paragraph{Other variants.} We can combine the switch-to-CRN coupling with the
%two-scale coupling to form other coupling strategies. For example, under the
%two-scale coupling setup, we could apply the switch-to-CRN coupling only when
%we are close with respect to our chosen metric and corresponding threshold, and
%perform CRN coupling when far away. Alternatively, we could always apply the
%switch-to-CRN coupling, and use component-wise meeting probability estimates to
%choose the order of the updates, such that components with higher meeting
%probabilities are updated first.  Overall many coupling strategies seem to be
%available in the present setting.
%
\subsection{Computational and statistical impact of degree of freedom $\nu$} \label{subsec:degree_of_freedom}

We empirically investigate the impact of the degree of freedom $\nu \geq 1$ for
Half-$t(\nu)$ priors.  Higher degrees of freedom $\nu$ corresponds to priors on
$\beta$ which have stronger shrinkage towards zero (Proposition
\ref{prop:marginal_beta_posterior}).  We consider meeting times resulting from
the two-scale coupling of Section \ref{subsec:two_scale_coupling}, as well as
the statistical performance of the posteriors.  On the computation side, recent
work on convergence analysis \citep{qin2019convergenceAOS, qin2019wasserstein}
has highlighted the impact of shrinkage on convergence of simpler MCMC
algorithms. Under a common random numbers coupling, \citet{qin2019wasserstein}
have studied the convergence of the Gibbs sampler
of 
\citet{albert1993bayesianJASA} for Bayesian probit regression, proving dimension-free convergence rates for priors with
sufficiently strong shrinkage towards zero. On the statistical estimation side,
Half-$t(\nu)$ priors have been long proposed \citep{gelman2006priorBA,
carvalho2009handlingPMLR, carvalho2010theBIOMETRIKA} for Bayesian hierarchical
models. \citet{vanderpas2014theEJS, vanderpas2017adaptiveEJS} have established
(near) minimax optimal estimation for the Horseshoe ($\nu=1$) in the Normal
means model. \citet{ghosh2017asymptoticBA} have extended the results of
\citet{vanderpas2014theEJS} to show minimax optimality of a wider class of
priors including Half-$t(\nu)$ priors. Recently, \citet{song2020bayesianEJS}
has established that the degree of freedom $\nu$ can impact the multiplicative
constant in the posterior contraction rate for the Normal means model.
Posterior contraction rates of continuous shrinkage priors in the 
regression setting beyond the Normal means model would deserve further
investigation.

%\begin{figure}
%\captionsetup[subfigure]{font=normalsize,labelfont=normalsize}
%    \centering
%    \begin{subfigure}[b]{0.32\textwidth}
%        \includegraphics[width=\textwidth]{../images/half_t_degree_of_freedom_plot/t_dist_df_plot_crn.pdf}
%        %\caption{$(n, p, s, \sigma_{*})=(100,1000,20,2)$}
%        \caption{}
%        \label{fig:t_dist_df_plot_crn}
%    \end{subfigure}
%    ~ %add desired spacing between images, e. g. ~, \quad, \qquad, \hfill etc. 
%    %(or a blank line to force the subfigure onto a new line)
%    \begin{subfigure}[b]{0.32\textwidth}
%        \includegraphics[width=\textwidth]{../images/half_t_degree_of_freedom_plot/half_t_degree_of_freedom_vary_p.pdf}
%        %\caption{$(n, s, \sigma_{*})=(100,20,2)$}
%        \caption{}
%        \label{fig:half_t_degree_of_freedom_vary_p}
%    \end{subfigure}
%    \begin{subfigure}[b]{0.32\textwidth}
%    	\includegraphics[width=\textwidth]{../images/half_t_degree_of_freedom_plot/half_t_degree_of_freedom_nu_2_vary_p.pdf}
%        %\caption{$(n, s, \sigma_{*})=(100,20,2)$}
%        \caption{}
%        \label{fig:half_t_degree_of_freedom_nu_2_vary_p}
%    \end{subfigure} 
%    \caption{Meeting times for Half-$t(\nu)$ priors under the two-scale coupling algorithm.
%    Unless specified otherwise, $n=100,p=1000,s=20$ and $\sigma_*=2$.}
%    \label{fig:half_t_coupling_plot}
%\end{figure}

\begin{figure}[!t]
\captionsetup[subfigure]{font=normalsize,labelfont=normalsize}
    \centering
    \begin{subfigure}[b]{0.43\textwidth}
        \includegraphics[width=\textwidth]{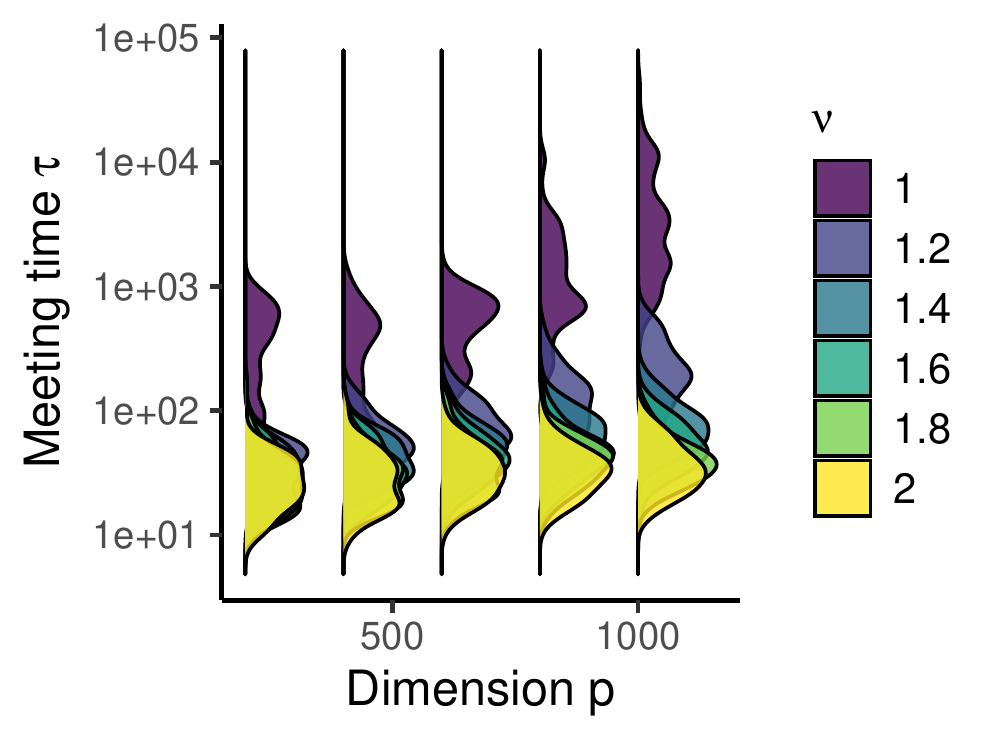}
        %\caption{$(n, p, s, \sigma_{*})=(100, 500, 20, 2)$}
        \caption{Meeting times}
        \label{fig:half_t_degree_of_freedom_vary_p}
    \end{subfigure}
    \begin{subfigure}[b]{0.43\textwidth}
        \includegraphics[width=\textwidth]{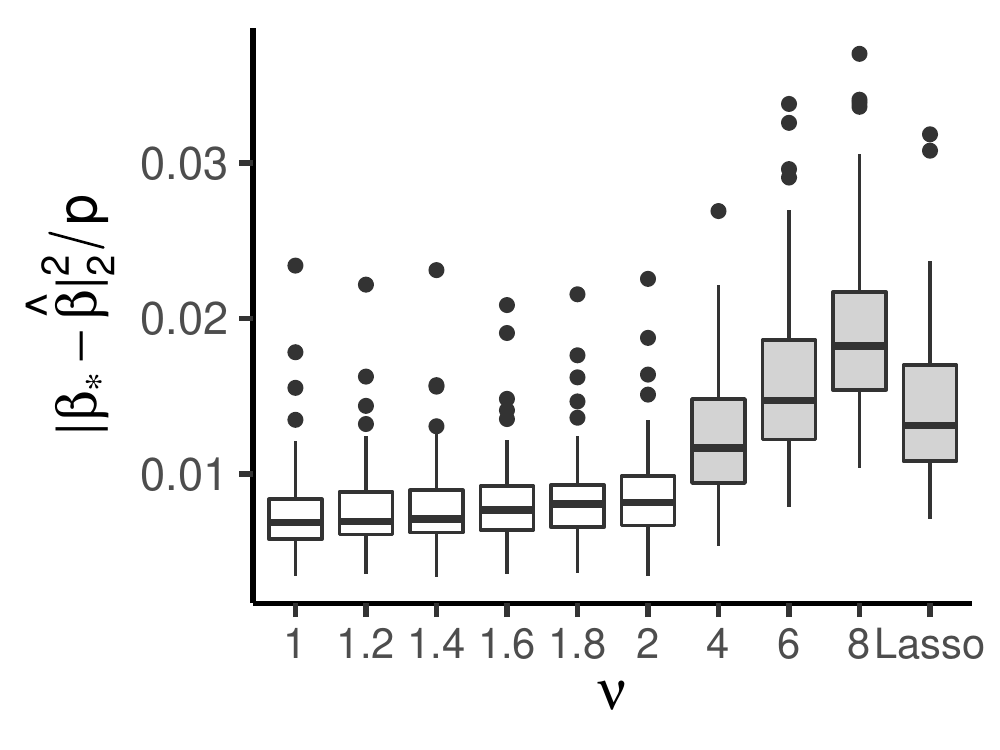}
        %\caption{$(n, s, \sigma_{*})=(100, 20, 2)$}
        \caption{Mean Squared Error}
        \label{fig:beta_mse}
    \end{subfigure}
    \caption{Meeting times under the two-scale coupling algorithm and 
    mean squared error with Half-$t(\nu)$ priors and LASSO estimates. 
    In Figure \ref{fig:half_t_degree_of_freedom_vary_p} $n=100,s=20$ and $\sigma_*=2$; 
    in Figure \ref{fig:beta_mse} $n=100,p=500,s=20, \sigma_*=2$.}
    \label{fig:half_t_coupling_plot}
\end{figure}

%\begin{figure}[!t]
%\captionsetup[subfigure]{font=normalsize,labelfont=normalsize}
%    \centering
%    \begin{subfigure}[b]{0.32\textwidth}
%        \includegraphics[width=\textwidth]{../images/half_t_degree_of_freedom_plot/half_t_llag_tv_bounds.pdf}
%        %\caption{$(n, p, s, \sigma_{*})=(100, 500, 20, 2)$}
%        \caption{}
%        \label{fig:half_t_llag_tv_bounds}
%    \end{subfigure}
%    \begin{subfigure}[b]{0.32\textwidth}
%        \includegraphics[width=\textwidth]{../images/half_t_degree_of_freedom_plot/half_t_mixing_times.pdf}
%        %\caption{$(n, s, \sigma_{*})=(100, 20, 2)$}
%        \caption{}
%        \label{fig:half_t_mixing_times}
%    \end{subfigure}
%    \begin{subfigure}[b]{0.32\textwidth}
%        \includegraphics[width=\textwidth]{../images/half_t_degree_of_freedom_plot/half_t_mixing_times_degree_2.pdf}
%        %\caption{$(n, s, \sigma_{*})=(100, 20, 2)$}
%        \caption{}
%        \label{fig:half_t_mixing_times_nu_2}
%    \end{subfigure}
%    \caption{$L$-lag total variation bounds for Half-$t(\nu)$ priors under the two-scale coupling algorithm.
%    Unless specified otherwise, $n=100,p=500,s=20$ and $\sigma_*=2$.}
%    \label{fig:half_t_llag_plots}
%\end{figure}

%\paragraph{Meeting times.} 
We consider the meeting times obtained with Half-$t(\nu)$ priors under the
two-scale coupling (with $R=1$ and threshold $d_0=0.5$) on synthetic datasets.
The synthetic datasets are generated as per Figure \ref{fig:max_coupling_plot}
of Section \ref{subsec:max_coupling}. 
Figures \ref{fig:half_t_degree_of_freedom_vary_p} is then based on $20$
synthetic datasets generated independently for different degrees of freedom
$\nu \in [1,2]$.  Figure \ref{fig:half_t_degree_of_freedom_vary_p} plots
meeting times $\tau$ of $1$-lag coupled Markov chains against dimension.  It
indicates that even a small increase in the degree of freedom $\nu > 1$
compared to the Horseshoe ($\nu=1$) can lead to orders of magnitude
computational improvements through shorter meeting times. 

We also consider the statistical performance of Bayesian shrinkage regression with
Half-$t(\nu)$ priors on synthetic datasets in Figure \ref{fig:beta_mse}. 
The synthetic datasets are generated as per Figure \ref{fig:max_coupling_plot} of Section
\ref{subsec:max_coupling}. We use the blocked Gibbs sampler of Algorithm
\ref{algo:Half_t_Exact} to draw samples from the posteriors corresponding to
different degrees of freedom $\nu$, by simulating chains of length $1000$, with
a burn-in period chosen using the coupling-based total variation distance upper bounds 
of \citet{biswas2019estimating}.
Specifically we employ a burn-in of $300$ steps for $\nu \geq 1.2$ and of $600$ steps for $\nu = 1$. 
%Figure \ref{fig:half_t_stat_performance_plot} highlights the statistical performance of Half-$t(\nu)$ 
%priors under different values of $\nu \geq 1$. 
Figure \ref{fig:beta_mse} shows similar mean squared error (MSE) for different 
values of $\nu \geq 1$. For a fixed synthetic dataset, the MSE in Figure 
\ref{fig:beta_mse} is calculated as $\| \beta_* - \hat{\beta} \|_{2}^2/p$ 
where 
$\hat{\beta} := \sum_{t=1}^{1000} \beta_t / 1000$,
is the MCMC estimator obtained from the Markov chain after burn-in. 
The error bars in Figure \ref{fig:beta_mse} are generated by simulating synthetic datasets 
independently $100$ times and each time calculating the corresponding MSE from a Markov chain 
for different values of $\nu \in [1,2]$. 
The MSE of the LASSO is also included, 
with regularization parameter chosen with cross-validation using the 
$\texttt{glmnet}$ package \citep{friedman2010regularizationJOSS}.  For $\nu \in [1,2]$, 
we observe a lower MSE using Half-$t(\nu)$ priors compared to the LASSO. 
The grey plots in Figures \ref{fig:beta_mse} show that much larger 
$\nu \in \{4,6,8\}$ can lead to higher MSE, as the corresponding posteriors 
are strongly concentrated about zero and less able to identify non-null signals. 
Overall, Figure \ref{fig:beta_mse} suggests
that Half-$t(\nu)$ priors with degrees of freedom $\nu \in [1,2]$ 
can result in comparable statistical performance 
while much larger values of $\nu$ should be discouraged.

\section{Results on GWAS datasets} \label{section:gwas}
Section \ref{subsec:degree_of_freedom} suggests 
that Half-$t(\nu)$ priors with higher degrees of freedom $\nu > 1$
can give similar statistical performance to the Horseshoe ($\nu = 1$),
whilst allowing orders of magnitude computational improvements. 

Motivated by this observation, we apply our algorithms to genome-wide association study 
(GWAS) datasets using $\nu = 2$. We consider a bacteria dataset \citep{buhlmann2014highARSA} 
with $n=71$ observations (corresponding to production of the vitamin riboflavin), 
and $p=4,088$ covariates (corresponding to single nucleotide polymorphisms (SNPs) in the genome), 
and a maize dataset \citep{romay2013genotypingGENOMEBIO, liu2016iterativePLOSGENETICS, 
zeng2017nonparametricNATURECOMM, johndrow2020scalableJMLR}
with $n=2,266$ observations (corresponding to average number of days taken for silk emergence
in different maize lines), and $p=98,385$ covariates (corresponding to
SNPs). Bayesian methods are well-suited to such GWAS datasets, as they 
provide interpretable notions of uncertainty through marginal posterior 
probabilities which allow for multimodality, 
enable the use of prior information, and account for the uncertainty 
in parameter estimates \citep[e.g.][]{guan2011bayesianAOAS, zhou2013polygenicPLOS}.
Furthermore, heavy-tailed continuous shrinkage priors can be particularly effective in the 
GWAS setting, as the associations between SNPs and the phenotype are expected to be sparse,
and such priors have shown competitive empirical performance in polygenic prediction 
\citep{ge2019polygenicNATCOMM}. 

For both datasets, we target the Half-$t(2)$ posterior with $a_0=b_0=1$. 
We use the two-scale coupling with $R=1$ and threshold $d_0= 0.5$, 
and initialize chains independently from the prior. 
For the $\xi$ updates, we use a Metropolis--Rosenbluth--Teller--Hastings 
step-size of $\sigma_{\text{MRTH}}=0.8$. 
Based on 100 independent coupled chains,
Figure \ref{fig:gwas_llag_tv} shows 
upper bounds on the total variation distance to stationarity, $\text{TV}(\pi_t, \pi)$
as a function of $t$,
using $L$-lag couplings 
with $L=750$ and $L=200$ for the maize and riboflavin datasets respectively.
The results indicate that these Markov chains converge to the corresponding 
posterior distributions in less than $1000$ and $500$ iterations respectively.
We note that the lags $L=750$ and $L=200$ were selected based on 
some preliminary runs of the coupled chains to obtain informative total variation bounds, 
which incurs an additional preliminary cost. Any choice of lag results in valid upper bounds,
but only large enough lags lead to upper bounds that are in the interval $[0,1]$ even for small iterations $t$. 

In these GWAS examples, the run-time per iteration may be significant.
For the maize dataset, on a 2015 high-end laptop,
each iteration of the coupled chain takes approximately $60$ seconds, and running 
one coupled chain until meeting can take more than one day. This run-time is dominated 
by the calculation of the weighted matrix cross-product $X \text{Diag}(\eta)^{-1} X^T$, 
with an $\mathcal{O}(n^2p)$ cost. In this setting, 
a scientist wanting to run chains for $10^4$ or $10^5$ iterations may have
to wait for weeks or months. 
Using the proposed couplings, a scientist with access to parallel computers
can be confident that samples obtained after $10^3$ iterations
would be indistinguishable from perfect samples from the target;
for a similar cost they can also obtain unbiased estimators 
as described in \citet{jacob2020unbiasedJRSSB}. 
For the riboflavin dataset, each iteration of the coupled chain takes
approximately $0.1$ seconds, and running one coupled chain until meeting takes
only $10$ to $20$ seconds.  The run-time there is dominated by the component-wise
$\eta$ updates which have $\mathcal{O}(p)$ cost.  In the riboflavin example, 
the coupling-based diagnostics enable the scientist to perform reliable Bayesian computation
directly from personal computers. Overall, the GWAS examples emphasize that 
couplings of MCMC algorithms can aid practitioners in high-dimensional
settings.

\begin{figure}[!t]
\captionsetup[subfigure]{font=normalsize,labelfont=normalsize}
    \centering
        \begin{subfigure}[b]{0.43\textwidth}
        \includegraphics[width=\textwidth]{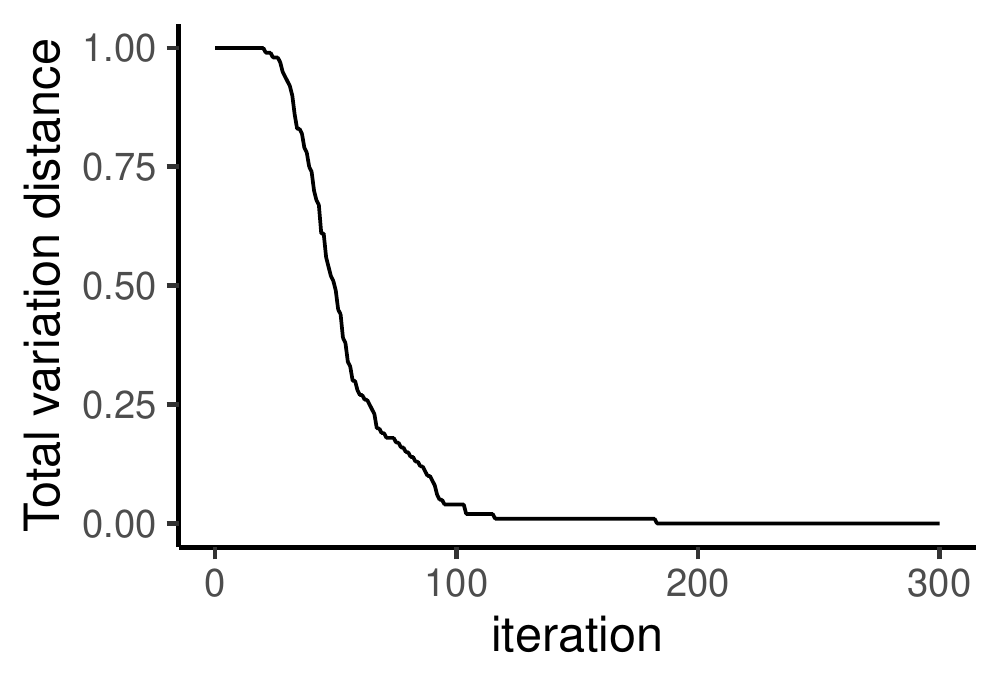}
        \caption{Riboflavin dataset with $n=71$ and $p=4,088$ \ }
        \label{fig:gwas_llag_tv_riboflavin}
    \end{subfigure}
            \begin{subfigure}[b]{0.43\textwidth}
        \includegraphics[width=\textwidth]{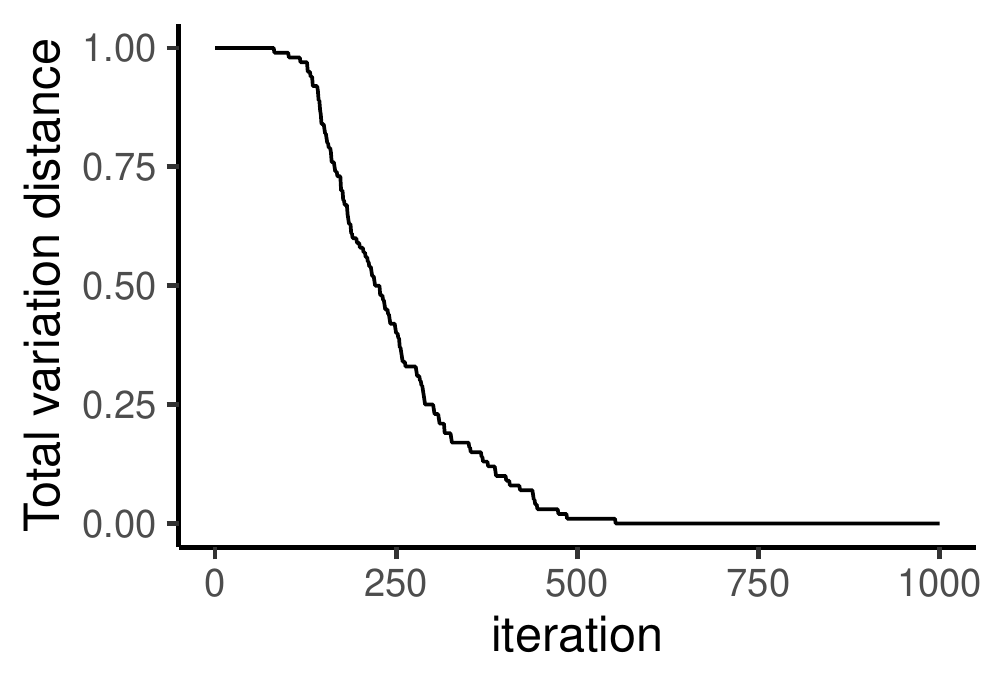}
        \caption{Maize dataset with $n=2,266$ and $p=98,385$ \ }
        \label{fig:gwas_llag_tv_maize}
    \end{subfigure}
%    \begin{subfigure}[b]{0.63\textwidth}
%        \includegraphics[width=\textwidth]{../images/dataset_examples/gwas_LlagTVUB_combined.pdf}
%    \end{subfigure}
    \caption{$L$-Lag coupling-based upper bounds on $\text{TV}(\pi_t, \pi)$,
    the distance between the chain at time $t$ and its limiting distribution,
      in Bayesian shrinkage regression with Half-$t(2)$ prior on GWAS datasets.
    Maize dataset: $n=2,266$ and $p=98,385$; Riboflavin dataset: $n=71$ and $p=4,088$.  }
    \label{fig:gwas_llag_tv}
\end{figure}

%\begin{figure}[!t]
%\captionsetup[subfigure]{font=normalsize,labelfont=normalsize}
%    \centering
%    \begin{subfigure}[b]{0.32\textwidth}
%        \includegraphics[width=\textwidth]{../images/dataset_examples/gwas_meetings.pdf}
%        \caption{Meeting times $\tau$ \\ \ }
%        \label{fig:gwas_meeting_time}
%    \end{subfigure}
%    \begin{subfigure}[b]{0.32\textwidth}
%        \includegraphics[width=\textwidth]{../images/dataset_examples/gwas_LlagTVUB.pdf}
%        \caption{$L$-lag TV bounds \\ \ }
%        \label{fig:gwas_llag_tv}
%    \end{subfigure}
%    \begin{subfigure}[b]{0.32\textwidth}
%        \includegraphics[width=\textwidth]{../images/dataset_examples/gwas_traceplots.pdf}
%        \caption{Marginal posterior densities of three components}
%        \label{fig:gwas_component_1}
%    \end{subfigure}
%    \caption{Coupled Bayesian shrinkage regression with Half-$t(2)$ prior on a GWAS dataset.}
%    \label{fig:gwas}
%\end{figure}

\section{Discussion} \label{section:discussion}
% Notes: mention how parallel computation allows bigger coupled inference than done so far in conclusion. 
% Mention tv/ convergence estimation via simulated coupled MCMC capture finer upper bounds than d & m coupling based proof. 

We have introduced couplings of Gibbs samplers
for Bayesian shrinkage regression with Half-$t(\nu)$ priors. 
The proposed two-scale 
coupling is operational in realistic 
settings, including a 
GWAS setting with $n \approx 2,000, p \approx 100,000$.

%Firstly, our work participates in an effort to apply MCMC algorithms
%and coupling techniques to large, high-dimensional problems.
%The problems considered here are orders of magnitude larger
%in dimension than 
%the related experiments on Bayesian variable selection in \citet{jacob2020unbiasedJRSSB}.

Firstly, our work participates in a wider effort to apply MCMC algorithms
and coupling techniques to ever more challenging settings 
\citep{middleton2020unbiasedEJS,Ruiz2020UnbiasedGE,Trippe2021OptimalTC,
lee2020coupled, xu2021couplings}. In particular, the short meeting times 
obtained here vindicate the use of coupling techniques in 
multimodal, high-dimensional sampling problems. 
This relates to questions raised in some comments of \citet{lee2020comment} 
and \citet{paulin2020comment} in the discussion of
\citet{jacob2020unbiasedJRSSB}, and shows that the dimension
of the state space may not necessarily be the most important driver of 
the performance of MCMC algorithms and couplings thereof. 

%This relates to some questions raised by e.g. the comments of \citet{lee2020comment} 
%and \citet{paulin2020comment} in the discussion of
%\citet{jacob2020unbiasedJRSSB}. 
Secondly, we have applied $L$-lag couplings 
\citep{biswas2019estimating} to obtain upper bounds on the total variation
distance between the Markov chain at a finite time and its stationary distribution.
This allows empirical investigations of the mixing time
and how it varies 
with the inputs of the problem, including number of observations, 
number of covariates, signal to noise ratio and sparsity. 
We find that coupling techniques constitute
a convenient  non-asymptotic tool to monitor the performance of MCMC algorithms.

Thirdly, we observe that Half-$t(\nu)$ priors with degrees of freedom $\nu$
higher than one give similar statistical estimation performance than the Horseshoe
whilst providing significant computational advantages. 
This contributes towards the discussion on the impact of the prior 
on the trade-off between statistical estimation and computational
feasibility, in the setting of high-dimensional Bayesian regression.

The following questions arise from our work.
\begin{enumerate}
\item \textit{Convergence of the blocked Gibbs sampler.} 
Short meeting times suggest that the blocked Gibbs sampler converges quickly,
even in high dimensions. This motivates a more formal 
study of the convergence 
of that Markov chain, to understand how the convergence rate
varies with features of
the data generating process and of the prior. The Lyapunov function 
in Proposition \ref{prop:drift} and the two-scale coupling may prove useful for such analysis. 
Our initial work in this area suggests that finding a convenient metric that gives 
sharp bounds on the metric $d$ used here, while simultaneously being amenable to
theoretical analysis, will be a key step. 
%Choosing a coupling is often among
%the most difficult aspects of using couplings for theoretical analysis. That we already have
%a general strategy for constructing a coupling that empirically gives rapid meetings
%and shows little dependence on dimension is a good start.

\item \textit{Alternative coupling algorithms.} 
%  Section \ref{subsec:coupling_variants}
%  suggests various couplings, and more can undoubtedly
%  be found.
  Section \ref{subsec:coupling_variants} indicates that many coupling strategies 
  are available in the present setting, and more generally for various
  Gibbs samplers. 
  For example, we could combine the switch-to-CRN coupling with the
  two-scale coupling to form other coupling strategies. Under the
  two-scale coupling setup, we could apply the switch-to-CRN coupling only when
  the chains are close with respect to a chosen metric, and
  employ CRN couplings when far away. Alternatively, we could always apply the
  switch-to-CRN coupling, and estimate component-wise meeting probabilities to
  select the order of the updates, for example such that components more
  likely to meet are updated first. Some couplings may result in shorter meeting times 
for the Horseshoe prior. This may allow the Horseshoe prior to 
remain competitive with Half-$t(\nu)$ priors
with higher degrees of freedom. In our experiments
we did not find ways to obtain shorter meetings 
for the Horseshoe, but we certainly cannot rule out that 
possibility. If one could obtain short meetings for the Horseshoe
in high dimensions, the apparent trade-off between statistical performance and computing cost
identified in the present work would disappear.
    
\item \textit{Interplay between posterior concentration and MCMC convergence.}
For Bayesian regression with spike-and-slab priors, \citet{yang2016onAOS} and 
\citet{atchade2019approximate} have shown that 
posterior contraction can aid the convergence of MCMC in 
high dimensions. The performance of the proposed couplings
motivates similar investigations for continuous shrinkage priors. 
\end{enumerate}

\paragraph{Acknowledgement.} %\label{section:acknowledgement}
The authors thank Xiaolei Liu and Xiang Zhou for sharing the Maize GWAS data, 
and Yves F. Atchad\'{e}, Qian Qin and Neil Shephard for helpful comments. 
Niloy Biswas gratefully acknowledges support from a GSAS Merit Fellowship and
a Two Sigma Fellowship Award. 
Anirban Bhattacharya gratefully acknowledges support 
by the National Science Foundation through an NSF CAREER award (DMS-1653404).
Pierre E. Jacob gratefully acknowledges support by the National Science Foundation 
through grants DMS-1712872 and DMS-1844695.
The maize dataset GWAS simulations in this article were run on the FASRC Cannon cluster supported 
by the FAS Division of Science Research Computing Group at Harvard University.
\bibliographystyle{abbrvnat}
\bibliography{references.bib}

\appendix

%\begin{center}
%{\large\bf SUPPLEMENTARY MATERIALS}
%\end{center}

In Section \ref{appendices:mcmc_challenges}, we describe some features of the posterior distributions arising in high-dimensional regression with shrinkage priors. 
In Section \ref{appendices:metric_details} we 
describe metrics to monitor the distance between pairs of chains and
provide details on the calculation of metrics to use in two-scale coupling strategies.
In Section \ref{appendices:proofs} we provide the proofs for all the theoretical results in this document.
%[writing this I am realizing that Appendix A.2 \& A.3 could be placed in the same section as Appendix C]
Section \ref{appendices:algos} provides pseudo-codes for various algorithms mentioned in this document. Section \ref{appendices:algo_derivations} provides some derivations of certain steps in the Gibbs sampler.
%[writing this, I am wondering if Appendices D \& E could be merged into one section].

%\section{Additional figures and discussion} \label{appendices:add_figures_discussion}
\section{Features of the target complicate use of generic MCMC} \label{appendices:mcmc_challenges}

While posterior distributions resulting from heavy-tailed shrinkage priors 
%with a pole at the origin 
have desirable statistical properties, their 
features pose challenges to generic MCMC
algorithms.
%, as we illustrate for Half-$t(\nu)$ priors. 
For Half-$t(\nu)$ priors, we show that the resulting posterior distributions 
present 1) multimodality, 2) heavy tails and 3) poles at zero. This hints at a trade-off
between statistical accuracy and computational difficulty, since the very
features that present computational challenges are crucial for optimal
statistical performance across sparsity levels and signal strengths.

Proposition \ref{prop:marginal_beta_posterior} gives the marginal prior and
posterior densities of $(\beta, \sigma^2, \xi)$ up to normalizing constants,
i.e. integrating over $\eta$. For any real-valued functions $f$ and $g$, we
write $f(x) \asymp g(x)$ for $x \rightarrow x_0$ when $m < \liminf_{x
\rightarrow x_0} \big| f(x)/g(x) \big| \leq \limsup_{x \rightarrow x_0}
\big| f(x)/g(x) \big| < M $ for some $0<m\leq M < \infty$.

\begin{Prop} \label{prop:marginal_beta_posterior}
The marginal prior of component $\beta_j$ on $\mathbb{R}$ given $\xi$ and $\sigma^2$ has density
\begin{equation} \label{eq:marginal_beta_prior}
\pi(\beta_j | \xi, \sigma^2) \propto U \Big( \frac{1+\nu}{2}, 1, \frac{\xi \beta_j^2}{2 \sigma^2 \nu} \Big) \asymp \begin{cases}
- \log|\beta_j| & \text{ for } \; |\beta_j| \rightarrow 0 \\
|\beta_j|^{-(1+\nu)} & \text{ for } \; |\beta_j| \rightarrow +\infty 
\end{cases},
\end{equation} 
where $\Gamma(\cdot)$ is the Gamma function, and $U(a,b,z) := \Gamma(a)^{-1} \int_0^\infty x^{a-1}(1+x)^{b-a-1}e^{-zx}dx$ is the confluent hypergeometric function of the second kind \citep[Section 13]{abramowitz1964handbook} for any $a, b, z>0$. 
The marginal posterior density of $\beta$ on $\mathbb{R}^p$ given $\xi$ and $\sigma^2$ is
\begin{flalign}
\pi( \beta | \sigma^2, \xi,  y ) &\propto \mathcal{N}( y; X \beta, \sigma^2 I_n ) \prod_{j=1}^p U \Big( \frac{1+\nu}{2}, 1, \frac{\xi \beta_j^2}{2 \sigma^2 \nu} \Big),
\label{eq:marginal_beta_fixed_xi_sigma2}
%\text{ and } \\
%- \frac{\partial }{\partial \beta_j} \log \pi(\beta | \sigma^2, \xi,  y ) &= -\Big[\frac{1}{\sigma^2} X^T (y-X \beta)\Big]_j + \frac{1+\nu}{2} \frac{U \Big( \frac{1+\nu}{2}+1, 2, \frac{\xi \beta^2_j}{2 \sigma^2 \nu} \Big)}{U \Big( \frac{1+\nu}{2}, 1, \frac{\xi \beta^2_j}{2 \sigma^2 \nu} \Big)} \frac{\xi \beta_j}{ \sigma^2 \nu}.
\end{flalign}
thus $ \pi(\beta | \sigma^2, \xi,  y ) \overset{\|\beta \| \rightarrow 0}{\asymp} - \prod_{j=1}^p \log ( |\beta_j| ) $ and $- \frac{\partial }{\partial \beta_j} \log \pi(\beta | \sigma^2, \xi,  y ) \overset{\beta_j \rightarrow 0}{\asymp} - (\beta_j \log |\beta_j|)^{-1}  $. 
\end{Prop}

Figure \ref{fig:beta_marginal_plot} gives an example of such marginal posterior
in a stylized setting. Here 
\[n=2, \quad p=3, \quad X = \begin{pmatrix} 1 & 1
  & 0\\ 1 & 0 & 1 \end{pmatrix},\quad y = X  (1\; 0\; 0)^T = (1\; 1)^T,\]
with $\nu=2$ and $\xi = \sigma^2=0.25$. 
%Figure \ref{fig:beta_marginal_plot} was generated by
%computing the posterior density \eqref{eq:marginal_beta_fixed_xi_sigma2} for a grid of
%values of $\beta_1,\beta_2,\beta_3$, and then marginalizing over one of the
%three components to obtain bivariate probability density functions. 
Component $\beta_1$ in Figure \ref{fig:beta_marginal_plot} illustrates
the potential multimodality in the posterior. The posterior distribution
also has polynomial tails in high dimensions, as stated in Corollary
\ref{cor:marginal_beta_posterior_heavy_tails}. 

\begin{figure}%[!h]
\centering
\includegraphics[width=\textwidth]{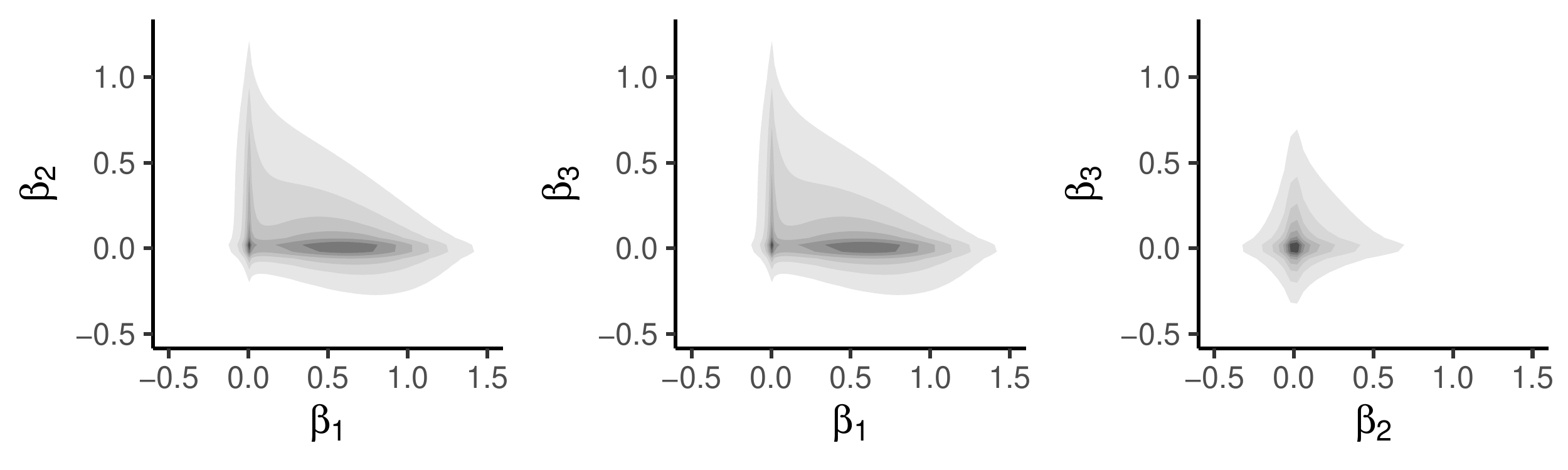}
\caption{Posterior densities $\pi( \beta_1, \beta_2 | \sigma^2, \xi,  y ), \pi( \beta_1, \beta_3 | \sigma^2, \xi,  y )$ and $\pi( \beta_2, \beta_3 | \sigma^2, \xi,  y )$.}
\label{fig:beta_marginal_plot}
\end{figure}

\begin{Cor} \label{cor:marginal_beta_posterior_heavy_tails}
Let $e^{\perp} \in \mathbb{R}^p$ be a unit vector such that $X e^{\perp} = 0$
and let $\lambda>0$. Assume entries $e^{\perp}_j \in \mathbb{R} $ are non-zero 
for all $j=1,\ldots,p$, which is necessary to ensure that 
$\pi( \lambda e^{\perp} | \sigma^2, \xi,  y )$ is
finite for all $\lambda>0$. Then
\begin{equation}
\pi( \lambda e^{\perp} | \sigma^2, \xi,  y ) \asymp \lambda^{-p(1+\nu)}
\end{equation}
%\begin{equation}
%\liminf_{ \lambda \rightarrow + \infty } \ \lambda^{p(1+\nu)} \pi( \lambda e^{\perp} | \sigma^2, \xi,  y ) \ > \ 0.
%\end{equation}
as $\lambda \rightarrow + \infty$. 
That is, the marginal posterior of $\beta$ on $\mathbb{R}^p$ 
has polynomial tails along all directions in the null space $\ker(X)$ of $X$. 
%\textcolor{blue}{JEJ: clash with notation for the Normal distribution. Choose a different script (maybe Roman $N$) or perhaps use the equivalent term ``kernel'' here.}
\end{Cor}

In our stylized example, $e^{\perp} = 3^{-1/2} (1\, -1 \, -1)^T$
satisfies $X e^{\perp} = 0$. It follows from Corollary
\ref{cor:marginal_beta_posterior_heavy_tails} that $ \pi( \lambda e^{\perp}
| \sigma^2, \xi,  y ) \asymp \lambda^{-3(1+\nu)} \asymp \lambda^{-9}) $ for
large $\lambda$. More generally, whenever $p > n$ and the marginal prior on
$\beta$ has polynomial tails, the subspace $\ker(X)$ of $\mathbb{R}^p$ has
vector space dimension at least $(p-n)$, and therefore the posterior has
polynomial tails over a non-trivial subspace. For such heavy-tailed targets,
the most popular ``generic'' MCMC algorithms encounter significant theoretical
and practical challenges. Here, we use the term generic for algorithms that
are not explicitly designed to target this class of distributions, but
instead rely on a general strategy to explore state spaces.
%, for example using local moves and gradient evaluations.
Specifically, for heavy-tailed targets, \textit{lack} of geometric ergodicity
has been established for Random Walk Metropolis--Hastings (RWMH)
\citep{jarner2003necessaryBERNOULLI}, Metropolis-Adjusted Langevin (MALA)
\citep{roberts1996exponentialBERNOULLI} and Hamiltonian Monte Carlo (HMC)
\citep{livingstone2019onBERNOULLI}\footnote{\citet{jarner2003necessaryBERNOULLI, roberts1996exponentialBERNOULLI}
and \citet{livingstone2019onBERNOULLI} show lack of geometric ergodicity for
targets which have heavy-tails in every direction.  Their arguments can be
modified to hold for targets which are heavy-tailed along any one direction as
in Corollary \ref{cor:marginal_beta_posterior_heavy_tails}, by considering a
projection of the full chain along that direction.}. Better performance
can sometimes be obtained using heavy-tailed proposals
\citep{jarner2007ConvergenceSJS,sherlock2010theSS,livingstone2019kinetic} or
variable transformations when available \citep{johnson2012variableAOS}. 

The posterior density in \eqref{eq:marginal_beta_fixed_xi_sigma2}
also approaches infinity when any component $\beta_j$ approaches zero, as 
$\lim_{\beta_j \rightarrow 0} U( (1+\nu)/2,1, (\xi \beta_j^2)/(2 \sigma^2 \nu)) = \infty$  
for any $\nu, \xi, \sigma^2> 0$. This is highlighted by the modes at the origin in Figure \ref{fig:beta_marginal_plot}. 
The \textit{pole} at zero corresponds to a discontinuity in the gradients of the negative log-density, as given in Proposition 
\ref{prop:marginal_beta_posterior}. 
Such diverging gradients can lead to numerically unstable
leapfrog integration for HMC samplers
\citep{bou-rabee2018geometricACTA}. 
%Similar difficulties arise for algorithms based
%on Piecewise-Deterministic Markov Processes (PDMPs)
%\citep{vanetti2018piecewise, fearnhead2018piecewiseSS}. 
%For state $\beta$ and a corresponding velocity $v$, 
%the Bouncy Particle Sampler \citep{bouchard-cotethe2018JASA} and
%the Zig-Zag process \citep{bierkens2019the2019AOS} both involve 
%sampling from inhomogeneous Poisson processes with intensity proportional to 
%$\max \{ 0, (- \nabla  \log \pi(\beta) )^T v \} $. This is accomplished
%by thinning \citep{devroye2003non}, which can become
%numerically unstable near the origin where
%$- \tfrac{\partial }{\partial \beta_j} \log \pi(\beta | \sigma^2, \xi,  y ) 
%\asymp - (\beta_j \log |\beta_j|)^{-1} $
%by Proposition 
%\ref{prop:marginal_beta_posterior}. Understanding ergodicity 
%properties and developing algorithms to simulate PDMPs 
%for these targets is an interesting open area for further investigation. 

%With these challenges in mind, we consider a Gibbs sampler 
%tailored for the class of target under consideration. We will show that it is geometrically
%ergodic without any additional assumptions on the design matrix $X$ or the
%values of $n$ and $p$, and we will use a coupling approach and extensive numerical experiments to assess
%its practical performance.

\section{Details on the choice and calculation of metric} \label{appendices:metric_details}
\subsection{Comparison with other metrics} \label{appendices:metric_comparison}
Figure \ref{fig:metric_plot} plots
$d(C_t, \tilde{C}_t)$ alongside the $L_1$ metric and the log-transformed
$L_1$ metric on parameters $m_{t,j}$ and $\tilde{m}_{t,j}$ as in 
\eqref{eq:distance_metric_estimate2}. It is based on one trajectory of our
two-scale coupled chain under common random numbers (threshold $d_0=0$), 
for a synthetic dataset generated as per Section \ref{subsec:max_coupling} targeting the Horseshoe
posterior ($\nu=1$) with $a_0=b_0=1$. Figure \ref{fig:metric_plot} shows
that $d(C_t, \tilde{C}_t)$ is bounded above by both the $L_1$ metric and
the $L_1$ metric between the logarithms of $m_j, \tilde{m}_j$. Therefore, we
could consider the capped $L_1$ metric ($\min \{ \sum_j | m_{t,j} -
\tilde{m}_{t,j} | , 1 \}$) or the capped $L_1$ metric on the logs ($\min \{
\sum_j | \log( |m_{t,j}/ \tilde{m}_{t,j}|) | , 1 \}$) to obtain upper
bounds on $d$. Finding metrics which are easily calculable and accurate
could be investigated further and  may prove valuable
in the theoretical analysis of the proposed algorithms.  
%The approximation of $d$ that we employ appears to be adequate for
%the purpose of achieving rapid meetings, as seen empirically below.

\begin{figure}
\captionsetup[subfigure]{font=normalsize,labelfont=normalsize}
    \centering
        \includegraphics[width=0.6\textwidth]{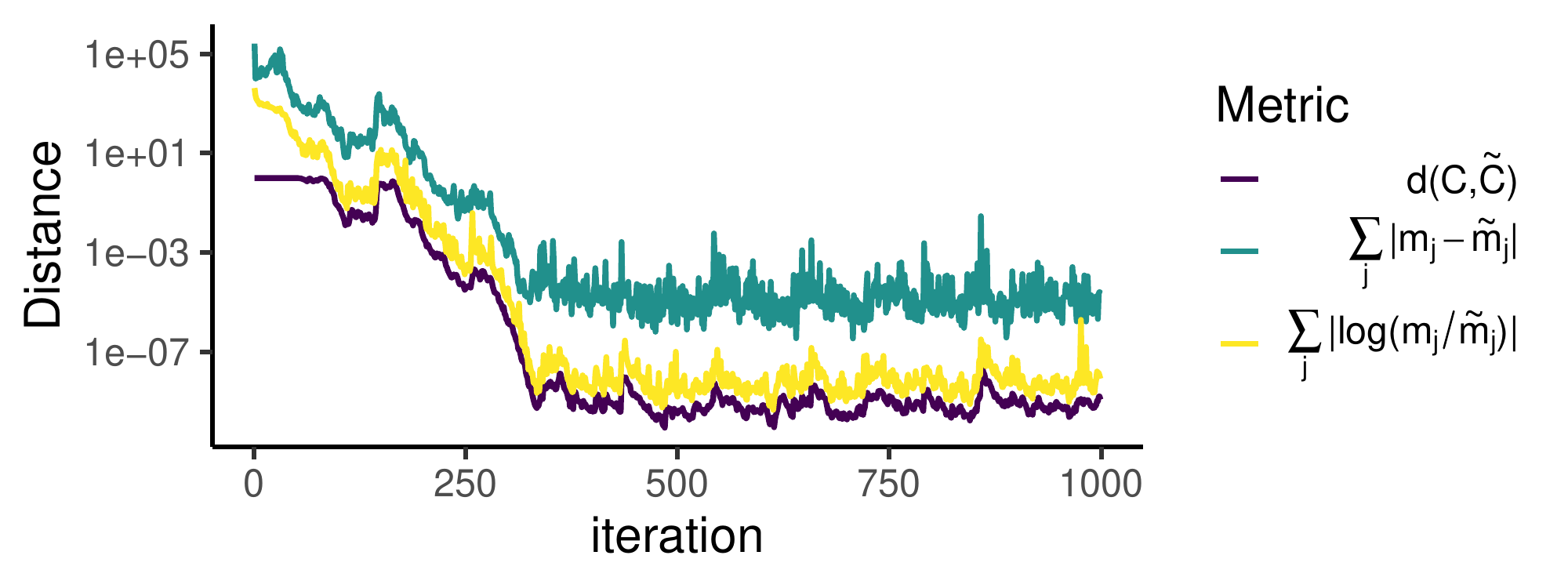}
    \caption{Metric trajectories for the posterior associated with a Horseshoe prior 
    under common random numbers coupling 
    with $n=100$, $p=1000$, $s=10$ and $\sigma_*=0.5$.}
    %$(n,p,s,\sigma_{*}) = (100,1000,10,0.5)$.}
    \label{fig:metric_plot}
\end{figure}

\subsection{Choice of metric threshold for the two-scale coupling} \label{appendices:threshold_choice}
%Figure \ref{fig:d_threshold_plot} shows that different thresholds $d_0$
%values away from $0$ and $1$ give similar meeting times. 
%It is based on 100 replicate simulations of $1$-Lag
%coupled Markov chains, where the synthetic datasets are generated as per
%Section \ref{subsec:max_coupling} and we target the Horseshoe posterior
%($\nu=1$) with $a_0=b_0=1$. 

For different values of $d_0$, Figure
\ref{fig:metric_by_threshold_p_500} shows averaged trajectories $(d(C_t,
\tilde{C}_t)) _{t \geq 0}$. When $d_0=1$ (one-scale coupling), metric
$d(C_t, \tilde{C}_t)$ remains near $1$, and when $d_0=0$ (common
random numbers coupling) metric $d(C_t, \tilde{C}_t)$ gets very close to
but does not exactly equal $0$. For all values of the threshold shown,
the chains exactly meet and $d(C_t, \tilde{C}_t)$ reaches zero at similar
times. Figure \ref{fig:crn_hist_by_p} considers higher dimensional settings. It
plots the histograms of $(d(C_t, \tilde{C}_t))_{t=1}^{2000}$ when
$d_0=0$. The histograms show that $d(C_t, \tilde{C}_t)$ takes
values close to either $1$ or $0$ for different dimensions $p$. This 
% appears to be a cutoff phenomenon, as in \citet{diaconis1996thePNAS}, and 
may explain why different threshold values sufficiently away from $0$ and $1$ give
similar meeting times. Figure \ref{fig:p_500_by_threshold} shows that different thresholds 
give similar meeting times for the Horseshoe. 

\begin{figure}[!t]
\captionsetup[subfigure]{font=normalsize,labelfont=normalsize}
    \centering
    \begin{subfigure}[b]{0.32\textwidth}
        \includegraphics[width=\textwidth]{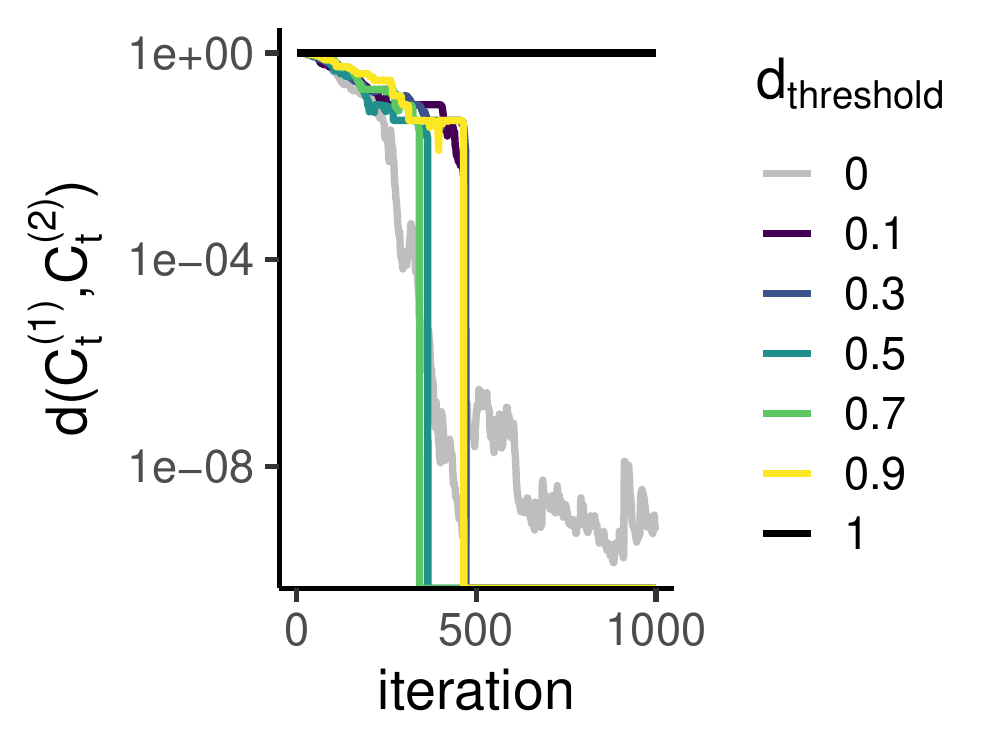}
        %\caption{$(n, p, s, \sigma_{*})=(100, 500, 20, 2)$}
        \caption{}
        \label{fig:metric_by_threshold_p_500}
    \end{subfigure}
    ~ %add desired spacing between images, e. g. ~, \quad, \qquad, \hfill etc. 
    \begin{subfigure}[b]{0.32\textwidth}
        \includegraphics[width=\textwidth]{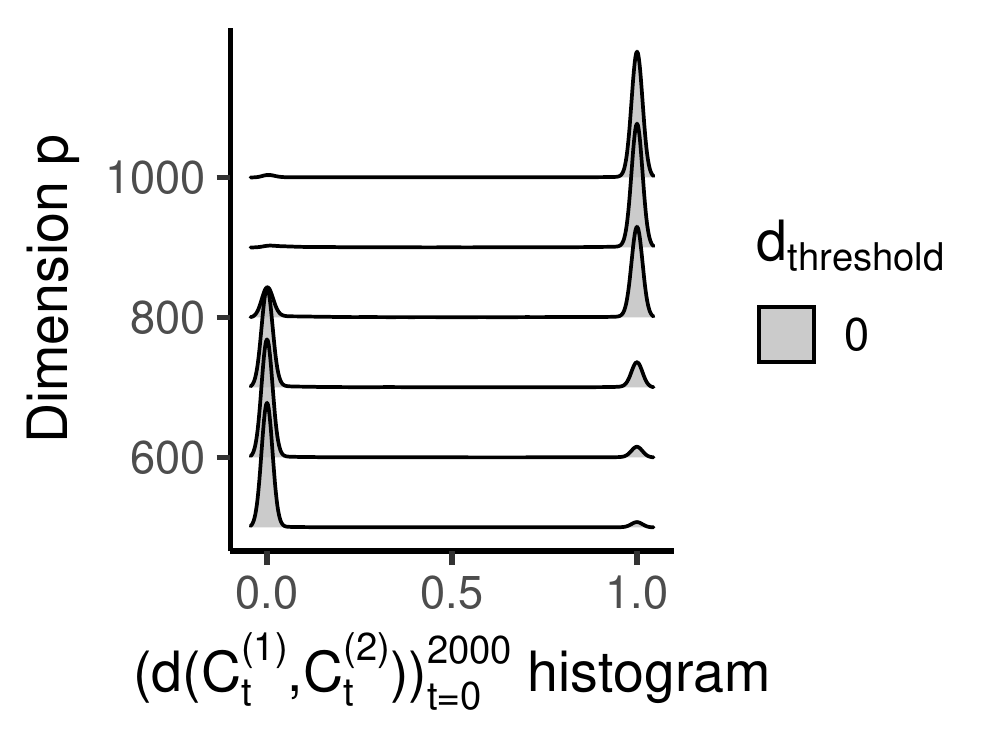}
        %\caption{$(n, s, \sigma_{*})=(100,20,2)$}
        \caption{}
        \label{fig:crn_hist_by_p}
    \end{subfigure}
    \begin{subfigure}[b]{0.32\textwidth}
        \includegraphics[width=\textwidth]{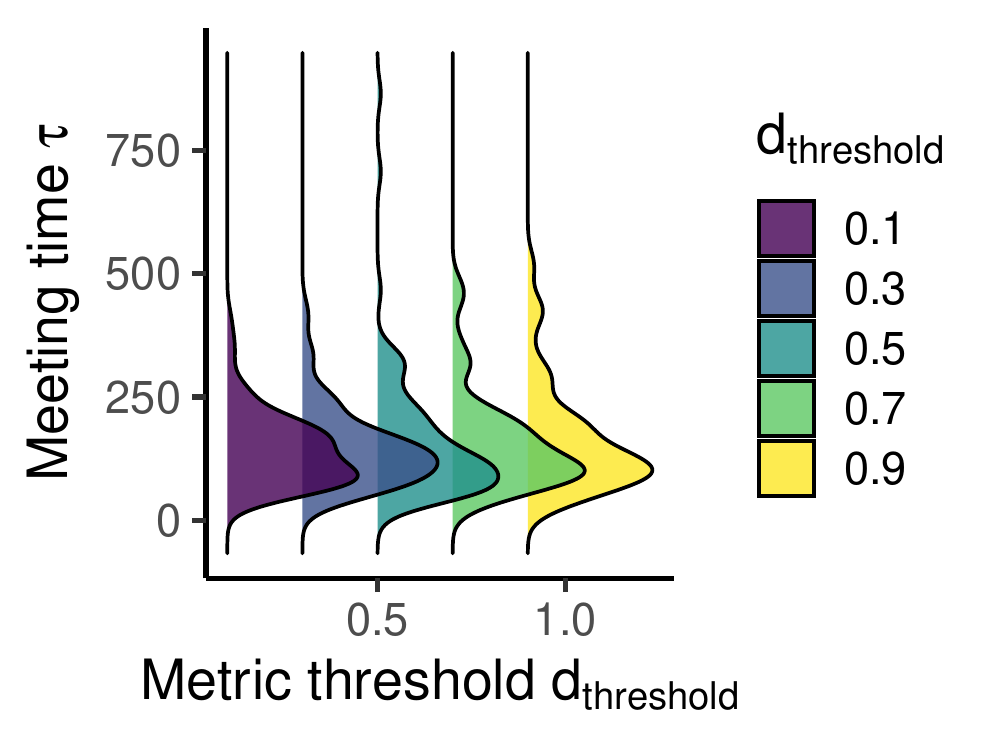}
        %\caption{$(n, p, s, \sigma_{*})=(100,500,20,2)$}
        \caption{}
        \label{fig:p_500_by_threshold}
    \end{subfigure}
    \caption{Two-scale coupling algorithm performance for the Horseshoe 
    prior under different values of $d_0$.
    Unless specified otherwise, $n=100,p=500,s=20$ and $\sigma_*=2$.
    }
    \label{fig:d_threshold_plot}
\end{figure}

\subsection{Calculating metric $d$} \label{appendices:metric_calculation}
We present details of estimating the metric $d$ defined in \eqref{eq:distance_metric}. 
As the coupling in Algorithm \ref{algo:slice_sampling_exact_meeting_coupling} is independent component-wise, we have
\begin{flalign}
d(C_t, \tilde{C}_t) &= 1 - \mathbb{P}\big( \eta_{t+1} = \tilde{\eta}_{t+1} | C_t, \tilde{C}_t \big) = 1 - \prod_{j=1}^p \mathbb{P} \big( \eta_{t+1,j} = \tilde{\eta}_{t+1,j} |  \eta_{t,j}, \tilde{\eta}_{t,j}, m_{t,j}, \tilde{m}_{t,j} \big) 
\end{flalign}
for $m_j = \xi \beta_j^2 / (2 \sigma^2)$
and $\tilde{m}_j = \tilde{\xi} \tilde{\beta}_j^2 / (2 \tilde{\sigma}^2)$.
\subsubsection{An integration based approximation}
We consider the approximation 
\begin{equation}
\mathbb{P} \big( \eta_{t+1,j} \neq \tilde{\eta}_{t+1,j} |  \eta_{t,j}, \tilde{\eta}_{t,j}, m_{t,j}, \tilde{m}_{t,j} \big) \approx \text{TV}(M_{t,j}, \tilde{M}_{t,j}),
\end{equation}
where $M_{t,j}$ and $\tilde{M}_{t,j}$ correspond to target distributions of the 
marginal slice samplers in Algorithm \ref{algo:slice_sampling_exact_meeting_coupling}. 
This corresponds to approximating the probability of not meeting under
Algorithm \ref{algo:slice_sampling_exact_meeting_coupling}
with the total variation distance (i.e. probability of not meeting under a maximal coupling)
between the marginal target distributions $M_{t,j}$ and $\tilde{M}_{t,j}$ of the slice samplers. 
% \textcolor{red}{Perhaps worth commenting in a few sentences above the above $\approx$ sign. What is this approximation?}
We can evaluate $\text{TV}(M_{t,j}, \tilde{M}_{t,j})$ analytically using Proposition \ref{prop:slice_sampler_tv}. 
This motivates the approximation
\begin{flalign} \tilde{d}(C_t, \tilde{C}_t) := 1 - \prod_{j=1}^p \Big( 1 - \text{TV}(M_{t,j}, \tilde{M}_{t,j}) \Big)
\end{flalign}
for $d(C_t, \tilde{C}_t)$. 
\begin{Prop} \label{prop:slice_sampler_tv}
Consider distributions $M$ and $\tilde{M}$ on $[0,\infty)$ with densities
\begin{equation} \label{eq:slice_sampler_target_density}
p(\eta_j | m) \propto \frac{1}{\eta_j^{\frac{1-\nu}{2}} (1+ \nu \eta_j)^{\frac{\nu+1}{2}}} e^{-m \eta_j} \text{ and } 
p(\eta_j | \tilde{m}) \propto \frac{1}{\eta_j^{\frac{1-\nu}{2}} (1+ \nu \eta_j)^{\frac{\nu+1}{2}}} e^{- \tilde{m} \eta_j}, 
\end{equation}
respectively on $[0,\infty)$, such that $M$ and $\tilde{M}$ are parameterized by the fixed constant $m > 0$ and $\tilde{m} > 0$. 
%\textcolor{red}{Pierre: this could be a bit confusing, as the density is read as a function of $\eta_j$
%  in other places. We could emphasize that we are reading this as a function of $m^{(i)}$
% given $\eta_j>0$?}
Then, $\text{TV}(M, \tilde{M}) = 0$ when $m=\tilde{m}$. When $m \neq \tilde{m}$, 
\begin{flalign} \label{eq:slice_sampler_tv}
\text{TV}(M, \tilde{M}) = \bigg| \frac{U_{\frac{1+\nu}{2}, 1, \frac{m}{\nu}}(\nu K)}{U(\frac{1+\nu}{2}, 1, \frac{m}{\nu})} - \frac{U_{\frac{1+\nu}{2}, 1, \frac{\tilde{m}}{\nu}}(\nu K)}{U(\frac{1+\nu}{2}, 1, \frac{\tilde{m}}{\nu})} \bigg|,
\end{flalign}
where $U_{a,b,z}(t) := \frac{1}{\Gamma(a)} \int_0^t x^{a-1}(1+x)^{b-a-1} e^{-zx} dx$ is defined as the \textit{lower incomplete} confluent hypergeometric function of the second kind, such that $U_{a,b,z}(\infty) = U(a,b,z)$ gives the confluent hypergeometric function of the second kind, and 
\begin{equation}
K := \frac{\log U(\frac{1+\nu}{2}, 1, \frac{m}{\nu})  - \log U(\frac{1+\nu}{2}, 1, \frac{\tilde{m}}{\nu}) }{\tilde{m}-m}.
\end{equation}
\end{Prop}
\begin{proof}
We have 
\begin{equation}
p(\eta_j | m) = \frac{1}{\eta_j^{\frac{1-\nu}{2}} (1+ \nu \eta_j)^{\frac{\nu+1}{2}}} e^{-m \eta_j} \frac{1}{Z}
\ \text{ and } \
p(\eta_j | \tilde{m}) = \frac{1}{\eta_j^{\frac{1-\nu}{2}} (1+ \nu \eta_j)^{\frac{\nu+1}{2}}} e^{-\tilde{m} \eta_j} \frac{1}{\tilde{Z}},
\end{equation}
where 
\begin{flalign}
Z &= \int_0^\infty \frac{1}{\eta_j^{\frac{1-\nu}{2}} (1+ \nu \eta_j)^{\frac{\nu+1}{2}}} e^{-m \eta_j} d\eta_j = \nu^{-\frac{\nu+1}{2}} \Gamma \Big( \frac{\nu+1}{2} \Big) U \Big(\frac{\nu+1}{2}, 1, \frac{m}{\nu} \Big), \\
\tilde{Z} &= \int_0^\infty \frac{1}{\eta_j^{\frac{1-\nu}{2}} (1+ \nu \eta_j)^{\frac{\nu+1}{2}}} e^{-\tilde{m} \eta_j} d\eta_j = \nu^{-\frac{\nu+1}{2}} \Gamma \Big( \frac{\nu+1}{2} \Big) U \Big(\frac{\nu+1}{2}, 1, \frac{\tilde{m}}{\nu} \Big).
\end{flalign}
\item \paragraph{Case $m = \tilde{m}$.} $\text{TV}(M, \tilde{M})=0$ is immediate.
\item \paragraph{Case $m < \tilde{m}$.} Note that 
\begin{flalign}
p(\eta_j | m) \leq p(\eta_j | \tilde{m}) &\Leftrightarrow  \exp \big( (\tilde{m}-m) \eta \big) \leq  \frac{Z}{\tilde{Z}} \\
&\Leftrightarrow \eta \leq \frac{\log( U(\frac{1+\nu}{2}, 1, \frac{m}{\nu}) ) - \log( U(\frac{1+\nu}{2}, 1, \frac{\tilde{m}}{\nu}) ) }{\tilde{m}-m} =: K.
\end{flalign}
Therefore, 
\begin{flalign}
1 - \text{TV}(M, \tilde{M}) &= \int_0^\infty \min(p( \eta_j | m), p( \eta_j | \tilde{m})) d \eta_j \\
&= \int_0^K p( \eta_j | m) d \eta_j + \int_K^\infty p( \eta_j | \tilde{m}) d \eta_j \\
&= \frac{1}{Z} \int_0^K \frac{1}{\eta_j^{\frac{1-\nu}{2}} (1+ \nu \eta_j)^{\frac{\nu+1}{2}}} e^{-m \eta_j} d \eta_j \\
& \qquad + \frac{1}{\tilde{Z}} \int_K^\infty \frac{1}{\eta_j^{\frac{1-\nu}{2}} (1+ \nu \eta_j)^{\frac{\nu+1}{2}}} e^{-\tilde{m} \eta_j} d \eta_j \\
&= \frac{U_{\frac{1+\nu}{2}, 1, \frac{m}{\nu}}(\nu K)}{U(\frac{1+\nu}{2}, 1, \frac{m}{\nu})} + 1 - \frac{U_{\frac{1+\nu}{2}, 1, \frac{\tilde{m}}{\nu}}(\nu K)}{U(\frac{1+\nu}{2}, 1, \frac{\tilde{m}}{\nu})} \\
&= 1 - \Bigg| \frac{U_{\frac{1+\nu}{2}, 1, \frac{\tilde{m}}{\nu}}(\nu K)}{U(\frac{1+\nu}{2}, 1, \frac{\tilde{m}}{\nu})} - \frac{U_{\frac{1+\nu}{2}, 1, \frac{m}{\nu}}(\nu K)}{U(\frac{1+\nu}{2}, 1, \frac{m}{\nu})} \Bigg|. 
\end{flalign}
Equation \eqref{eq:slice_sampler_tv} directly follows. 
\item \paragraph{Case $m > \tilde{m}$.} Follows from the $m < \tilde{m}$ case by symmetry.
\end{proof}

\subsubsection{A Rao--Blackwellized estimator}
We have
\begin{flalign}
& \mathbb{P}  \big( \eta_{t+1,j} = \tilde{\eta}_{t+1,j} |  \eta_{t,j}, \tilde{\eta}_{t,j}, m_{t,j}, \tilde{m}_{t,j} \big)\\
&= \mathbb{E} \Big[\mathbb{P} \big( \eta_{t+1,j} = \tilde{\eta}_{t+1,j} \Big| U_{j,*}, \tilde{U}_{j,*}, m_{t,j}, \tilde{m}_{t,j} \big) \big| \eta_{t,j}, \tilde{\eta}_{t,j}, m_{t,j}, \tilde{m}_{t,j} \Big],
\end{flalign}
where $m_{t,j} = \xi_{t} \beta_{t,j}^2 / (2 \sigma_{t}^2)$ and 
$\tilde{m}_{t,j} = \tilde{\xi}_{t} \tilde{\beta}_{t,j}^2 / (2 \tilde{\sigma}_{t}^2)$, 
$(U_{j,*}, \tilde{U}_{j,*}) | m_{t,j}, \tilde{m}_{t,j}$ are sampled using common random numbers as in Algorithm \ref{algo:slice_sampling_exact_meeting_coupling}, 
and $(\eta_{t+1,j}, \tilde{\eta}_{t+1,j}) | U_{j,*}, \tilde{U}_{j,*}, m_{t,j}, \tilde{m}_{t,j}$ are sampled as in Algorithm \ref{algo:slice_sampling_exact_meeting_coupling}. 
We can analytically evaluate probability $\mathbb{P} \big( \eta_{t+1,j} = \tilde{\eta}_{t+1,j} \Big| U_{j,*}, \tilde{U}_{j,*}, m_{t,j}, \tilde{m}_{t,j} \big)$, which motivates the Rao-Blackwellized estimator from \eqref{eq:distance_metric_estimate3}, 
\begin{equation}
\widehat{d_R}^{(2)}(C_t, \tilde{C_t}) := 1 - \prod_{j=1}^p \Big(\frac{1}{R} \sum_{r=1}^R \mathbb{P} \Big( \eta_{t+1,j} = \tilde{\eta}_{t+1,j} \Big| U_{j,r}, \tilde{U}_{j,r}, m_{t,j}, \tilde{m}_{t,j} \Big) \Big),
\end{equation}
where $R$ is the number of samples. For each $r=1,\ldots,R$, $(U_{j,r}, \tilde{U}_{j,r})$ is sampled independently as in Algorithm \ref{algo:slice_sampling_exact_meeting_coupling} and maximal coupling probability $\mathbb{P} \big( \eta_{t+1,j} = \tilde{\eta}_{t+1,j} \big| U_{j,r}, \tilde{U}_{j,r}, m_{t,j}, \tilde{m}_{t,j} \big)$ is calculated analytically using Proposition \ref{prop:distance_metric_estimate3}. 

\begin{Prop} \label{prop:distance_metric_estimate3}
Suppose $\eta_j \big| U_{j,*}, m_{j} \sim P_j$ and $\tilde{\eta}_j \big| \tilde{U}_{j,*}, \tilde{m}_{j} \sim \tilde{P}_j$
where marginal distributions ${P}_j$ and $\tilde{P}_j$ correspond to the distribution in the second step of the slice sampler 
in Algorithm \ref{algo:slice_sampling_1}. Then,
\begin{flalign} \label{eq:couple_prob_componentwise}
\mathbb{P} & \big( \eta_{j} = \tilde{\eta}_{j} \Big| U_{j,*}, \tilde{U}_{j,*}, m_{j}, \tilde{m}_{j} \big) 
 = \left\{\begin{matrix}
 \frac{\gamma_{s}(m_{j} \tilde{K}_{j})}{\gamma_{s}(m_{j} T_{j,*})} + \frac{\gamma_{s}(\tilde{m}_{j} (T_{j,*} \wedge \tilde{T}_{j,*}))-\gamma_{s}(\tilde{m}_{j} \tilde{K}_{j})}{\gamma_{s}(\tilde{m}_{j} \tilde{T}_{j,*})} & for \ m_{j} < \tilde{m}_{j} \\ 
 \frac{\gamma_{s}(m_{j} (T_{j,*} \wedge \tilde{T}_{j,*}))}{\gamma_{s}(m_{j} (T_{j,*} \vee \tilde{T}_{j,*}))} & for \ m_{j} = \tilde{m}_{j} \\
 \frac{\gamma_{s}(\tilde{m}_{j} \tilde{K}_{j})}{\gamma_{s}(\tilde{m}_{j} \tilde{T}_{j,*})} + \frac{\gamma_{s}(m_{j} (T_{j,*} \wedge \tilde{T}_{j,*}))-\gamma_{s}(m_{j} \tilde{K}_{j})}{\gamma_{s}(m_{j} T_{j,*})} & for \ m_{j} > \tilde{m}_{j}
\end{matrix}\right.
\end{flalign}
for $m_{t,j} = \xi_{t} \beta_{t,j}^2 / (2 \sigma_{t}^2)$ and 
$\tilde{m}_{t,j} = \tilde{\xi}_{t} \tilde{\beta}_{t,j}^2 / (2 \tilde{\sigma}_{t}^2)$, 
$T_{j,*} :=(U_{j,*}^{-\frac{2}{1+\nu}}-1)/\nu$ and
$\tilde{T}_{j,*} :=(\tilde{U}_{j,*}^{-\frac{2}{1+\nu}}-1)/\nu$, 
$s = \frac{1 + \nu}{2}$, 
$\gamma_{s}(x) = \frac{1}{\Gamma(s)} \int_0^x t^{s-1} e^{-t} dt \in [0,1]$ the regularized incomplete lower Gamma function, $a \wedge b := \min \{ a, b \} $, $a \vee b := \max \{ a, b \} $, and 
\begin{flalign} \tilde{K}_{j} := \bigg( 0 \vee \Big( \frac{ \log \Big(\Big(\frac{\tilde{m}_{j}}{m_{j}}\Big)^s \frac{\gamma_{s}(m_{j} T_{j,*})}{\gamma_{s}(\tilde{m}_{j} \tilde{T}_{j,*})}\Big)\Big)}{\tilde{m}_{j}-m_{j}} \Big) \bigg) \wedge \Big( T_{j,*} \wedge \tilde{T}_{j,*} \Big). \end{flalign}
\end{Prop}

\begin{proof}
We work component-wise and drop the subscripts for ease of notation. Then distributions 
$P$ and $\tilde{P}$ have densities
\begin{flalign}
p(\eta ; m, T) &= \eta^{s-1} e^{-m \eta} \frac{m^s}{\Gamma(s) \gamma_{s}(m T)} \quad \text{ on } (0, T] \\
p(\eta ; \tilde{m}, \tilde{T}) &= \eta^{s-1} e^{- \tilde{m} \eta} \frac{ \tilde{m}^s}{\Gamma(s) \gamma_{s}( \tilde{m} \tilde{T})} 
\quad \text{ on } (0, \tilde{T}]
\end{flalign}
respectively where $T :=(U_{*}^{-\frac{2}{1+\nu}}-1)/\nu$ and $\tilde{T} :=(\tilde{U}_{*}^{-\frac{2}{1+\nu}}-1)/\nu$, 
$s = \frac{1+\nu}{2}$ and \(\gamma_{s}(x) = \frac{1}{\Gamma(s)} \int_0^x t^{s-1} e^{-t} dt \in [0,1] \) is the regularized incomplete lower Gamma function. Note that 
\begin{flalign*}
\mathbb{P}_{max}\Big( \eta = \tilde{\eta} \Big| T, \tilde{T}, m, \tilde{m} \Big) = \int_0^{T \wedge \tilde{T}} p(\eta ; m, T) \wedge p(\eta ; \tilde{m}, \tilde{T}) d \eta.
\end{flalign*}
\item \paragraph{Case $m = \tilde{m}$. }
\begin{flalign*}
\int_0^{T \wedge \tilde{T}} p(\eta ; m, T) \wedge p(\eta ; \tilde{m}, \tilde{T}) d \eta
&= \int_0^{T \wedge \tilde{T}} \eta^{s-1} \frac{(m)^s}{\Gamma(s)} \Big( \frac{1}{\gamma_s(m T)} \wedge \frac{1}{\gamma_s(m \tilde{T})} \Big) d \eta \\
&= \int_0^{T \wedge \tilde{T}} \eta^{s-1} \frac{(m)^s}{\Gamma(s)} \frac{1}{\gamma_s(m(T \vee \tilde{T}) )} \\
&= \frac{\gamma_s(m(T \wedge \tilde{T}) )}{\gamma_s(m(T \vee \tilde{T}) )}
\end{flalign*}
as required.
\item \paragraph{Case $m < \tilde{m}$. } For $0 < \eta \leq T \wedge \tilde{T}$, 
\begin{flalign*}
p(\eta ; m, T) \leq p(\eta ; \tilde{m}, \tilde{T}) &\Leftrightarrow \exp((\tilde{m}-m)\eta) \leq \Big(\frac{\tilde{m}}{m}\Big)^s \frac{\gamma_{s}(m T)}{\gamma_{s}(\tilde{m} \tilde{T})} 
\Leftrightarrow \eta \leq \frac{ \log \Big(\Big(\frac{\tilde{m}}{m}\Big)^s \frac{\gamma_{s}(m T)}{\gamma_{s}(\tilde{m} \tilde{T})}\Big)\Big)}{\tilde{m}-m} =: K. 
\end{flalign*}
Denote $\tilde{K} := \big( 0 \vee K \big) \wedge \big( T \wedge \tilde{T} \big) $, such that 
\begin{flalign*}
\tilde{K} := \left\{\begin{matrix}
0 & when \ K \leq 0 \\ 
K &  when \ 0 < K \leq T \wedge \tilde{T} \\
 T \wedge \tilde{T} & when \  T \wedge \tilde{T} < K.
\end{matrix}\right.
\end{flalign*}
This gives
\begin{flalign*}
\int_0^{T \wedge \tilde{T}} p(\eta ; m, T) \wedge p(\eta ; \tilde{m}, \tilde{T}) d \eta
&= \int_0^{\tilde{K}} p(\eta ; m, T) d \eta + \int_{\tilde{K}}^{T \wedge \tilde{T}} p(\eta ; \tilde{m}, \tilde{T}) d \eta \\
&= \Big( \frac{\Gamma(s) \gamma_{s}(m \tilde{K})}{(m)^{s}} \Big) \Big( \frac{(m)^{s}}{\Gamma(s) \gamma_{s}(m T)}\Big) \\
& \qquad \qquad + \Big( \frac{\Gamma(s) \big(\gamma_{s}(\tilde{m} (T \wedge \tilde{T}))-\gamma_{s}(\tilde{m} K)\big)}{(\tilde{m})^{s}} \Big) \Big( \frac{(\tilde{m})^{s}}{\Gamma(s) \gamma_{s}(\tilde{m} \tilde{T}))}\Big) \\
&= \frac{\gamma_{s}(m \tilde{K})}{\gamma_{s}(m T)} + \frac{\gamma_{s}(\tilde{m} (T \wedge \tilde{T}))-\gamma_{s}(\tilde{m} \tilde{K})}{\gamma_{s}(\tilde{m} \tilde{T})}.
\end{flalign*}
as required. 
\item \paragraph{Case $m > \tilde{m}$. } Follows from the $m < \tilde{m}$ case by symmetry. 
\end{proof}

\section{Proofs} \label{appendices:proofs}
\subsection{Marginal $\beta$ prior and posterior densities}
%\textcolor{blue}{P: this proof could perhaps be shortened. Mostly it's about noticing that the integral over $\eta_j$ has a known, lovely name (confluent hypergeometric function of the second kind).} 

\begin{proof}[Proof of Proposition \ref{prop:marginal_beta_posterior}]
\item \paragraph{Marginal Prior on $\beta$. }
Let $\pi_{\eta}$ denote the component-wise prior density of each $\eta_j$ based on the Half-$t(\nu)$ distribution, and
$\mathcal{N}( \cdot; \mu, \Sigma )$ the density of a multivariate Normal distribution with mean $\mu$ and $\Sigma$.
It suffices to calculate the marginal prior for each $\beta_j$ given $\xi, \sigma^2$, by marginalizing over each $\eta_j$.
We obtain,
\begin{flalign}
\pi(\beta_j | \xi, \sigma^2) &= \int_{0}^\infty \mathcal{N} \Big( \beta_j; 0, \frac{\sigma^2}{\xi \eta_j} \Big) \pi_{\eta}(\eta_j) d\eta_j \\
&\propto \int_{0}^\infty \frac{ \sqrt{\xi \eta_j} }{\sqrt{2 \pi \sigma^2}} \exp \Big( - \frac{\xi \beta_j^2}{2\sigma^2} \eta_j \Big) \frac{1}{\eta_j^{\frac{2-\nu}{2}}(1+\nu \eta_j)^{\frac{\nu + 1}{2}}} d\eta_j \\
&\propto \int_{0}^\infty  \frac{1}{\eta_j^{\frac{1-\nu}{2}}(1+\nu \eta_j)^{\frac{\nu + 1}{2}}} \exp \Big( - \frac{\xi \beta_j^2}{2\sigma^2} \eta_j \Big)  d\eta_j \\
&= \int_0^\infty \frac{1}{\big(\frac{t}{\nu}\big)^{\frac{1-\nu}{2}}(1+t)^{\frac{\nu + 1}{2}}} \exp \Big( -\frac{\xi \beta_j^2}{2\sigma^2 \nu} t\Big) dt \frac{1}{\nu} \\
&= \frac{\Gamma(\frac{\nu + 1}{2})}{\nu^{\frac{\nu + 1}{2}}} U \Big( \frac{\nu+1}{2}, 1, \frac{\xi \beta_j^2}{2\sigma^2 \nu} \Big) \\
&\propto U \Big( \frac{\nu+1}{2}, 1, \frac{\xi \beta_j^2}{2\sigma^2 \nu} \Big) \label{eq:prop1proofeq1},
\end{flalign}
where \eqref{eq:prop1proofeq1} follows from the definition of the confluent hypergeometric function of the second kind: $U(a,b,z) := \frac{1}{\Gamma(a)} \int_0^\infty x^{a-1}(1+x)^{b-a-1}e^{-zx}dx$ 
for any $a, b, z>0$. By the asymptotic expansion of the confluent hypergeometric function of the second kind \citep[Section 13.1]{abramowitz1964handbook}, we obtain
\begin{align} \label{eq:prop1proofeqmarginal_beta_prior}
\pi(\beta_j | \xi, \sigma^2) \propto U \Big( \frac{\nu+1}{2}, 1, \frac{\xi \beta_j^2}{2\sigma^2 \nu} \Big) \asymp \begin{cases}
-\frac{1}{\Gamma(\frac{1+\nu}{2})} \log \Big( \frac{\xi \beta_j^2}{2 \sigma^2 \nu} \Big) & \asymp - \log (|\beta_j|) \text{ for } |\beta_j| \rightarrow 0 \\
\Big( \frac{\xi \beta_j^2}{2 \sigma^2 \nu} \Big)^{-\frac{1+\nu}{2}} & \asymp  |\beta_j|^{-(1+\nu)} \text{ for } |\beta_j| \rightarrow +\infty 
\end{cases}
\end{align}
as required for \eqref{eq:marginal_beta_prior}. 
\item \paragraph{Marginal Posterior of $\beta$.}
Let $\pi_\xi$ denote the prior density on $\xi$ and $\pi_{\sigma^2}$ the $\InvGamma(\frac{a_0}{2}, \frac{b_0}{2})$ prior density on $\sigma^2$. Then,
\begin{flalign}
\pi(\beta | \xi, \sigma^2, y) \propto \mathcal{N}( y; X\beta, \sigma^2 I_n ) \Big( \prod_{j=1}^p \pi( \beta_j | \xi, \sigma^2 ) \Big) \propto \mathcal{N}( y; X\beta, \sigma^2 I_n ) \prod_{j=1}^p U \Big( \frac{1+\nu}{2}, 1, \frac{\xi \beta^2_j}{2 \sigma^2 \nu} \Big) \label{eq:prop1proofmarginal_beta_fixed_xi_sigma2}
\end{flalign}
as required for \eqref{eq:marginal_beta_fixed_xi_sigma2}. 
%\textcolor{red}{Pierre: it seems to me that, in the third argument of the $U$ function, 
%$\frac{\xi \beta^2_j}{2 \sigma^2}$ should be further divided by $\nu$. Can this be checked thoroughly? I believe it occurs in various places. }
%\textcolor{red}{NB: Corrected. Remove comment.}
Equations \eqref{eq:prop1proofeqmarginal_beta_prior} 
and \eqref{eq:prop1proofmarginal_beta_fixed_xi_sigma2} directly give 
\begin{equation}
\pi (\beta | \xi, \sigma^2, y) \asymp \prod_{j=1}^p U \Big( \frac{\nu+1}{2}, 1, \frac{\xi \beta_j^2}{2\sigma^2 \nu} \Big)  \asymp - \prod_{j=1}^p \log | \beta_j | \text{ for } \ \|\beta \| \rightarrow 0.
\end{equation}
\item \paragraph{Gradient of negative log-density of posterior of $\beta$.} 
From \eqref{eq:prop1proofmarginal_beta_fixed_xi_sigma2}, we obtain
\begin{flalign}
-\frac{\partial}{\partial \beta_j} \Big( \log \pi(\beta | \xi, \sigma^2, y) \Big) &= 
-\frac{\partial}{\partial \beta_j} \bigg( - \frac{1}{2 \sigma^2} \|y-X \beta \|_2^2 + \sum_{j=1}^p \log U \Big( \frac{1+\nu}{2}, 1, \frac{\xi \beta^2_j}{2 \sigma^2 \nu} \Big)  \bigg) \\
&= -\Big[\frac{1}{\sigma^2} X^T (y-X \beta)\Big]_j -  \frac{ \frac{d}{d \beta_j} U \Big( \frac{1+\nu}{2}, 1, \frac{\xi \beta^2_j}{2 \sigma^2 \nu} \Big)}{U \Big( \frac{1+\nu}{2}, 1, \frac{\xi \beta^2_j}{2 \sigma^2 \nu} \Big)}. \label{eq:prop1gradnegativelogdensity}
\end{flalign}
Recall  \citep[Section 13.4]{abramowitz1964handbook} for $a,b > 0$, 
\begin{align}
\frac{d}{dz} U(a,b,z) \overset{ \forall z > 0}{=} -a U(a+1,b+1,z), \ 
U(a,2,z) \overset{z \rightarrow 0^+}{\asymp} \frac{1}{\Gamma(a)} \frac{1}{z}, \
U(a,1,z) \overset{z \rightarrow 0^+}{\asymp} -\frac{1}{\Gamma(a)} \log(z).
\end{align}
This gives 
\begin{flalign}
\frac{ \frac{d}{d \beta_j} U \Big( \frac{1+\nu}{2}, 1, \frac{\xi \beta^2_j}{2 \sigma^2 \nu} \Big)}{U \Big( \frac{1+\nu}{2}, 1, \frac{\xi \beta^2_j}{2 \sigma^2 \nu} \Big)} &= 
- \frac{1+\nu}{2} \frac{U \Big( \frac{1+\nu}{2}+1, 1+1, \frac{\xi \beta^2_j}{2 \sigma^2 \nu} \Big)}{U \Big( \frac{1+\nu}{2}, 1, \frac{\xi \beta^2_j}{2 \sigma^2 \nu} \Big)} \frac{\xi \beta_j}{ \sigma^2 \nu} \overset{\beta_j \rightarrow 0}{\asymp} \frac{1}{ \beta_j \log | \beta_j |}.
\end{flalign}
Therefore by \eqref{eq:prop1gradnegativelogdensity}, $-\frac{\partial}{\partial \beta_j} \Big( \log \pi(\beta | \xi, \sigma^2, y) \Big) \overset{\beta_j \rightarrow 0}{\asymp} - \frac{1}{ \beta_j \log | \beta_j |}$ as required. 
\end{proof}

\begin{proof}[Proof of Corollary \ref{cor:marginal_beta_posterior_heavy_tails}]
By Proposition \ref{prop:marginal_beta_posterior}, 
\begin{flalign}
\pi (\beta | \xi, \sigma^2, y) = \frac{1}{Z_{\xi, \sigma^2, \nu}} \mathcal{N}( y; X\beta, \sigma^2 I_n ) \prod_{j=1}^p U \Big( \frac{1+\nu}{2}, 1, \frac{\xi \beta^2_j}{2 \sigma^2 \nu} \Big),
\end{flalign}
where 
$Z_{\xi, \sigma^2, \nu}:= \int_{\mathbb{R}^p} \mathcal{N}( y; X\beta, \sigma^2 I_n ) \prod_{j=1}^p U \Big( \frac{1+\nu}{2}, 1, \frac{\xi \beta^2_j}{2 \sigma^2 \nu} \Big) d \beta$
is the normalizing constant. As $p >n$, there exists some $e^{\perp} \in \mathbb{R}^p$ such that $Xe^{\perp}=0$. 
When $e^{\perp}_j=0$ for some $j=1,\ldots,p$, the posterior density is infinite
by Proposition \ref{prop:marginal_beta_posterior}. 
Consider $e^{\perp} \in \mathbb{R}^p$ 
such that each $e^{\perp}_j \in \mathbb{R} $ is non-zero for $j=1,\ldots,p$. 
For any such fixed such $e^{\perp}$ and large $\lambda>0$, we obtain
\begin{flalign}
\pi(\lambda e^{\perp} | \xi, \sigma^2, y) &= \frac{1}{Z_{\xi, \sigma^2, \nu}} \mathcal{N}( y; 0, \sigma^2 I_n ) \prod_{j=1}^p U \Big( \frac{1+\nu}{2}, 1, \frac{\xi \lambda^2 (e^{\perp}_j)^2}{2 \sigma^2 \nu} \Big) \\
&\propto \prod_{j=1}^p U \Big( \frac{1+\nu}{2}, 1, \frac{\xi \lambda^2 (e^{\perp}_j)^2}{2 \sigma^2 \nu} \Big) \\
&\asymp \prod_{j=1}^p \Big( \frac{\xi \lambda^2 (e^{\perp}_j)^2}{2 \sigma^2 \nu} \Big)^{-\frac{1+\nu}{2}} \label{eq:prop1proofeq2} \\
&\asymp \lambda^{-p(1+\nu)}, 
\end{flalign}
where \eqref{eq:prop1proofeq2} follows from the asymptotic expansion of the 
confluent hypergeometric function of the second kind \citep[Section 13.1]{abramowitz1964handbook}. 
\end{proof}

\subsection{Geometric Ergodicity}
\paragraph{Initial technical results.} 
We first record some technical results. 
\begin{Lemma} \label{lemma:matrix_row_bound} Consider any matrix $A$ such that $A=\begin{bmatrix} B \\ C \end{bmatrix}$
for submatrices $B$ and $C$. Then $\|A\|_2 \leq \|B\|_2 + \|C\|_2$, where $\| \cdot \|_2$ is the spectral norm of a matrix. 
\end{Lemma}
\begin{proof}
Take any vector $v$ of length equal to the number of columns of $A$ with $\| v \|_2 =1$. Then
\begin{flalign}
\| Av \|_2^2 = \Big\| \begin{bmatrix} B \\ C \end{bmatrix} v \Big\|_2^2 = \| Bv \|_2^2 + \| Cv \|_2^2 \leq \|B\|_2^2 + \|C\|_2^2.
\end{flalign}
Taking the supremum over all such $v$ with $\| v \|_2 =1$ gives
$\|A\|_2^2 \leq \|B\|_2^2 + \|C\|_2^2 \leq (\|B\|_2 + \|C\|_2)^{2}$.
\end{proof}

\begin{Lemma} \label{lemma:woodbury} (Sherman--Morrison--Woodbury formula \citep{Hager1989updatingSIAM})
For any $p\times p$ invertible matrix $A$, $p\times n$ matrix $U$ 
and $n\times p$ matrix $V$, 
\begin{equation} \label{eq:woodbury}
(A+UV)^{-1}=A^{-1} - A^{-1}U(I_n+V A^{-1} U )^{-1} V A^{-1}.
\end{equation}
When $U=V^T=x \in \mathbb{R}^p$, we obtain 
$(A+xx^T)^{-1}=A^{-1} - A^{-1} x  x^T A^{-1}/(1+x^T A^{-1} x)$.
\end{Lemma}

\begin{Lemma}  \label{lemma:uniform_bound_mean} 
(A uniform bound on the full conditional mean of $\beta$).
Let $X$ be an $n\times p$ matrix with columns $x_i \in \mathbb{R}^n$ for $i=1, \ldots, p$, 
and let $R=\text{Diag}(r_1,\ldots,r_p)$ be a $p\times p$ diagonal matrix with positive entries. Then
\begin{equation}
\sup_{r_1, \ldots, r_p>0} \|(X^TX + R)^{-1} X^T\|_2 \leq  C_p,
\end{equation}
for some $C_p < \infty$, where $\| \cdot \|_2$ is the spectral norm of a matrix. 
That is, the spectral norm of $(X^TX + R)^{-1} X^T$ is uniformly bounded over all positive diagonal matrices $R$.
\end{Lemma}

In an earlier version of the present manuscript, we had explicitly
characterized the uniform upper bound as $\| X^\dagger y \|_2$ for
$X^{\dagger}$ the Moore-Penrose pseudoinverse of $X$, which is not true in
general. Existence of an uniform upper bound has been also established by
\citet{bhattacharya2021geo} using block matrix inversions and induction.

\begin{proof}
Fix any $n \geq 1$. We first consider the case $p=1$, such that $X=x_1 \in \mathbb{R}^n$, $R=r_1$ for some scalar $r_1>0$. We obtain
\begin{equation} \label{eq:one_dim_bound}
\| (X^TX + R)^{-1} X^T \|_2 = \frac{1}{r_1+x_1^T x_1} \| x_1^T \|_2 = \frac{\| x_1 \|_2}{r_1+ \|x_1\|_2^2} \leq 
\frac{\mathrm{I}_{x_1 \neq 0}}{\|x\|_2},
\end{equation}
which is an uniform bound for all $r_1>0$. 

We now consider $p > 1$. We have $X=[x_1 \ldots x_p]$ with $x_i \in \mathbb{R}^n$ for $i=1, \ldots, p$ and 
$R=\text{Diag}(r_1,\ldots,r_p)$ with $r_i >0$ for $i=1, \ldots, p$. For any
$p\times p$ permutation matrix $P$, the matrices $XP^T$ and $PRP^T$ correspond to a reordering of the columns of $X$ and the diagonal entries of $R$. We
have
\begin{equation}
\| \big( (XP^T)^T(XP^T) + PRP^T \big)^{-1} (XP^T)^T\|_2 = \| P(X^TX + R)^{-1}X^T \|_2 = \| (X^TX + R)^{-1}X^T \|_2.
\end{equation}
This implies that $\| (X^TX + R)^{-1}X^T \|_2$ is invariant under any
reordering of the $p$ columns of $X$
and the $p$ diagonal entries of $R$. By considering such reordering, we can assume the diagonal entries of 
$R$ are non-decreasing without loss of generality; that is $R=\text{Diag}(r_1,\ldots,r_p)$ with $0<r_1\leq \ldots\leq r_p$.
Denote $A_p :=(X^TX + R)^{-1}X^T$. Then,
\begin{flalign} 
A_p &= \big( R^{-1} - R^{-1} X^T ( I_n + X R^{-1} X^T )^{-1} X R^{-1} \big) X^T \ \text{by Lemma \ref{lemma:woodbury}} \\
&= R^{-1} X^T - R^{-1} X^T ( I_n - \big( I_n + X R^{-1} X^T )^{-1} \big) \\
&= R^{-1} X^T (I_n + X R^{-1} X^T)^{-1}. \label{eq:matrix_A1}
\end{flalign}
%\begin{equation} \label{eq:matrix_A1}
%A_p = (X^TX + R)^{-1}X^T = R^{-1} X^T (I_n + X R^{-1} X^T)^{-1}.
%\end{equation}
It therefore suffices to show that
\begin{equation} \label{eq:Cp_bound}
C_{p} := \sup_{0<r_1 \leq \ldots \leq r_{p}} \| A_p \|_2 = \sup_{0<r_1 \leq \ldots \leq r_{p}} \|{R}^{-1} {X}^T (I_n+ {X} {R}^{-1} {X}^T)^{-1}\|_2
\end{equation}
is finite for any $p > 1$. For $k=1, \ldots, p-1$, define 
\begin{equation} \label{eq:Ck_bound}
C_{k} := \sup_{0<r_1 \leq \ldots \leq r_{k}} \|A_{k}\|_2 = \sup_{0<r_1 \leq \ldots \leq r_{k}} \|{R}_k^{-1} {X}_k^T (I_n+ {X}_k {R}_k^{-1} {X}_k^T)^{-1}\|_2,
\end{equation}
where $X_k=[x_1 \ldots x_k]$ is the $n\times k$ matrix corresponding to the first $k$ columns of $X$
and ${R}_k = \text{Diag}(r_1, \ldots, r_k)$ with $0<r_1\leq \ldots \leq r_k$ 
is the $k\times k$ diagonal matrix corresponding to the top $k$ diagonal entries of $R$, 
and $A_{k} := R_{k}^{-1} X_{k}^T (I_n + X_{k} R_{k}^{-1} X_{k}^T)^{-1}$. 
Then $X_{p-1} R_{p-1}^{-1} X_{p-1}^T = \sum_{i=1}^{p-1} x_i x_i^T/r_i$.
By Lemma \ref{lemma:woodbury} with $U=V^T=x_{p} / r_p^{1/2}$, 
%By the Sherman–Morrison formula (a special case of Woodbury's matrix identity),
\begin{flalign}
(I_n+ X R^{-1} X^T)^{-1} &= \Big( (I_n+ X_{p-1} R_{p-1}^{-1} X_{p-1}^T) + \frac{x_p x_p^T}{r_p} \Big)^{-1} \\
&=(I_n+ X_{p-1} R_{p-1}^{-1} X_{p-1}^T)^{-1} \\
&\quad - \frac{(I_n+ X_{p-1} R_{p-1}^{-1} X_{p-1}^T)^{-1} x_{p} x_{p}^T (I_n+ X_{p-1} R_{p-1}^{-1} X_{p-1}^T)^{-1} }
{r_{p} + x_{p}^T (I_n+ X_{p-1} R_{p-1}^{-1} X_{p-1}^T)^{-1} x_{p}}.
\end{flalign}

Substituting this decomposition into \eqref{eq:matrix_A1}, we obtain
\begin{flalign} \label{eq:matrix_A}
A_p = \begin{bmatrix}
R_{p-1}^{-1} X_{p-1}^T \\
r_{p}^{-1} x_{p}^T
\end{bmatrix} (I_n+ X R^{-1} X^T)^{-1} =
\begin{bmatrix}
V_{p-1} \\
v_p^T
\end{bmatrix}
\end{flalign} 
where, after some simplifications, $v_p \in \mathbb{R}^n$ and $V_{p-1}$ is a $p-1$ by $n$ matrix such that
\begin{flalign}
v_p^{T} &:= r_{p}^{-1} x_{p}^T(I_n+ X R^{-1} X^T)^{-1} = 
\frac{x_{p}^T (I_n+ X_{p-1} R_{p-1}^{-1} X_{p-1}^T)^{-1}}{ r_{p} + x_{p}^T (I_n+ X_{p-1} R_{p-1}^{-1} X_{p-1}^T)^{-1} x_{p} },
\label{eq:matrix_part_one} \\
V_{p-1} &:= R_{p-1}^{-1} X_{p-1}^T (I_n+ X R^{-1} X^T)^{-1} = A_{p-1} \Big( I_n- x_{p} v_p^{T} \Big). \label{eq:matrix_part_two}
\end{flalign}

We have $\|V_{p-1}\|_2 \leq C_{p-1} (1 + \|x_p\|_2 \|v_p^{T}\|_2 )$ by \eqref{eq:Ck_bound}, subadditivity
and submultiplicativity. Then by Lemma \ref{lemma:matrix_row_bound}, 
\begin{flalign} 
\| A_p \|_2 &\leq \|V_{p-1}\|_2 + \|v_p^{T}\|_2 \leq C_{p-1} (1 + \|x_p\|_2 \|v_p^{T}\|_2 ) + \|v_p^{T}\|_2. \label{eq:lemma_apply}
\end{flalign}

Therefore to show $C_p$ in \eqref{eq:Cp_bound} is finite, it suffices to bound $\|v_p^{T}\|_2$ uniformly over all $R=\text{Diag}(r_1,\ldots,r_p)$ with 
$0<r_1\leq \ldots\leq r_p$. Note that $x_{p} = X_{p-1} a_{p} + b_{p}$ for some $a_{p} \in \mathbb{R}^{p-1}$ and some
$b_{p} \in \mathbb{R}^n$ such that $X_{p-1}^T b_{p} = 0$. Here
$X_{p-1} a_{p}$ corresponds to the projection of $x_{p}$ to the column space of $X_{p-1}$, and 
$b_{p}$ corresponds to the component of $x_{p}$ that is perpendicular to the column space of $X_{p-1}$. 
Note that $b_{p} = (I_n+ X_{p-1} R_{p-1}^{-1} X_{p-1}^T)b_{p}$. This implies
\begin{flalign}
 (I_n+ X_{p-1} R_{p-1}^{-1} X_{p-1}^T)^{-1} b_{p} &= b_{p}, \\  
b_{p}^T (I_n+ X_{p-1} R_{p-1}^{-1} X_{p-1}^T)^{-1} X_{p-1} a_{p}  &= b_{p}^T X_{p-1} a_{p} = 0 .
%\\ x_p^T (I_n+ X_{p-1} R_{p-1}^{-1} X_{p-1}^T)^{-1} x_p &= 
%a_{p}^T X_{p-1}^T  (I_n+ X_{p-1} R_{p-1}^{-1} X_{p-1} a_{p} + b_{p}^T b_{p}.
\end{flalign}

Therefore by \eqref{eq:matrix_part_one}, 
\begin{flalign}
\| v_p^{T} \|_2 &= \| \frac{a_{p}^T X_{p-1}^T (I_n+ X_{p-1} R_{p-1}^{-1} X_{p-1}^T)^{-1} + b_{p}^T}{ r_{p} + a_{p}^T X_{p-1}^T (I_n+ X_{p-1} R_{p-1}^{-1} X_{p-1}^T)^{-1} X_{p-1} a_{p} + b_{p}^T b_{p} } \|_2 \\
&\leq \big\| \frac{a_{p}^T X_{p-1}^T (I_n+ X_{p-1} R_{p-1}^{-1} X_{p-1}^T)^{-1}}{ r_{p} + a_{p}^T X_{p-1}^T (I_n+ X_{p-1} R_{p-1}^{-1} X_{p-1}^T)^{-1} X_{p-1} a_{p}} \big\|_2 +
\big\| \frac{b_{p}^T}{ r_{p} + b_{p}^T b_{p}} \big\|_2 \ \text{ by subadditivity} \\
&\leq \| \frac{a_{p}^T X_{p-1}^T (I_n+ X_{p-1} R_{p-1}^{-1} X_{p-1}^T)^{-1}}{ r_{p} } \|_2 + \frac{\mathrm{I}_{b_{p} \neq 0}}{\|b_{p}\|_2} \ \text{ as } (I_n+ X_{p-1} R_{p-1}^{-1} X_{p-1}^T)^{-1} \text{ is p.s.d.} \\
&= \| \frac{(R_{p-1} a_{p})^T R_{p-1}^{-1} X_{p-1}^T(I_n+ X_{p-1} R_{p-1}^{-1} X_{p-1}^T)^{-1}}{ r_{p} } \|_2 + \frac{\mathrm{I}_{b_{p} \neq 0}}{\|b_{p}\|_2} \\
%&= \| \frac{(R_{p-1} a_{p})^T A_{p-1}}{ r_{p} } \|_2 + \frac{\mathrm{I}_{b_{p} \neq 0}}{\|b_{p}\|_2} \\
&\leq \big\| \frac{R_{p-1}}{r_{p}} \big\|_2 \| a_{p} \|_2 C_{p-1}
+ \frac{\mathrm{I}_{b_{p} \neq 0}}{\|b_{p}\|_2} \ \text{ by submultiplicativity and \eqref{eq:Ck_bound}} \\
& \leq \big\| a_{p} \big\|_2 C_{p-1} + \frac{\mathrm{I}_{b_{p} \neq 0}}{\|b_{p}\|_2}
\ \text{ as $r_{p} \geq r_j$ for all $j = 1,\ldots,p-1$}. \label{eq:matrix_part_one_proof}
\end{flalign}

Note that $x_p, a_p$ and $b_p$  do not depend on entries of the diagonal matrix $R$. 
Taking the supremum over all $R=\text{Diag}(r_1,\ldots,r_p)$ with 
$0<r_1\leq \ldots\leq r_p$ in \eqref{eq:lemma_apply}, by \eqref{eq:matrix_part_one_proof} we obtain the recurrence relation 
\begin{flalign} \label{eq:rec_relation}
C_p \leq C_{p-1} \Big(1 + \|x_p\|_2 \big( \big\| a_{p} \big\|_2 C_{p-1} + \frac{\mathrm{I}_{b_{p} \neq 0}}{\|b_{p}\|_2} \big) \Big) + \Big( \big\| a_{p} \big\|_2 C_{p-1} + \frac{\mathrm{I}_{b_{p} \neq 0}}{\|b_{p}\|_2} \Big)
\end{flalign}
for all $p \geq 1$. As $C_1$ is finite by \eqref{eq:one_dim_bound},
\eqref{eq:rec_relation} implies that $C_{p}$ is finite for all $p$.  
\end{proof}

\begin{Lemma}(Moments of the full conditionals of each $\eta_j$). \label{lemma:eta_moments}
For fixed constants $m>0$ and $\nu>0$, define a distribution on $(0, \infty)$ 
with probability density function 
\begin{equation} \label{eq:eta_full_conditional_pdf}
p(\eta | m) \propto \frac{1}{\eta^{\frac{1-\nu}{2}} (1+ \nu \eta)^{\frac{\nu+1}{2}}} e^{-m \eta}.
\end{equation} 
Then for any $c > \max \{-\frac{\nu+1}{2}, -1 \}$ and $\eta \sim p(\cdot | m) $, we have 
\begin{flalign}
\mathbb{E}[\eta^{c} | m] =\frac{1}{\nu^{c}}  \frac{\Gamma(\frac{1+\nu}{2}+c)}{ \Gamma(\frac{1+\nu}{2})} \frac{U(\frac{1+\nu}{2}+c, 1+c, \frac{m}{\nu})}{U(\frac{1+\nu}{2}, 1, \frac{m}{\nu})}.
\end{flalign}
\end{Lemma}
\begin{proof}
Recall the definition of the confluent hypergeometric function of the second kind $U(a,b,z) := \frac{1}{\Gamma(a)} \int_0^\infty x^{a-1}(1+x)^{b-a-1} e^{-zx} dx $. We obtain
\begin{flalign}
\int_0^\infty \frac{\eta^{c}}{\eta^{\frac{1-\nu}{2}} (1+ \nu \eta)^{\frac{\nu+1}{2}}} e^{-m \eta} d \eta &= \nu^{-\frac{\nu+1}{2}-c} \int_0^\infty x^{\frac{\nu+1}{2}+c-1}(1+x)^{-\frac{\nu+1}{2}} e^{- \frac{m}{\nu} x} dx \\
&= \nu^{-\frac{\nu+1}{2}-c} \Gamma \Big(\frac{\nu+1}{2}+c \Big) U \Big(\frac{\nu+1}{2}+c, 1+c, \frac{m}{\nu} \Big).
\end{flalign}
Therefore, 
\begin{flalign}
\mathbb{E}[\eta^{c} | m] &=  \frac{\nu^{-\frac{\nu+1}{2}-c} \Gamma \big(\frac{\nu+1}{2}+c \big) U \big(\frac{\nu+1}{2}+c, 1+c, \frac{m}{\nu} \big)}{\nu^{-\frac{\nu+1}{2}} \Gamma \big(\frac{\nu+1}{2} \big) U \big(\frac{\nu+1}{2}, 1, \frac{m}{\nu} \big)} \\
&= \frac{1}{\nu^c} \frac{\Gamma \big(\frac{\nu+1}{2}+c \big) U \big(\frac{\nu+1}{2}+c, 1+c, \frac{m}{\nu} \big)}{\Gamma \big(\frac{\nu+1}{2} \big) U \big(\frac{\nu+1}{2}, 1, \frac{m}{\nu} \big)}.
\end{flalign}
\end{proof}

Ratios of confluent hypergeometric function evaluations, such as the ones
arising in Lemma \ref{lemma:eta_moments},
can be upper-bounded using the following result.

\begin{Lemma} \label{lemma:confluent_ratio_bound}
Let $U$ denote the confluent hypergeometric function of the second kind, and fix $r>0$. 
\begin{enumerate}
\item Let $c>0$. Then 
\begin{equation} \label{eq:confluent_ratio_bound1}
\frac{U \big(r+c, 1+c, x \big)}{ U \big(r, 1, x \big)} < x^{-c} \quad \text{for all } x>0.
\end{equation}
Furthermore, for every $\epsilon > 0$, there exists a $K^{(1)}_{\epsilon, c, r} < \infty$ such that 
\begin{equation} \label{eq:confluent_ratio_bound2}
\frac{U \big(r+c, 1+c, x \big)}{ U \big(r, 1, x \big)} < \epsilon x^{-c} + K^{(1)}_{\epsilon, c, r} \quad \text{for all } x>0.
\end{equation}
\item Let $\max \{ -1, -r \} <c<0$. Then there exists some $K^{(2)}_{c, r} < \infty$ such that 
\begin{equation} \label{eq:confluent_ratio_bound3}
\frac{ U \big(r+c, 1+c, x \big)}{U \big(r, 1, x \big)} < x^{-c} + K^{(2)}_{c, r} \quad \text{for all } x>0.
\end{equation}
\end{enumerate}
\end{Lemma}
\begin{proof} \item \paragraph{Case: $c>0$.} 
By definition, we have 
\begin{flalign*}
x^c \frac{ \Gamma(r+c) U \big(r+c, 1+c, x \big)}{ \Gamma(r) U \big(r, 1, x \big)}  &= \frac{x^c \int_0^\infty t^{r+c-1} (1+t)^{-r} e^{-xt} dt}{\int_0^\infty t^{r-1} (1+t)^{-r} e^{-xt} dt} \\
&= 
\frac{\int_0^\infty s^{r+c-1} (x+s)^{-r} e^{-s} ds}{\int_0^\infty s^{r-1} (x+s)^{-r} e^{-s} ds} \text{ for } s=xt \\
&= \frac{\Gamma(r+c)}{\Gamma(r)} \frac{\mathbb{E}[ (x+S)^{-r} ]}{\mathbb{E}[ (x+\tilde{S})^{-r} ]} \text{ for } S \sim \mathrm{Gamma}(r+c, 1) \text{ and } \tilde{S} \sim \mathrm{Gamma}(r, 1) \\
&< \frac{\Gamma(r+c)}{\Gamma(r)}
\end{flalign*}
where the final equality follows as the distribution $\Gamma(r+c, 1)$ stochastically dominates the distribution $\Gamma(r, 1)$, and hence 
$\mathbb{E}[ (x+S)^{-r} ] < \mathbb{E}[ (x+\tilde{S})^{-r} ]$ for all $x>0$ as $r>0$ and $c>0$. This gives $U \big(r+c, 1+c, x \big) / U \big(r, 1, x \big) < x^{c}$ as required for \eqref{eq:confluent_ratio_bound1}.

We now prove \eqref{eq:confluent_ratio_bound2}. By definition, we have 
\begin{flalign*}
\frac{ \Gamma(r+c) U \big(r+c, 1+c, x \big)}{ \Gamma(r) U \big(r, 1, x \big)}  = \frac{\int_0^\infty t^{r+c-1} (1+t)^{-r} e^{-xt} dt}{\int_0^\infty t^{r-1} (1+t)^{-r} e^{-xt} dt} = \mathbb{E}[T_x^c ] = x^{-c} \mathbb{E}[Y_x^c ]
\end{flalign*}
where $T_x$ is a random variable with density $p_{T_x}(t) \propto t^{r-1}(1+t)^{-r} e^{-xt}$ on $[0, \infty)$, and $Y_x:= xT_x$ is a random variable with density $p_{Y_x}(t) \propto y^{r-1}(x+y)^{-r} e^{-y}$ on $[0, \infty)$. 

\begin{enumerate}
\item Consider $\mathbb{E}[Y_x^c ]$ as $x \rightarrow 0^+$. We obtain 
\begin{flalign*}
\lim_{x \rightarrow 0^+} \mathbb{E}[Y_x^c ] = \lim_{x \rightarrow 0^+}  \frac{\int_0^\infty y^{c + r-1}(x+y)^{-r} e^{-y} dy  }{ \int_0^\infty y^{r-1}(x+y)^{-r} e^{-y} dy } = \frac{\Gamma(c) }{  \lim_{x \rightarrow 0^+} \int_0^\infty \big( \frac{y}{x+y} \big)^r  y^{-1} e^{-y} dy } = 0.
\end{flalign*}
\item Consider $\mathbb{E}[T_x^c ]$ as a function of $x$ for $x >0$. From the density of $T_x$, we note that for all $0 < x_1 < x_2$, $T_{x_1}$ stochastically dominates $T_{x_2}$, and as $c >0$, $T^c_{x_1}$ also stochastically dominates $T^c_{x_2}$. Therefore, $\mathbb{E}[T_x^c ]$ is a decreasing function of $x$ for $x > 0$. 
\end{enumerate}

As $\lim_{x \rightarrow 0^+} \mathbb{E}[Y_x^c ]=0$, for any $\epsilon>0$, there exists $\delta>0$ such that $\mathbb{E}[Y_x^c ] < \frac{\Gamma(r+c) }{\Gamma(r)} \epsilon$ for all $x \leq \delta$. This gives $x^{-c}\mathbb{E}[Y_x^c ] < \frac{\Gamma(r+c) }{\Gamma(r)} \epsilon x^{-c}$ for all $x \leq \delta$. Also for all $x > \delta$, we have $x^{-c}\mathbb{E}[Y_x^c ] = \mathbb{E}[T_x^c ] \leq \mathbb{E}[T_\delta^c ] < \infty $ as $\mathbb{E}[T_x^c ]$ is a decreasing function of $x$ for $x > 0$. Overall, we obtain 
\begin{flalign*}
\frac{ \Gamma(r+c) U \big(r+c, 1+c, x \big)}{ \Gamma(r) U \big(r, 1, x \big)}  &= x^{-c} \mathbb{E}[Y_x^c ] < \frac{\Gamma(r+c) }{\Gamma(r)} \epsilon x^{-c} + \mathbb{E}[T_\delta^c ] \\
\Rightarrow \frac{U \big(r+c, 1+c, x \big)}{ U \big(r, 1, x \big)} &< \epsilon x^{-c} + K^{(1)}_{\epsilon, c, r} 
\end{flalign*} 
for $K^{(1)}_{\epsilon, c, r} = \frac{\Gamma(r)}{\Gamma(r+c)} \mathbb{E}[T_\delta^c ] = \frac{ U \big(r+c, 1+c, \delta \big)}{U \big(r, 1, \delta \big)} < \infty $ as required. 

\item \paragraph{Case: $\max \{ -1, -r \} <c<0$.} Note that $U \big(r+c, 1+c, x \big)$ and $U \big(r, 1, x \big)$ are well-defined as $r, r+c, c+1>0$.
\begin{enumerate}
\item Consider $\frac{ U \big(r+c, 1+c, x \big)}{U \big(r, 1, x \big)} - x^{-c}$ as $x \rightarrow \infty$. We first recall the asymptotic expansion of the confluent hypergeometric function of the second kind \citep[Section 13.5]{abramowitz1964handbook} at infinity,
\begin{equation} \label{eq:asympt_U_expansion}
U(a,b,z) = z^{-a} - a(a-b+1)z^{-(a+1)} + \mathcal{O}( z^{-(a+2)} )
\end{equation}
as $z \rightarrow \infty$ for any $a,b > 0$. 
For large $x > 0$, this gives
\begin{flalign}
x^{c+1} \Big( \frac{ U \big(r+c, 1+c, x \big)}{U \big(r, 1, x \big)} - x^{-c} \Big) &= x^{c+1} \Big( \frac{ x^r U \big(r+c, 1+c, x \big) - x^{r-c} U \big(r, 1, x \big) }{x^r U \big(r, 1, x \big)} \Big) \\
&=x \Big( \frac{ x^{r+c}  U \big(r+c, 1+c, x \big) - x^{r} U \big(r, 1, x \big) }{x^r U \big(r, 1, x \big)} \Big) \\
&=x \Big( \frac{(1 - \frac{(r+c)r}{x}) - (1 - \frac{r^2}{x} ) }{1} + \mathcal{O}(x^{-2}) \Big) \label{eq:apply_asympt_U_expansion} \\
&= -rc + \mathcal{O}(x^{-1})
\end{flalign}
where \eqref{eq:apply_asympt_U_expansion} follows from the asymptotic expansion in \eqref{eq:asympt_U_expansion}. As $c+1>0$, we obtain
\begin{flalign} \label{eq:hyper_geom_asymp1}
\lim_{x \rightarrow \infty} \Big( \frac{ U \big(r+c, 1+c, x \big)}{U \big(r, 1, x \big)} - x^{-c} \Big) = 0 .
\end{flalign}
\item Consider $\frac{ U \big(r+c, 1+c, x \big)}{U \big(r, 1, x \big)} - x^{-c}$ as $x \rightarrow 0^+$. By definition of confluent hypergeometric function of the second kind, we have  
\begin{flalign}
\lim_{x \rightarrow 0^+} \frac{ \Gamma(r+c) U \big(r+c, 1+c, x \big)}{ \Gamma(r) U \big(r, 1, x \big)}  = \lim_{x \rightarrow 0^+} \frac{\int_0^\infty t^{r+c-1}(1+t)^{-r} e^{-xt} dt}{\int_0^\infty t^{r-1}(1+t)^{-r} e^{-xt} dt}.
\end{flalign}
Note that as $r+c >0$ and $c-1 < -1$, 
\begin{equation}
\int_0^\infty t^{r+c-1}(1+t)^{-r} dt = \int_0^\infty \big(\frac{t}{1+t}\big)^r t^{c-1} dt \leq \int_0^1 t^{r+c-1} dt + \int_1^\infty t^{c-1} dt = \frac{1}{r+c} - \frac{1}{c} < \infty .
\end{equation}
Therefore, $\lim_{x \rightarrow 0^+} \int_0^\infty t^{r+c-1}(1+t)^{-r} e^{-xt}dt = \frac{1}{r+c} - \frac{1}{c}  < \infty$. However, we have 
\begin{equation}
\lim_{x \rightarrow 0^+} \int_0^\infty t^{r-1}(1+t)^{-r} e^{-xt} dt  = \lim_{x \rightarrow 0^+} \int_0^\infty \big(\frac{t}{1+t}\big)^r t^{-1} e^{-xt} dt  = \infty .
\end{equation}
Therefore $\lim_{x \rightarrow 0^+} \frac{ \Gamma(r+c) U \big(r+c, 1+c, x \big)}{ \Gamma(r) U \big(r, 1, x \big)}  = 0$. As $c < 0$, we obtain
\begin{equation} \label{eq:hyper_geom_asymp2}
\lim_{x \rightarrow 0^+} \frac{ U \big(r+c, 1+c, x \big)}{U \big(r, 1, x \big)} - x^{-c} = 0.
\end{equation}
\end{enumerate}

From \eqref{eq:hyper_geom_asymp1} and \eqref{eq:hyper_geom_asymp2}, we have $ \lim_{x \rightarrow \infty} \frac{ U \big(r+c, 1+c, x \big)}{U \big(r, 1, x \big)} - x^{-c} = 0$ and $ \lim_{x \rightarrow 0^+} \frac{ U \big(r+c, 1+c, x \big)}{U \big(r, 1, x \big)} - x^{-c} = 0$. As $x \mapsto \frac{ U \big(r+c, 1+c, x \big)}{U \big(r, 1, x \big)} - x^{-c}$ is a continuous function on $(0, \infty)$, this implies that $\frac{ U \big(r+c, 1+c, x \big)}{U \big(r, 1, x \big)} - x^{-c}$ is bounded above for $x \in (0, \infty)$. In particular, there exists some $K^{(2)}_{c, r}  < \infty$ such that 
\begin{equation}
\frac{ U \big(r+c, 1+c, x \big)}{U \big(r, 1, x \big)} - x^{-c} < K^{(2)}_{c, r} 
\end{equation}
for all $x > 0$ as required for \eqref{eq:confluent_ratio_bound3}. 
\end{proof}

\begin{Cor} (Bound on moments of the full conditionals of each $\eta_j$). \label{cor:eta_moments_bound}
Consider a random variable $\eta$ on $(0, \infty)$ with density as in \eqref{eq:eta_full_conditional_pdf} 
in Lemma \ref{lemma:eta_moments} for some $\nu >0$. 

Consider when $c>0$. There for every $\epsilon > 0$, 
there exists some $K^{(1)}_{\epsilon, c, \nu} < \infty$ such that
\begin{flalign}
\mathbb{E}[\eta^{c} | m] \leq \epsilon m^{-c} + K^{(1)}_{\epsilon, c, \nu} \quad \text{for all } m>0.
\end{flalign} 

Consider when $\max \{-1, - \frac{1+\nu}{2} \} < c < 0$. Then, there exists some
$K^{(2)}_{c, \nu} < \infty$ such that
\begin{flalign} \label{eq:eta_positive_moment_bound}
\mathbb{E}[\eta^{c} | m] \leq  \frac{\Gamma(\frac{1+\nu}{2}+c)}{ \Gamma(\frac{1+\nu}{2})} m^{-c} + K^{(2)}_{c, \nu}
 \quad \text{for all } m>0.
\end{flalign}
\end{Cor}
\begin{proof}
By Lemma \ref{lemma:eta_moments}, for all $\max \{-1, - \frac{1+\nu}{2} \} < c$,
\begin{flalign} \label{eq:eta_negative_moment_bound}
\mathbb{E}[\eta^{c} | m] =\frac{1}{\nu^{c}}  \frac{\Gamma(\frac{1+\nu}{2}+c)}{ \Gamma(\frac{1+\nu}{2})} \frac{U(\frac{1+\nu}{2}+c, 1+c, \frac{m}{\nu})}{U(\frac{1+\nu}{2}, 1, \frac{m}{\nu})}.
\end{flalign} 

Consider when $c >0$. We apply Lemma \ref{lemma:confluent_ratio_bound} with $r = \frac{1+\nu}{2}$, $x=\frac{m}{\nu}$, and $\tilde{\epsilon}=\epsilon \Big( \frac{\Gamma(\frac{1+\nu}{2}+c)}{ \Gamma(\frac{1+\nu}{2})} \Big)^{-1}$. Then there exists some $ \tilde{K}^{(1)}_{\epsilon, c, \nu} < \infty$ such that
\begin{flalign}
\mathbb{E}[\eta^{c} | m] \leq \frac{1}{\nu^{c}} \frac{\Gamma(\frac{1+\nu}{2}+c)}{ \Gamma(\frac{1+\nu}{2})} \Big(\tilde{\epsilon} \Big(\frac{m}{\nu} \Big)^{-c} + \tilde{K}^{(1)}_{\epsilon, c, \nu} \Big) = \epsilon m^{-c} + K^{(1)}_{{\epsilon}, c, \nu}
\end{flalign}
for $K^{(1)}_{{\epsilon}, c, \nu} = \frac{1}{\nu^{c}}  \frac{\Gamma(\frac{1+\nu}{2}+c)}{ \Gamma(\frac{1+\nu}{2})} \tilde{K}^{(1)}_{\epsilon, c, \nu} < \infty $. 

Consider when $\max \{-1, - \frac{1+\nu}{2} \} < c < 0$. By Lemma \ref{lemma:confluent_ratio_bound} with $r = \frac{1+\nu}{2}$ and $x=\frac{m}{\nu}$, there exists some $\tilde{K}^{(2)}_{c, \nu} < \infty$ such that
\begin{flalign}
\mathbb{E}[\eta^{c} | m] \leq \frac{1}{\nu^{c}}  \frac{\Gamma(\frac{1+\nu}{2}+c)}{ \Gamma(\frac{1+\nu}{2})} \Big( \Big( \frac{m}{\nu} \Big)^{-c} + \tilde{K}^{(2)}_{c, \nu}\Big) = \frac{\Gamma(\frac{1+\nu}{2}+c)}{ \Gamma(\frac{1+\nu}{2})} m^{-c} + K^{(2)}_{c, \nu}
\end{flalign}
for $K^{(2)}_{c, \nu} = \frac{1}{\nu^{c}}  \frac{\Gamma(\frac{1+\nu}{2}+c)}{ \Gamma(\frac{1+\nu}{2})} \tilde{K}^{(2)}_{c, \nu} < \infty $. 
\end{proof}

\begin{Lemma} \label{lemma:normal_inverse_moments}
Let $X \sim \mathcal{N}(\mu, \sigma^2)$ on $\mathbb{R}$. Then, for $c \in (0,1)$, 
\begin{equation}
\mathbb{E}[ |X|^{-c} ] \leq \frac{1}{\sigma^{c} 2^{c/2}} \frac{\Gamma(\frac{1-c}{2})}{\sqrt{\pi}}
\end{equation} 
\end{Lemma}
\begin{proof} We follow the proof in \citet{johndrow2020scalableJMLR}, which is
included here for completeness. Recall the inverse moments of the univariate Normal distribution,
\begin{equation}
\mathbb{E}[ |X|^{-c} ]  = \frac{1}{\sigma^{c} 2^{c/2}} \frac{\Gamma(\frac{1-c}{2})}{\sqrt{\pi}} \exp( - \frac{\mu^2}{2 \sigma^2}) M \Big( \frac{1-c}{2}, \frac{1}{2}, \frac{\mu^2}{2 \sigma^2} \Big)
\end{equation}
where $M(a, b , z) := {}_{1}{F}_{1}(a ; b ; z) $ is the confluent hypergeometric function of the first kind.
% \textcolor{red}{Pierre: consistency issue in the notation $M(a,b;z)$, semi-colon vs comma.}
For $a<b$, we know $M(a, b , z) = \frac{\Gamma(b)}{\Gamma(a) \Gamma(b-a)} \int_{0}^{1} e^{zt} t^{a-1}(1-t)^{b-a-1}dt$.
For $a=\frac{1-c}{2} < b = \frac{1}{2}$ and $z=\frac{\mu^2}{2 \sigma^2}$, this gives
\begin{flalign} 
\exp(-z) M(a, b, z) &= \frac{\Gamma(b)}{\Gamma(a) \Gamma(b-a)} \int_{0}^{1} e^{-z(1-t)} t^{a-1}(1-t)^{b-a-1}dt \\
&= \frac{\Gamma(b)}{\Gamma(a) \Gamma(b-a)} \int_{0}^{1} e^{-zt} (1-t)^{a-1} t^{b-a-1} dt.
\end{flalign}
Therefore, $\exp(-z) M(a, b , z)$ is decreasing for $z\geq 0$ and $\exp(-z) M(a, b , z) \leq \exp(-0) M(a, b , 0)=1$ for all $z\geq 0$. This gives 
\begin{equation}
\mathbb{E}[ |X|^{-c} ] = \frac{1}{\sigma^{c} 2^{c/2}} \frac{\Gamma(\frac{1-c}{2})}{\sqrt{\pi}} \exp( - \frac{\mu^2}{2 \sigma^2}) M \Big( \frac{1-c}{2}, \frac{1}{2}, \frac{\mu^2}{2 \sigma^2} \Big) \leq \frac{1}{\sigma^{c} 2^{c/2}} \frac{\Gamma(\frac{1-c}{2})}{\sqrt{\pi}}.
\end{equation}
\end{proof}

\begin{Lemma} (A uniform bound on the full conditional mean of $1/\sigma^2$). \label{lemma:uniform_bound_mean_sigma2} 
Consider the full conditional distribution of $\sigma^2$ given $\eta, \xi$ in Algorithm \ref{algo:Half_t_Exact}. Given $\eta \in \mathbb{R}_{>0}^p, \xi>0$, we have
\begin{equation}
\sigma^2 | \xi, \eta \sim \InvGamma \bigg( \frac{a_0+n}{2}, \frac{y^T M_{\xi, \eta}^{-1} y + b_0}{2} \bigg),
\end{equation}
where $M_{\xi, \eta} = I_n + \xi^{-1} X \text{Diag}(\eta^{-1}) X^T$. Then 
\begin{equation} \label{eq:uniform_bound_mean_sigma2}
\mathbb{E} \Big[ \frac{1}{\sigma^2} \Big| \xi, \eta \Big] = \frac{a_0+n}{y^T M_{\xi, \eta}^{-1} y + b_0} < \frac{a_0+n}{b_0} < \infty.
\end{equation}
\end{Lemma}
\begin{proof}
Recall that the mean of a $\mathrm{Gamma}(a,b)$ distribution (under shape and rate parameterization) is $a/b$. Equation \eqref{eq:uniform_bound_mean_sigma2} now directly follows as $M_{\xi, \eta}^{-1}$ is positive semi-definite. 
\end{proof}

\begin{Lemma} (Maximal coupling probability of each component $\eta_j$). \label{lemma:eta_coupling_bound}
For fixed constants $m>0$ and $\nu>0$, define a distribution on $(0, \infty)$ 
with probability density function 
\begin{equation} \label{eq:eta_full_conditional_pdf_1}
p(\eta | m) \propto \frac{1}{\eta^{\frac{1-\nu}{2}} (1+ \nu \eta)^{\frac{\nu+1}{2}}} e^{-m \eta}.
\end{equation}
Consider random variables $\eta \sim p(\cdot | m) $, and $\tilde{\eta} \sim p(\cdot | \tilde{m})$
for some fixed $\nu>0$. Fix some $0<a<b<\infty$, and take $m, \tilde{m} \in (a,b)$. 
Then under a maximal coupling,
\begin{equation} \label{eq:eta_coupling_bound}
\mathbb{P}( \eta = \tilde{\eta} | m, \tilde{m} ) = 
\int_0^\infty \min \{ p(\eta^*|m), p(\eta^*|\tilde{m}) \} d\eta^* \geq 
U \Big( \frac{1+\nu}{2}, 1, \frac{b}{\nu}\Big) \Big/ U \Big( \frac{1+\nu}{2}, 1, \frac{a}{\nu}\Big).
\end{equation}
\end{Lemma}
\begin{proof}[Proof of Lemma \ref{lemma:eta_coupling_bound}]
Fix some $0<a<b<\infty$, and take $m, \tilde{m} \in (a,b)$. Then 
\begin{flalign}
\mathbb{P}( \eta = \tilde{\eta} | m, \tilde{m} ) & = 
\int_0^\infty \min \{ p(\eta^*|m), p(\eta^*|\tilde{m}) \} d\eta^* \text{, as } (\eta, \tilde{\eta}) | m, \tilde{m} \text{ are maximally coupled} \\
&= \int \frac{1}{(\eta^*)^{\frac{1-\nu}{2}} (1+ \nu \eta^*)^{\frac{\nu+1}{2}}}  \min \Big\{ \frac{e^{-m \eta^*}}{U \big( \frac{1+\nu}{2}, 1, \frac{m}{\nu}\big)} , \frac{e^{-\tilde{m} \eta^*}}{U \big(\frac{1+\nu}{2}, 1, \frac{\tilde{m}}{\nu}\big)}  \Big\} d \eta^*.
\end{flalign}
Recall $m, \tilde{m} \in (a, b) $. Note also that $x \mapsto U \big(\frac{1+\nu}{2}, 1, x\big)$ is a decreasing function for $x>0$. This gives
\begin{equation}
\min \Big\{ \frac{e^{-m \eta^*}}{U \big( \frac{1+\nu}{2}, 1, \frac{m}{\nu}\big)} , \frac{e^{-\tilde{m} \eta^*}}{U \big(\frac{1+\nu}{2}, 1, \frac{\tilde{m}}{\nu}\big)}  \Big\}
> \frac{e^{-b \eta^*}}{U \big( \frac{1+\nu}{2}, 1, \frac{a}{\nu}\big)}. 
\end{equation}
Therefore for all $m, \tilde{m} \in (a,b)$,
\begin{flalign}
\mathbb{P}( \eta = \tilde{\eta} | m, \tilde{m} ) 
&> \int \frac{1}{(\eta^*)^{\frac{1-\nu}{2}} (1+ \nu \eta^*)^{\frac{\nu+1}{2}}}  \frac{e^{-b \eta^*}}{U \big( \frac{1+\nu}{2}, 1, \frac{a}{\nu}\big)} d \eta^* 
= \frac{U \big( \frac{1+\nu}{2}, 1, \frac{b}{\nu}\big)}{U \big( \frac{1+\nu}{2}, 1, \frac{a}{\nu}\big)}.
\end{flalign}
Equation \eqref{eq:eta_coupling_bound} follows from taking an infimum 
over all $m, \tilde{m} \in (a,b)$.
\end{proof}

\paragraph{Drift and Minorization conditions.}
\begin{proof}[Proof of Proposition \ref{prop:drift}] 
  For ease of notation, denote $\Sigma_{t+1}:= X^TX + \xi_{t+1} \text{Diag}(\eta_{t+1})$ and let $e_j$ be the $j^{th}$ unit vector of the standard basis on $\mathbb{R}^p$. We first consider $\sum_{j=1}^p m_j^{d}$. Note that
\begin{flalign}
\mathbb{E} \Big[ \sum_{j=1}^p m_{t+1,j}^{d} | \eta_{t+1}, \xi_{t+1}, \sigma_{t+1}^2 \Big] &= \Big( \frac{\xi_{t+1}}{2 \sigma_{t+1}^2} \Big)^{d} \sum_{j=1}^p  \mathbb{E} \Big[ \big| \beta_{t+1,j} \big|^{2d} \big| \eta_{t+1}, \xi_{t+1}, \sigma_{t+1}^2 \Big] \\
&\leq \Big( \frac{\xi_{t+1}}{2 \sigma_{t+1}^2} \Big)^{d} \sum_{j=1}^p  \mathbb{E} \Big[ \beta^2_{t+1,j} | \eta_{t+1} , \xi_{t+1}, \sigma_{t+1}^2 \Big]^{d} \label{eq:apply_jensen} \\
&= \Big( \frac{\xi_{t+1}}{2 \sigma_{t+1}^2} \Big)^{d} \sum_{j=1}^p  \bigg( \big(e_j^T \Sigma_{t+1}^{-1} X^Ty\big)^2 + \sigma_{t+1}^2 \big(e_j^T \Sigma_{t+1}^{-1} e_j\big) \bigg)^{d}\\
&\leq \Big( \frac{\xi_{t+1}}{2 \sigma_{t+1}^2} \Big)^{d} \sum_{j=1}^p  \Big( |e_j^T \Sigma_{t+1}^{-1} X^Ty|^{2d} + \sigma_{t+1}^{2d} \big(e_j^T \Sigma_{t+1}^{-1} e_j\big)^{d} \Big)  \label{eq:apply_(a+b)^d}  \\
&\leq \Big( \frac{\xi_{t+1}}{2 \sigma_{t+1}^2} \Big)^{d} \Big( \sum_{j=1}^p \big(|e_j^T \Sigma_{t+1}^{-1} X^Ty|^{2d}\big)^{1/d} \Big)^{d} p^{1-d} \nonumber \\
&\qquad + \sigma_{t+1}^{2d} \Big( \frac{\xi_{t+1}}{2 \sigma_{t+1}^2} \Big)^{d} \sum_{j=1}^p \big(e_j^T \big( \xi_{t+1} 
\text{Diag}(\eta_{t+1}) \big)^{-1} e_j\big)^{d} \label{eq:apply_holder} \\
&= \Big( \frac{\xi_{t+1}}{2 \sigma_{t+1}^2} \Big)^{d} \Big( \|\Sigma_{t+1}^{-1} X^Ty \|^2_2 \Big)^{d} p^{1-d} + \frac{1}{2^{d}} \sum_{j=1}^p \frac{1}{\eta_{t+1,j}^{d}} \\
&\leq \Big( \frac{\xi_{t+1}}{2 \sigma_{t+1}^2} \Big)^{d} \Big( C_p^{2} \| y\|_2^2 \Big)^d p^{1-d} + \frac{1}{2^{d}} \sum_{j=1}^p \frac{1}{\eta_{t+1,j}^{d}}, \label{eq:apply_mean_bound}
\end{flalign}
where \eqref{eq:apply_jensen} follows from Jensen's inequality, \eqref{eq:apply_(a+b)^d} follows as $(x+y)^d \leq x^d + y^d$ for $x, x \geq 0$ and $d\in (0,1)$, \eqref{eq:apply_holder} follows from H\"{o}lder's inequality and 
from $e_j^T \Sigma_{t+1}^{-1}e_j \leq e_j^T\big( \xi_{t+1} \text{Diag}(\eta_{t+1}) \big)^{-1} e_j$ as $X^TX$ is positive-definite matrix, and \eqref{eq:apply_mean_bound} follows from Lemma \ref{lemma:uniform_bound_mean} for some finite constant $C_p $ which does not depend on $\eta_{t+1}$ or $\xi_{t+1}$. 

As the prior $\pi_{\xi}(\cdot)$ on $\xi$ has compact support, we have $b \leq \xi_{t+1} \leq B$ for all $t \geq 0$ for some $0<b \leq B < \infty$. This upper bound on $\xi_{t+1}$, Lemma \ref{lemma:uniform_bound_mean_sigma2} and Jensen's inequality gives
\begin{flalign}
\mathbb{E} \Big[\Big( \frac{\xi_{t+1}}{2 \sigma_{t+1}^2} \Big)^{d} | \eta_{t+1} \Big] \leq \frac{B^{d}}{2^d} \mathbb{E} \Big[ \Big( \frac{1}{\sigma_{t+1}^2} \Big) | \eta_{t+1} \Big]^{d} \leq \frac{B^{d}}{2^d} \big(\frac{a_0+n}{b_0}\big)^d = \Big( \frac{B (a_0+n)}{2 b_0} \Big)^d . \label{eq:xi_sigma_bound}
\end{flalign}

Therefore, by \eqref{eq:apply_mean_bound} and \eqref{eq:xi_sigma_bound}, we obtain
\begin{flalign} \label{eq:m_positive_given_eta}
\mathbb{E} \Big[ \sum_{j=1}^p m_{t+1,j}^{d} | \eta_{t+1} \Big] = \mathbb{E} \Big[ \mathbb{E} \Big[ \sum_{j=1}^p m_{t+1,j}^{d} | \eta_{t+1}, \xi_{t+1}, \sigma_{t+1}^2 \Big] | \eta_{t+1} \Big] \leq \frac{1}{2^{d}} \sum_{j=1}^p \frac{1}{\eta_{t+1,j}^{d}} + K_1
\end{flalign}
for $K_1 := \Big( \frac{B (a_0+n)}{2 b_0} \Big)^d \Big( C_p^{2} \| y\|_2^2 \Big)^d p^{1-d} < \infty$. By \eqref{eq:m_positive_given_eta} and Corollary \ref{cor:eta_moments_bound}, we obtain
\begin{flalign}
\mathbb{E} \Big[ \sum_{j=1}^p m_{t+1,j}^{d} | \beta_{t}, \xi_{t}, \sigma_t^2 \Big] &= \mathbb{E} \Big[ \mathbb{E} \Big[ \sum_{j=1}^p m_{t+1,j}^{d} | \eta_{t+1} \Big] | \beta_{t}, \xi_{t}, \sigma_t^2 \Big] \\
&\leq \frac{1}{2^{d}} \sum_{j=1}^p \mathbb{E} \Big[ \eta_{t+1,j}^{-d} | \beta_{t}, \xi_{t}, \sigma_t^2 \Big] + K_1 \\
&\leq \frac{1}{2^{d}} \Big( \sum_{j=1}^p \frac{\Gamma(\frac{1+\nu}{2}-d)}{ \Gamma(\frac{1+\nu}{2})}  m_{t,j}^{d} + K^{(2)}_{d, \nu} \Big) + K_1 \\
&\leq \frac{1}{2^{d}} \frac{\Gamma(\frac{1+\nu}{2}-d)}{ \Gamma(\frac{1+\nu}{2})} \sum_{j=1}^p   m_{t,j}^{d} + K_2 \label{eq:lyapunov_cutaway_inf}
\end{flalign}
for $K_2 := \frac{p}{2^{d}} K^{(2)}_{d, \nu} + K_1 < \infty$. Note that for every $\nu \geq 1$, there exists $d \in (0,1)$ 
such that $\frac{1}{2^{d}} \frac{\Gamma(\frac{\nu+1}{2}-d)}{\Gamma(\frac{\nu+1}{2})} < 1$. We can verify this by considering the
function $f(d) := \frac{1}{2^{d}} \frac{\Gamma(\frac{\nu+1}{2}-d)}{\Gamma(\frac{\nu+1}{2})}$. It suffices to note that
$f(0) = 1$ and $f'(0)= -\psi(\frac{\nu+1}{2}) - \log(2) < 0$ for $\nu \geq 1$, where $\psi(x):= \frac{\Gamma'(x)}{\Gamma(x)}$
is the Digamma function.

We now consider $\sum_{j=1}^p m_j^{-c}$ for $c\in (0,1/2)$. We obtain
\begin{flalign}
\mathbb{E} \Big[ \sum_{j=1}^p & m_{t+1,j}^{-c} | \sigma^2_{t+1}, \xi_{t+1}, \eta_{t+1} \Big] = \Big( \frac{\xi_{t+1}}{2 \sigma_{t+1}^2} \Big)^{-c} \sum_{j=1}^p  \mathbb{E} \Big[ \big| \beta_{t+1,j} \big|^{-2c} \big| \sigma^2_{t+1}, \xi_{t+1}, \eta_{t+1}  \Big]  \\
&\leq \Big( \frac{\xi_{t+1}}{2 \sigma_{t+1}^2} \Big)^{-c} \frac{1}{2^{c}}  \frac{\Gamma(\frac{1-2c}{2})}{\sqrt{\pi}}  \sum_{j=1}^p \big( \sigma_{t+1}^2 e_j^T \Sigma_{t+1}^{-1} e_j \big)^{-c} \label{eq:apply_normal_inverse_moments} \\
&\leq \Big( \frac{\xi_{t+1}}{2 \sigma_{t+1}^2} \Big)^{-c} \frac{1}{2^{c}}  \frac{\Gamma(\frac{1-2c}{2})}{\sqrt{\pi}}  \sum_{j=1}^p \big( \sigma_{t+1}^2 e_j^T (\|X\|_2^2 I_p + \xi_{t+1} \text{Diag}(\eta_{t+1}))^{-1} e_j \big)^{-c} \label{eq:apply_matrix_norm_bound} \\
&= \Big( \frac{\xi_{t+1}}{2 \sigma_{t+1}^2} \Big)^{-c} \frac{1}{2^{c}}  \frac{\Gamma(\frac{1-2c}{2})}{\sqrt{\pi}}  \sum_{j=1}^p \frac{(\|X\|_2^2 + \xi_{t+1} \eta_{t+1,j})^c}{\sigma_{t+1}^{2c}} \\
&\leq  \frac{\Gamma(\frac{1}{2}-c)}{\sqrt{\pi}}  \sum_{j=1}^p \Big( \Big(\frac{\|X\|_2^{2}}{\xi_{t+1}} \Big)^c + \eta_{t+1,j}^c \Big) \label{eq:m_negative_given_eta1} \\
&\leq  \frac{\Gamma(\frac{1}{2}-c)}{\sqrt{\pi}}  \sum_{j=1}^p \eta_{t+1,j}^c + p \frac{\Gamma(\frac{1}{2}-c)}{\sqrt{\pi}} \Big(\frac{\|X\|_2^{2}}{b} \Big)^c \label{eq:m_negative_given_eta2} 
\end{flalign}
for some $b>0$. Here \eqref{eq:apply_normal_inverse_moments} follows from Lemma \ref{lemma:normal_inverse_moments}, \eqref{eq:apply_matrix_norm_bound} follows as $e_j^T ( \|X\|_2^2 I_p + \xi_{t} \text{Diag}(\eta_t) )^{-1} e_j \leq e_j^T \Sigma_{t+1}^{-1} e_j$ for $\| \cdot \|_2$ the $L_2$ operator norm of a matrix, \eqref{eq:m_negative_given_eta1} follows as $(x+y)^c \leq x^c + y^c$ for all $x,y>0$ and $0< c <1/2$, and \eqref{eq:m_negative_given_eta2} follows as the prior on $\xi$ has compact support, so $\xi_{t+1} \geq b$ for some $b>0$. Therefore, by \eqref{eq:m_negative_given_eta2},
\begin{flalign} \label{eq:m_negative_given_eta3}
& \mathbb{E} \Big[ \sum_{j=1}^p m_{t+1,j}^{-c} | \eta_{t+1} \Big] = \mathbb{E} \Big[ \mathbb{E} \Big[ \sum_{j=1}^p m_{t+1,j}^{-c} | \sigma^2_{t+1}, \xi_{t+1}, \eta_{t+1} \Big] | \eta_{t+1} \Big] \leq \frac{\Gamma(\frac{1}{2}-c)}{\sqrt{\pi}}  \sum_{j=1}^p \eta_{t+1,j}^c + K_3. &
\end{flalign}
for $K_3 := p \frac{\Gamma(\frac{1}{2}-c)}{\sqrt{\pi}} \Big(\frac{\|X\|_2^{2}}{b} \Big)^c < \infty$. By \eqref{eq:m_negative_given_eta3} and Corollary \ref{cor:eta_moments_bound} with $r=\frac{1+\nu}{2}$ and any $\epsilon < \frac{\sqrt{\pi}}{\Gamma(\frac{1}{2}-c)}$, there exists $K^{(1)}_{\epsilon, c, \nu} < \infty$ such that 
\begin{flalign} 
\mathbb{E} \Big[ \sum_{j=1}^p m_{t+1,j}^{-c} | \sigma^2_{t}, \xi_{t}, \beta_{t}  \Big] &= \mathbb{E} \Big[  \mathbb{E} \Big[ \sum_{j=1}^p m_{t+1,j}^{-c} | \eta_{t+1} \Big] | \sigma^2_{t}, \xi_{t}, \beta_{t}  \Big] \\
&\leq \frac{\Gamma(\frac{1}{2}-c)}{\sqrt{\pi}} \sum_{j=1}^p \mathbb{E} \Big[ \eta_{t+1,j}^c | \sigma^2_{t}, \xi_{t}, \beta_{t}  \Big] + K_3 \\
&\leq \frac{\Gamma(\frac{1}{2}-c)}{\sqrt{\pi}} \sum_{j=1}^p \Big( \epsilon m_{t,j}^{-c} +K^{(1)}_{\epsilon, c, \nu}  \Big) + K_3 \\
&= \Big( \frac{\Gamma(\frac{1}{2}-c)}{\sqrt{\pi}} \epsilon \Big) \sum_{j=1}^p m_{t,j}^{-c}  + K_4, \label{eq:lyapunov_cutaway_zero}
\end{flalign}
where $\frac{\Gamma(\frac{1}{2}-c)}{\sqrt{\pi}}  \epsilon < 1$ and 
$K_4 = p \frac{\Gamma(\frac{1}{2}-c)}{\sqrt{\pi}}K^{(1)}_{\epsilon, c, \nu} + K_3$. 

We can now combine \eqref{eq:lyapunov_cutaway_inf} and \eqref{eq:lyapunov_cutaway_zero} together. Then, for some fixed $d \in (0,1)$ such that $\frac{1}{2^{d}} \frac{\Gamma(\frac{\nu+1}{2}-d)}{\Gamma(\frac{\nu+1}{2})} < 1$ and any $c \in (0,1/2)$, we obtain 
\begin{flalign}
\mathbb{E} \Big[ \sum_{j=1}^p m_{t+1,j}^{-c} + m_{t+1,j}^{d} | m_{t} \Big] \leq \gamma_{drift} \Big( \sum_{j=1}^p m_{t,j}^{-c} + m_{t,j}^{d} \Big) + K_{drift}
\end{flalign}
for some
\begin{flalign}
0 < \gamma_{drift} &:= max \Big\{\frac{\Gamma(\frac{1}{2}-c)}{\sqrt{\pi}}  \epsilon, \frac{1}{2^{d}} \frac{\Gamma(\frac{\nu+1}{2}-d)}{\Gamma(\frac{\nu+1}{2})} \Big\} <1 \\
0 < K_{drift} &:= K_2 + K_4 < \infty.
\end{flalign}
Therefore, for any such $d \in (0,1)$ and $c\in (0,1/2)$,
\begin{equation}
V(\beta) = \sum_{j=1}^p m_j^{-c} + m_j^{d}
\end{equation}
for $m_j = \frac{\xi \beta_j^2}{2 \sigma^2}$ is a Lyapunov function. 
\end{proof}

\begin{Prop} (Minorization condition). \label{prop:minorization}
For $R>0$, let 
$S(R) = \{ (\beta, \xi, \sigma^2) \in \mathbb{R}^p \times \mathbb{R}_{>0} \times \mathbb{R}_{>0} : V(\beta, \xi, \sigma) < R \}$ 
denote the sub-level sets of the Lyapunov function in Proposition
\ref{prop:drift}. Let $\mathcal{P}$ denote the Markov transition kernel 
associated with the update from $(\beta_t,\xi_t, \sigma_t^2)$ to $(\beta_{t+1},\xi_{t+1}, \sigma_{t+1}^2)$ 
implied by the update rule in \eqref{eq:blocked_gibbs_ergodicity_version}. 
Let $Z = (\beta, \xi, \sigma^2)$ and $\tilde{Z} = (\tilde{\beta}, \tilde{\xi}, \tilde{\sigma}^2)$.
Then for all $R > 0$,
%\textcolor{blue}{every $R > 2K/(1-\gamma)$} 
there exists $\epsilon \in (0,1)$ such that 
\begin{equation} \label{eq:minorization}
\text{TV} \Big( \mathcal{P} ( Z, \cdot ), \mathcal{P} (\tilde{Z}, \cdot) \Big) < 1 - \epsilon \quad \text{ for all } Z, \tilde{Z} \in S(R),
\end{equation}
where $\text{TV}$ stands for the total variation distance between two probability distributions.
In particular for $R > 1$, $\epsilon = ( U \big( \frac{1+\nu}{2}, 1, \frac{R^{1/d}}{\nu}\big) / U \big( \frac{1+\nu}{2}, 1, \frac{R^{-1/c}}{\nu}\big))^p$
suffices. 
\end{Prop}

\begin{proof}[Proof of Proposition \ref{prop:minorization}]
Denote $Z = (\beta, \xi, \sigma^2)$ and $\tilde{Z} = 
(\tilde{\beta}, \tilde{\xi}, \tilde{\sigma}^2)$, and let $Z, \tilde{Z} \in S(R)$
for some $R>1$. We obtain
\begin{flalign}
\text{TV} & \big( \mathcal{P}(Z, \cdot), \mathcal{P}(\tilde{Z}, \cdot) \big) 
\leq \mathbb{P} \big( Z_{1} \neq \tilde{Z}_{1} \big| Z, \tilde{Z} \big) < 1 - \epsilon ,
\end{flalign}
where $(Z_{1}, \tilde{Z}_{1})$ corresponds to one-scale coupling of 
$\mathcal{P}(Z, \cdot)$ and $\mathcal{P}(\tilde{Z}, \cdot)$ as in
Proposition \ref{prop:one_scale_meeting}, the first inequality follows from the 
coupling inequality, and the second inequality follows from the proof of Proposition
\ref{prop:one_scale_meeting} for 
$\epsilon = \Big( U \big( \frac{1+\nu}{2}, 1, \frac{R^{1/d}}{\nu}\big) \Big/ U \big( \frac{1+\nu}{2}, 1, \frac{R^{-1/c}}{\nu}\big) \Big)^p \in (0,1)$ 
as $R>1$. Note that $\epsilon$ does not depend on $Z$ and $\tilde{Z}$
as required. 
\end{proof}

Geometric ergodicity follows from Proposition
\ref{prop:drift} and Proposition \ref{prop:minorization} with any $R > \max \{1,2K/(1-\gamma)\}$
\citep{rosenthal1995minorizationJASA}. 

%\begin{proof}[Proof of Theorem \ref{thm:geo_ergodicity}]
%Directly follows from Propositions \ref{prop:drift}, \ref{prop:minorization} and Theorem \ref{Th:geo_ergodicity_d_and_m}.
%\end{proof}

%\textcolor{red}{Pierre: perhaps here we could add a sentence on how the $\epsilon$
%  varies in, say, $p$, and how this aligns with other observations that the drift+minorization
% approach does not provide a sharp understanding of convergence rates for high-dimensional
% inference settings.}

\subsection{Coupled Kernel Proofs}

\begin{proof}[Proof of Proposition \ref{prop:one_scale_meeting}]
For $R>0$, let 
$S(R) := \{ (\beta, \xi, \sigma^2) \in \mathbb{R}^p \times \mathbb{R}_{>0} \times \mathbb{R}_{>0} : V(\beta, \xi, \sigma) < R \}$
denote the sub-level sets of the Lyapunov function in Proposition \ref{prop:drift}, with
corresponding contraction constant $\gamma \in (0,1)$ and additive constant $K>0$. 
Let $ Z_t := (\beta_t, \xi_t, \sigma_t^2)$ and $\tilde{Z}_t = (\tilde{\beta}_t, \tilde{\xi}_t, \tilde{\sigma}_t^2)$
for all $t \geq 0$. By \citet[][Proposition 4]{jacob2020unbiasedJRSSB}, it suffices to prove that  
for any $R>0$ sufficiently large such that $\gamma + K/(1+R) < 1$, there exists 
$\epsilon \in (0,1)$ such that 
\begin{equation} \label{eq:one_scale_condition}
\inf_{Z_t, \tilde{Z}_t \in S(R)} 
\mathbb{P} \big( Z_{t+1} = \tilde{Z}_{t+1} \big| Z_{t}, \tilde{Z}_{t} \big) > \epsilon,
\end{equation}
under the one-scale coupling. We now prove \eqref{eq:one_scale_condition} directly for the 
one-scale coupling in Proposition \ref{prop:one_scale_meeting}. 
Fix some $R>1$ such that $\gamma + K/(1+R) < 1$, and take any $Z_t, \tilde{Z}_t \in S(R)$. 
Then
\begin{flalign}
\mathbb{P} & ( Z_{t+1} = \tilde{Z}_{t+1} \big| Z_{t}, \tilde{Z}_{t}) \\
& = \mathbb{E} \Big[ \mathbb{P} \Big( Z_{t+1} = \tilde{Z}_{t+1} \big| \eta_{t+1}, \tilde{\eta}_{t+1}, Z_{t}, \tilde{Z}_{t} \Big)
\Big| Z_{t}, \tilde{Z}_{t} \Big] \\
& = \mathbb{E} \Big[ \mathbb{P} \Big( Z_{t+1} = \tilde{Z}_{t+1} \big| \eta_{t+1}, \tilde{\eta}_{t+1} \Big)
\Big| Z_{t}, \tilde{Z}_{t} \Big] \\
& = \mathbb{E} \Big[ \mathrm{1}_{ \{ \eta_{t+1} = \tilde{\eta}_{t+1} \} } \Big| Z_{t}, \tilde{Z}_{t} \Big]
\text{, as } (Z_{t+1},\tilde{Z}_{t+1}) | (\eta_{t+1}, \tilde{\eta}_{t+1}) \text{ follows CRN coupling} \\
& = \mathbb{P} \big( \eta_{t+1} = \tilde{\eta}_{t+1} \Big| Z_{t}, \tilde{Z}_{t} \big) \\
& = \prod_{j=1}^p \mathbb{P} \big( \eta_{t+1,j} = \tilde{\eta}_{t+1,j} \Big| Z_{t}, \tilde{Z}_{t} \big)
\text{, as } (\eta_{t+1}, \tilde{\eta}_{t+1}) | Z_{t}, \tilde{Z}_{t} \text{ is maximally coupled componentwise} \\
&= \prod_{j=1}^p   \mathbb{P} \big( \eta_{t+1,j} = \tilde{\eta}_{t+1,j} \Big| m_j, \tilde{m}_j \big)
\text{, for} \ m_j = \frac{\xi_t \beta_{t,j}^2}{2\sigma_t^2} \text{ and }
\tilde{m}_j = \frac{ \tilde{\xi}_t \tilde{\beta}_{t,j}^2}{2 \tilde{\sigma}_t^2}.
%&= \int \prod_{j=1}^p   \min \Big\{ p( \eta_j^* | m_j ) , p( \eta_j^* | \tilde{m}_j ) \Big\} d \eta_j^* 
%\ \text{, for} \ m_j = \frac{\xi_t \beta_{t,j}^2}{2\sigma_t^2}, 
%\tilde{m}_j = \frac{ \tilde{\xi}_t \tilde{\beta}_{t,j}^2}{2 \tilde{\sigma}_t^2} \\
%&=  \prod_{j=1}^p \Big( \int \min \Big\{ p( \eta_j^* | m_j ) , p( \eta_j^* | \tilde{m}_j ) \Big\} d \eta_j^* \Big) \\
%&=  \prod_{j=1}^p \Big( \int \frac{1}{(\eta^*_j)^{\frac{1-\nu}{2}} (1+ \nu \eta_j^*)^{\frac{\nu+1}{2}}}  \min \Big\{ \frac{e^{-m_j \eta_j^*}}{U \big( \frac{1+\nu}{2}, 1, \frac{m_j}{\nu}\big)} , \frac{e^{-\tilde{m}_j \eta_j^*}}{U \big(\frac{1+\nu}{2}, 1, \frac{\tilde{m}_j}{\nu}\big)}  \Big\} d \eta_j^* \Big) .
%& \geq \prod_{j=1}^p \frac{U \big( \frac{1+\nu}{2}, 1, \frac{R^{1/d}}{\nu}\big)}{U \big( \frac{1+\nu}{2}, 1, \frac{R^{-1/c}}{\nu}\big)}
%\text{ by Lemma \ref{lemma:eta_coupling_bound}, as in Proof of Proposition \ref{prop:minorization} for } Z_t, \tilde{Z}_t \in S(R).\\
%&= \Big( \frac{U \big( \frac{1+\nu}{2}, 1, \frac{R^{1/d}}{\nu}\big)}{U \big( \frac{1+\nu}{2}, 1, \frac{R^{-1/c}}{\nu}\big)} \Big)^p =: \epsilon,
\end{flalign}

As $Z, \tilde{Z} \in S(R)$ for some $R>1$, $m_j, \tilde{m}_j \in (R^{-1/c}, R^{1/d}) $ for all $j=1,\ldots, p$. 
By Lemma \ref{lemma:eta_coupling_bound}, 
\begin{flalign}
\mathbb{P} \big( \eta_{t+1,j} = \tilde{\eta}_{t+1,j} \Big| m_j, \tilde{m}_j \big) \geq 
U \big( \frac{1+\nu}{2}, 1, \frac{R^{1/d}}{\nu}\big) \Big/ U \big( \frac{1+\nu}{2}, 1, \frac{R^{-1/c}}{\nu}\big)
\end{flalign}
%\begin{flalign}
%\int \frac{1}{(\eta^*_j)^{\frac{1-\nu}{2}} (1+ \nu \eta_j^*)^{\frac{\nu+1}{2}}}  \frac{e^{-R^{1/d} \eta_j^* }}{U \big( \frac{1+\nu}{2}, 1, \frac{R^{-1/c}}{\nu}\big)} d \eta_j^* 
%\geq \frac{U \big( \frac{1+\nu}{2}, 1, \frac{R^{1/d}}{\nu}\big)}{U \big( \frac{1+\nu}{2}, 1, \frac{R^{-1/c}}{\nu}\big)}. 
%\end{flalign}
for all $j=1,\ldots, p$. Therefore, 
\begin{flalign}
\mathbb{P} & ( Z_{t+1} = \tilde{Z}_{t+1} \big| Z_{t}, \tilde{Z}_{t}) 
\geq \prod_{j=1}^p \frac{U \big( \frac{1+\nu}{2}, 1, \frac{R^{1/d}}{\nu}\big)}{U \big( \frac{1+\nu}{2}, 1, \frac{R^{-1/c}}{\nu}\big)} 
= \epsilon
\end{flalign}
for $\epsilon := ( U ( \frac{1+\nu}{2}, 1, \frac{R^{1/d}}{\nu}) / U ( \frac{1+\nu}{2}, 1, \frac{R^{-1/c}}{\nu}))^p \in (0,1)$ as $R>1$.
\end{proof}

\begin{proof}[Proof of Proposition \ref{prop:crn_max_expectation}]
As $\eta_{t+1}$ given $m_t$ and $\tilde{\eta}_{t+1}$ given $\tilde{m}_t$ are 
conditionally independent component-wise, it suffices to consider 
the univariate case where $p=1$. Henceforth we assume $p=1$ 
and we drop the component subscripts for ease of notation. 
Without loss of generality, take $\tilde{m}_t \leq m_t$
and $h:(0, \infty)$ to be monotone increasing. 
Then the densities of $\eta_{t+1}|m_t$ and $\tilde{\eta}_{t+1}|\tilde{m}_t$ 
on $(0, \infty)$ are given by 
\begin{flalign} \label{eq:half_t_full_conditionals}
p(\eta | m_t) \propto 
\frac{e^{-m_t \eta}}{\eta^{\frac{1-\nu}{2}} (1+\nu\eta)^{\frac{\nu+1}{2}}} 
\text{ and } 
p(\tilde{\eta} | \tilde{m}_t) \propto 
\frac{e^{-\tilde{m}_t \eta}}{\eta^{\frac{1-\nu}{2}} (1+\nu\eta)^{\frac{\nu+1}{2}}}
\end{flalign}
respectively. As $\tilde{m}_t \leq m_t$, the distribution of 
$\tilde{\eta}_{t+1}$ stochastically dominates the distribution of $\eta_{t+1}$. 
Therefore under a CRN coupling of ${\eta}_{t+1}$ and $\tilde{\eta}_{t+1}$ 
(specifically a CRN coupling of Algorithm \ref{algo:eta_perfect_sampling}), $\tilde{\eta}_{t+1} \geq \eta_{t+1}$
almost surely, and $h(\tilde{\eta}_{t+1}) \geq h(\eta_{t+1})$ almost surely. 
We obtain 
\begin{equation}
\mathbb{E}_{ \text{crn} } \big[ | h(\eta_{t+1}) - h(\tilde{\eta}_{t+1}) | \big| m_{t}, \tilde{m}_{t} \big] = 
\mathbb{E} \big[h(\tilde{\eta}_{t+1}) - h(\eta_{t+1}) \big| m_{t}, \tilde{m}_{t} \big].
\end{equation}

Now consider a maximal coupling of $\eta_{t+1}$ and $\tilde{\eta}_{t+1}$ with independent residuals
\citep{johnson1998coupling}. Let $q=\text{TV}(p(\cdot | m_t) , p(\cdot | \tilde{m}_t))$. 
Then $\eta_{t+1}=\tilde{\eta}_{t+1}$ with probability
$1-q$. Otherwise with probability $q$, 
$(\eta_{t+1}, \tilde{\eta}_{t+1}) = (\eta^*_{t+1}, \tilde{\eta}^*_{t+1})$
where $\eta^*_{t+1}$ and $\tilde{\eta}^*_{t+1}$ are independent random variables 
with densities
\begin{flalign}
p(\eta^*_{t+1} | m_{t+1}, \tilde{m}_{t+1}) &= 
\max \{ p(\eta^*_{t+1} | m_t) - p(\eta^*_{t+1} | \tilde{m}_t), 0 \}/q \text{ and } \\
p(\tilde{\eta}^*_{t+1} | m_{t+1}, \tilde{m}_{t+1}) &= 
\max \{ p(\tilde{\eta}^*_{t+1} | \tilde{m}_t) - p(\tilde{\eta}^*_{t+1} | {m}_t), 0 \}/q
\end{flalign}
respectively on $(0, \infty)$. 
This implies that the supports of $\eta^*_{t+1}$ and $\tilde{\eta}^*_{t+1}$ are disjoint. 
%is $[0, K)$  and $[K, \infty)$ respectively for some $K>0$ (e.g. see Proposition \ref{prop:slice_sampler_tv}
%for an explicity expression for $K$), 
By \eqref{eq:half_t_full_conditionals}, we further have $\eta^*_{t+1} \leq \tilde{\eta}^*_{t+1}$ almost surely. 
Overall, we obtain $\tilde{\eta}_{t+1} \geq \eta_{t+1}$ almost surely under a maximal
coupling, and hence 
\begin{equation}
\mathbb{E}_{ \text{max} } \big[ | h(\eta_{t+1}) - h(\tilde{\eta}_{t+1}) | \big| m_{t}, \tilde{m}_{t} \big] = 
\mathbb{E} \big[ h(\tilde{\eta}_{t+1}) - h(\eta_{t+1}) \big| m_{t}, \tilde{m}_{t} \big] =
\mathbb{E}_{ \text{crn} } \big[ | h(\eta_{t+1}) - h(\tilde{\eta}_{t+1}) | \big| m_{t}, \tilde{m}_{t} \big]
\end{equation}
as required.
%By symmetry, the above argument holds for all $m_t, \tilde{m}_t > 0$ and all monotone 
%functions $h$. The general multivariate case with $p \geq 1$ follows by linearity 
%and as $\eta_{t+1}$ given $m_t$ and $\tilde{\eta}_{t+1}$ give $\tilde{m}_t$ are 
%conditionally independent component-wise.
\end{proof}

\begin{proof}[Proof of Proposition \ref{prop:crn_max_high_prob}]
\item \paragraph{CRN coupling result.} To prove \eqref{eq:crn}, it suffices to show that under CRN,
\begin{equation} \label{eq:crn_as_bound}
| \log (1+\eta_{t+1,j}) - \log (1+\tilde{\eta}_{t+1,j}) | \leq | \log m_{t, j} - \log \tilde{m}_{t, j} |,
\end{equation}
almost surely for all $j=1, \ldots, p$. Henceforth, to prove \eqref{eq:crn_as_bound}
we work component-wise and drop the subscripts for ease of notation. 
For the Horseshoe prior, we then have $p(\eta | m) \propto \exp(-m\eta)/(1+\eta)$
on $(0, \infty)$. The cumulative density function is then given by 
$F_m(t) = 1 - E_1(m(t+1))/E_1(m)$ where $E_1$ is the exponential integral function. Therefore, 
under CRN coupling, we obtain
\begin{flalign}
1 + \eta = \frac{1}{m} E_1^{-1}(U^* E_1(m)) \text{ and }
1 + \tilde{\eta} = \frac{1}{\tilde{m}} E_1^{-1}(U^* E_1(\tilde{m})),
\end{flalign}
for $U^* \sim \Uniform(0,1)$. Then 
\begin{flalign}
\frac{d}{d \log m} \big( \log(1 + \eta) \big) &= -1 + \frac{d}{d \log m} \big( \log \big( E_1^{-1}(U^* E_1(m)) \big) \big) \\
&= -1 + \frac{m}{ E_1^{-1}(U^* E_1(m)) } \Big( \frac{d}{d m} \big( E_1^{-1}(U^* E_1(m)) \big) \Big) \\
&= -1 + \frac{m}{ E_1^{-1}(U^* E_1(m)) } \frac{U^* E_1'(m) }{  E_1'\big( E_1^{-1}(U^* E_1(m)) \big) } 
\text{, where } E'_1(z) = -e^{-z}/z \\
&= -1 + \frac{m}{ m(1+\eta) } \frac{U^* E_1'(m) }{  E_1'\big( m(1+\eta) \big) } 
\text{, as } m(1+\eta) = E_1^{-1}(U^* E_1(m)) \\
&=  -1 + U^* e^{m \eta} \text{, as } E'_1(z) = -e^{-z}/z \\
%&= -1 + U^* e^{m \eta} \\
&= -1 +\frac{E_1(m(1+\eta)) e^{m (1+\eta)}}{E_1(m) e^{m}} \text{, as } U^* = E_1(m(1+\eta))/E_1(m) \\
&\in (-1,0) \text{ as } x \mapsto e^x E_1(x) \text{ is positive and decreasing on } (0, \infty). 
\end{flalign}
Therefore, as $|\frac{d \log(1 + \eta)}{d \log m}| \leq 1$ and by the mean value theorem we obtain 
\begin{equation}
| \log(1 + \eta) - \log(1 + \tilde{\eta}) | \leq | \log m - \log \tilde{m} |,
\end{equation}
almost surely as required for \eqref{eq:crn_as_bound}. Finally, \eqref{eq:crn_as_bound} 
implies \eqref{eq:crn} with $m_{t, j} = \mu$ and $\tilde{m}_{t, j} = \tilde{\mu}$ for all $j = 1, \ldots, p$.
\item \paragraph{Maximal coupling result.} To prove \eqref{eq:max}, take $\tilde{\mu} < \mu$
without loss of generality. For $c>0$, denote the event
$$E_c := \big\{ \max_{j=1}^p | \log(1+\eta_{t+1,j}) - \log(1+\tilde{\eta}_{t+1,j}) | > c \big\},$$
where $(\eta_{t+1,j}, \tilde{\eta}_{t+1,j})$ are i.i.d. random variables jointly distributed 
according to a maximal coupling.
% of $p(\cdot | \mu)$ and $p(\cdot | \tilde{\mu})$. 
By component-wise independence, 
\begin{flalign}
\mathbb{P}(E_c | \mu, \tilde{\mu}) = 
1 - \prod_{j=1}^p \mathbb{P} \big( | \log(1+\eta_{t+1,j}) - \log(1+\tilde{\eta}_{t+1,j}) | \leq c \big| \mu, \tilde{\mu} \big) 
= 1 - x_c^p
\end{flalign}
for $x_c=\mathbb{P} \big( | \log(1+\eta_{t+1,1}) - \log(1+\tilde{\eta}_{t+1,1}) | \leq c \big| \mu, \tilde{\mu} \big)$.
Henceforth, we focus on this quantity $x_c$ and drop the subscripts for ease of notation. 
For the Horseshoe prior, the marginal densities densities of $\eta|\mu$ and $\tilde{\eta}|\tilde{\mu}$ 
are given by 
\begin{flalign} \label{eq:horseshoe_pdfs}
p(\eta | \mu) = \frac{e^{-\mu \eta}}{1+\eta} \frac{1}{e^\mu E_1(\mu)} \text{ and } 
p(\eta | \tilde{\mu}) = \frac{e^{-\tilde{\mu} \eta}}{1+\eta} \frac{1}{e^{\tilde{\mu}} E_1(\tilde{\mu})}
\text{ on } (0, \infty).
\end{flalign}

Consider a maximal coupling of $\eta$ give $\mu$ and $\tilde{\eta}$ given $\tilde{\mu}$
with independent residuals \citep{johnson1998coupling}. Let $q=\text{TV}(p(\cdot | \mu) , p(\cdot | \tilde{\mu}))$. 
Then $\eta=\tilde{\eta}$ with probability $1-q$. Otherwise with probability $q$, 
$(\eta, \tilde{\eta}) = (\eta^*, \tilde{\eta}^*)$ where $\eta^*$ and $\tilde{\eta}^*$ 
are independent random variables on $(0, \infty)$ with densities
\begin{flalign} 
p(\eta^* | \mu, \tilde{\mu}) &= \max \{ p(\eta^* | \mu) - p(\eta^* | \tilde{\mu}), 0 \}/q \text{ and } \\
p(\tilde{\eta}^* | \mu, \tilde{\mu}) &= \max \{ p(\tilde{\eta}^* | \tilde{\mu}) - p(\tilde{\eta}^* | {\mu}), 0 \}/q, \label{eq:horseshoe_residuals}
\end{flalign}
respectively. Note that $p(\eta^* | \mu) \geq p(\eta^* | \tilde{\mu})$ if and only if  $1 + \eta^* \leq K$ for constant 
$K:=(\log E_1(\mu) - \log E_1(\tilde{\mu}))/(\tilde{\mu}-\mu)$. Therefore, $1+\eta^*$ has support $[1, K)$
and $1+\tilde{\eta}^*$ has support $[K, \infty)$. 
We obtain, 
\begin{flalign}
x_c &= (1-q) + q \mathbb{P} \big( | \log(1+\eta^*) - \log(1+\tilde{\eta}^*) | \leq c \big) \\
&= 1 - q \mathbb{P} \big( | \log(1+\eta^*) - \log(1+\tilde{\eta}^*) | > c \big) \\
&= 1 - q \mathbb{P} \big( \log(1+\tilde{\eta}^*) - \log(1+\eta^*) > c \big)  \text{, as } \eta^* \leq \tilde{\eta}^* \\
&= 1 - q \mathbb{P} \big( 1+\tilde{\eta}^* > (1+\eta^*) e^c \big) \\
&\leq 1 - q \mathbb{P} \big( 1+\tilde{\eta}^* > t \big) \text{ for } t=K e^c \text{, as } 1+\tilde{\eta}^* \leq K \text{ almost surely } \\ 
&= 1 - \int_{t}^\infty \frac{1}{1+\tilde{\eta}^* } \Big( \frac{e^{-\tilde{\mu}\tilde{\eta}^*}}{e^{\tilde{\mu}} E_1(\tilde{\mu})} - 
\frac{e^{-\mu \tilde{\eta}^*}}{ e^{\mu} E_1(\mu)} \Big) dy \text{, by } \eqref{eq:horseshoe_residuals} \\
&= 1 - \int_{t}^\infty \frac{1}{y} \Big( \frac{e^{-\tilde{\mu}y}}{E_1(\tilde{\mu})} - 
\frac{e^{-\mu y}}{E_1(\mu)} \Big) dy \text{ with } y=1+\tilde{\eta}^* \\
&= 1  - \Big( \frac{E_1(\tilde{\mu} t)}{E_1(\tilde{\mu})} - \frac{E_1(\mu t)}{E_1(\mu)}\Big). \label{eq:prob_ub}
\end{flalign}

We now consider this upper bound on $x_p$ as $t \rightarrow \infty$. 
Note that for $\mu, \tilde{\mu}$ fixed with $\tilde{\mu}<\mu$,
\begin{flalign}
\lim_{t \rightarrow \infty} \Big( E_1(\tilde{\mu}) \tilde{\mu}t e^{\tilde{\mu}t} \Big( \frac{E_1(\tilde{\mu}t)}{E_1(\tilde{\mu})} - \frac{E_1(\mu t)}{E_1(\mu )} \Big) \Big) 
&= \lim_{t \rightarrow \infty} \Big( \tilde{\mu }t e^{\tilde{\mu }t} E_1(\tilde{\mu }t) \Big) \lim_{t \rightarrow \infty} \Big( 1 - \frac{E_1(\mu t) E_1(\tilde{\mu })}{E_1(\tilde{\mu } t) E_1(\mu )} \Big) \\
&= \lim_{t \rightarrow \infty} \Big( \tilde{\mu }t e^{\tilde{\mu }t} E_1(\tilde{\mu }t) \Big) \text{ as } \lim_{t \rightarrow \infty} \frac{E_1(\mu t)}{E_1(\tilde{\mu } t)} = 0  \\
&= 1 \text{ as } \lim_{x \rightarrow \infty} \big( x e^x E_1(x) \big) = 1. \label{eq:prob_limit1}
\end{flalign}
Therefore, there exists some constant $D_1>0$ such that 
\begin{equation} \label{eq:prob_limit_bound}
\Big( \frac{E_1(\tilde{\mu }t)}{E_1(\tilde{\mu })} - \frac{E_1(\mu t)}{E_1(\mu )} \Big) > \frac{D_1}{ \tilde{\mu }t e^{\tilde{\mu }t}}
\end{equation}
for all $t>0$. By \eqref{eq:prob_ub} and \eqref{eq:prob_limit_bound}, we obtain
\begin{flalign}
x_c^p &< \Big(1 - \frac{D_1}{t e^{\tilde{\mu }t}}\Big)^p \\
&= \Big( \Big(1 - \frac{1}{D_1^{-1} \tilde{\mu }t e^{\tilde{\mu }t}}\Big)^{D_1^{-1} \tilde{\mu }t e^{\tilde{\mu }t}} \Big)^{p/(D_1^{-1} \tilde{\mu }t e^{\tilde{\mu }t})} \\
&< \exp \big(-p/(D_1^{-1} \tilde{\mu }t e^{\tilde{\mu }t}) \big) \text{ as }  (1-1/x)^x < 1/e \text{ for } x \geq 1. 
\end{flalign}
For $p > 1$, now take $c= \log \big( \frac{\alpha \log p}{\tilde{\mu } K} \big)$ for any fixed constant $\alpha \in (0,1)$. 
Recall $t=K e^c$, such that $\tilde{\mu }t = \alpha \log p$. Then 
\begin{flalign}
x_c^p &< \exp \big(-p/(D_1^{-1} \tilde{\mu }t e^{\tilde{\mu }t}) \big) = \exp \big( -D_1 p^{1-\alpha}/(\alpha \log p) \big) 
< \exp \big( -D p^{1-\alpha'} \big)
\end{flalign} 
for any $\alpha' \in (\alpha, 1)$ and some constant $D$ which does not depend on $p$. 
Therefore, 
\begin{flalign} \label{eq:max_init}
\mathbb{P} \big(E_c \big| \ m_t, \tilde{m}_t \big) &> 1 - \exp \big( -D p^{1-\alpha'} \big)
\end{flalign}
for $c = \log (\alpha/(\tilde{\mu } K)) + \log \log p$. 
Recall that $K:=(\log E_1(\mu ) - \log E_1(\tilde{\mu }))/(\tilde{\mu }-\mu )$. 
We obtain
\begin{flalign}
\frac{1}{\tilde{\mu } K} &= \frac{\mu -\tilde{\mu }}{\tilde{\mu } (\log E_1(\tilde{\mu }) - \log E_1(\mu ))} \\
&= \frac{\mu /\tilde{\mu }-1}{\log E_1(\tilde{\mu }) - \log E_1(\mu )} \\
&\geq \frac{\log \mu  - \log \tilde{\mu } }{\log E_1(\tilde{\mu }) - \log E_1(\mu )} \text{ as } x-1\geq \log (x) \text{ for } x>0 \\
&= e^{\hat{\mu }} E_1(\hat{\mu }) \text{ for some } \hat{\mu } \in [\tilde{\mu }, \mu ] \text{ as } \frac{d}{d (\log \mu )} (-\log E_1(\mu )) = \frac{1}{e^{\mu } E_1(\mu )} \\
&\geq e^{\mu } E_1(\mu ) \text{ as } x \mapsto e^x E_1(x) \text{ is decreasing for all } x>0  \\
&\geq \frac{1}{2} \log(1 + 1/\mu ) \text{ as } e^{x} E_1(x) \geq \frac{1}{2} \log(1+1/x) \text{ for all } x>0 \\
&= \frac{1}{2} \log(1 + 1/\max\{\mu, \tilde{\mu} \} ) \\
&=: L.
\end{flalign}
Therefore, $c \geq \log (\alpha L) + \log \log p$.
By \eqref{eq:max_init},
\begin{flalign}
\mathbb{P} \big(\max_{j=1}^p | \log(1+\eta_{t+1,j}) - \log(1+\tilde{\eta}_{t+1,j}) | 
> \log (\alpha L) + \log \log p  \ \big| \ m_t, \tilde{m}_t \big)
%> 1 - e^{-D p^{1-\alpha'}}.
> 1 - \exp( -D p^{1-\alpha'}).
\end{flalign}
The case $\tilde{\mu } > \mu $ follows by symmetry to give \eqref{eq:max}.
\end{proof}

\begin{proof}[Proof of Proposition \ref{prop:tv_metric_upper_bound}]
Fix some $\nu \geq 1$. For any $m > 0$, let $p(\cdot | m)$ denote the component-wise
full conditional density of $\eta$ given $m$. This is given by 
\begin{equation} %\label{eq:eta_full_conditional_pdf_2}
p(\eta | m) = \frac{\exp(-m \eta)}{\eta^{\frac{1-\nu}{2}} (1+ \nu \eta)^{\frac{\nu+1}{2}}} \frac{1}{Z_m}
\end{equation}
on $(0, \infty)$, where $Z_m$ is the normalization constant. We will first show 
\begin{equation} \label{eq:tv_metric_upper_bound_componentwise}
\text{TV} \big( p(\cdot | m), p(\cdot | \tilde{m}) \big)^2 \leq \frac{(1+\nu)(e-1)}{4} | \log(m) - \log(\tilde{m}) |.
\end{equation}
By a generalization of Pinkser's inequality \citep[e.g.][]{vanErven2014renyiIEEE}, we obtain
\begin{equation}
\text{TV} \big( p(\cdot | m), p(\cdot | \tilde{m}) \big)^2 \leq \frac{1}{2 \alpha} D_{\alpha}( p(\cdot | m) \| p(\cdot | \tilde{m}) )
\end{equation}
for all $\alpha \in (0,1]$, where 
$$D_{\alpha}(p(\cdot | m) \| p(\cdot | \tilde{m})) := \frac{1}{\alpha -1} 
\log \Big( \int p(\eta | m)^\alpha p(\eta | \tilde{m})^{1-\alpha} d \eta \Big)$$
is the R\'{e}nyi divergence between probability densities $p$ and $q$ with respect to the 
Lebesgue measure. For $\alpha=1/2$, we obtain
\begin{flalign}
D_{\alpha}(p(\cdot | m) \| p(\cdot | \tilde{m})) &= -2 \log \Big(
\int_0^\infty \frac{\exp(-\frac{m + \tilde{m}}{2} \eta)}{\eta^{\frac{1-\nu}{2}} (1+ \nu \eta)^{\frac{\nu+1}{2}}} \frac{1}{\sqrt{Z_{m} Z_{\tilde{m}}}} d\eta \Big)
= \log \Big( \frac{Z_{m} Z_{\tilde{m}}}{Z_{(m+\tilde{m})/2}^2}  \Big)
\end{flalign}
Let $L(m)= \log Z_m$. Then for $\eta | m \sim p(\cdot | m)$, note that 
\begin{flalign}
\frac{d L(m)}{dm} = \frac{\frac{d Z_m}{dm}}{Z_m} = - \mathbb{E}[\eta | m] \leq 0, \text{ and } 
\frac{d^2 L(m)}{dm^2} = \frac{\frac{d^2 Z_m}{dm^2} Z_m - (\frac{d Z_m}{dm})^2}{Z_m^2} = 
\mathbb{E}[\eta^2 | m] - \mathbb{E}[\eta | m]^2 \geq 0,
\end{flalign}
such that $L(m)$ is a decreasing convex function. For $\tilde{m} \geq m$, this gives 
\begin{flalign} \label{eq:apply_convex_1}
L(m) + L(\tilde{m}) - 2 L \Big( \frac{m+\tilde{m}}{2} \Big) 
&\leq \big( L'(\tilde{m})-L'(m) \Big) \Big(\frac{\tilde{m}-m}{2} \big)\\
&= \big( \mathbb{E}[\eta | m]-\mathbb{E}[\eta | \tilde{m}] \big) \Big(\frac{\tilde{m}-m}{2}\Big).
\end{flalign}
Furthermore, note that by Lemmas \ref{lemma:eta_moments} and \ref{lemma:confluent_ratio_bound},
\begin{flalign}
\mathbb{E}[\eta | m] =\frac{1}{\nu}  \frac{\Gamma(\frac{1+\nu}{2}+1)}{ \Gamma(\frac{1+\nu}{2})} \frac{U(\frac{1+\nu}{2}+1, 1+1, \frac{m}{\nu})}{U(\frac{1+\nu}{2}, 1, \frac{m}{\nu})} \leq \frac{1}{\nu}  \frac{\Gamma(\frac{1+\nu}{2}+1)}{ \Gamma(\frac{1+\nu}{2})} \big( \frac{m}{\nu} \big)^{-1} = \frac{1+\nu}{2} m^{-1}
\end{flalign}
where $U(a,b,z)$ corresponds to the confluent hypergeometric function of 
the second kind. 
By \eqref{eq:apply_convex_1}, we therefore obtain
\begin{flalign} 
D_{\alpha}(p(\cdot | m) \| p(\cdot | \tilde{m}))
&\leq \big(\mathbb{E}[\eta | m]-\mathbb{E}[\eta | \tilde{m}] \big) \Big(\frac{\tilde{m}-m}{2}\Big) \\
& \leq \frac{1+\nu}{2} \Big(\frac{\tilde{m}-m}{2 m} \Big) \\
&= \frac{1+\nu}{4}  \big( \exp( \log \tilde{m} - \log m ) - 1) \big),
\end{flalign}
where $ \exp( |\log \tilde{m} - \log m |) - 1 \leq (e-1) |\log m - \log \tilde{m}|$
for $|\log m - \log \tilde{m}| \leq 1$. This gives 
\begin{equation} \label{eq:renyi_div}
\text{TV} \big( p(\cdot | m), p(\cdot | \tilde{m}) \big)^2 \leq
\frac{(1+\nu)(e-1)}{4} \big( \log(\tilde{m}) - \log(m) \big),
\end{equation}
for $\tilde{m} \geq m$. The case $\tilde{m} \leq m$ follows by symmetry, 
giving \eqref{eq:tv_metric_upper_bound_componentwise}.

We now extend \eqref{eq:tv_metric_upper_bound_componentwise} to the 
multivariate setting. Let $\pi(\eta | m):= \prod_{j=1}^p p(\eta_j| m_j)$ now denote the 
full conditional density of $\eta \in (0, \infty)^p$ given $m \in (0, \infty)^p$. Then we obtain
\begin{flalign} \label{eq:tv_metric_upper_bound_full}
\text{TV} \big( \pi(\cdot | m), \pi(\cdot | \tilde{m}) \big)^2 
&\leq D_{1/2}( \pi(\cdot | m) \| \pi(\cdot | \tilde{m}) ) \\
&= \sum_{j=1}^p D_{1/2}( p(\cdot | m_j) \| p(\cdot | \tilde{m}_j) ) \text{ by conditional independence } \\
&= \frac{(1+\nu)(e-1)}{4}  \sum_{j=1}^p | \log(m_j) - \log(\tilde{m}_j) | \text{ by \eqref{eq:renyi_div}}.
\end{flalign}
Finally, let $Z = (\beta, \xi, \sigma^2)$ and $\tilde{Z} = (\tilde{\beta}, \tilde{\xi}, \tilde{\sigma}^2)$, and
$m_{j} = \xi \beta_{j}^2/(2 \sigma^2)$ and $\tilde{m}_{j} = \tilde{\xi} \tilde{\beta}_{j}^2/(2 \tilde{\sigma}^2)$
for $j=1, \ldots, p$. For any such $Z, \tilde{Z}$, we obtain
\begin{flalign} 
\text{TV} \big( \mathcal{P}(Z, \cdot), \mathcal{P}(\tilde{Z} | \cdot) \big) 
&= \frac{1}{2} \int \big| p( \eta^* | Z) - p(\eta^*|\tilde{Z}) \big| p( Z^* | \eta^* ) d\eta^* dZ^* \\
&= \frac{1}{2} \int | p( \eta^* | Z) - p(\eta^* | \tilde{Z}) | d\eta^* \\
&= \frac{1}{2} \int | \pi( \eta^* | m) - \pi(\eta^* | \tilde{m}) | d\eta^* \\
&= \text{TV} \big( \pi(\cdot | m), \pi(\cdot | \tilde{m}) \big).
\end{flalign}
Therefore, by \eqref{eq:tv_metric_upper_bound_full}, we obtain 
\begin{flalign} 
\text{TV} \big( \mathcal{P}(Z, \cdot), \mathcal{P}(\tilde{Z} | \cdot) \big)^2
&\leq \frac{(1+\nu)(e-1)}{4}  \sum_{j=1}^p | \log(m_j) - \log(\tilde{m}_j)|.
\end{flalign}
\end{proof}

\begin{proof}[Proof of Proposition \ref{prop:switch_crn_meeting}]
This proposition follows from a modification of 
the proof of 
Proposition \ref{prop:one_scale_meeting}. Following the setup and notation 
in the proof of Proposition \ref{prop:one_scale_meeting}, we have
\begin{flalign}
\mathbb{P} & ( Z_{t+1} = \tilde{Z}_{t+1} \big| Z_{t}, \tilde{Z}_{t}) \\
& = \mathbb{E} \Big[ \mathbb{P} \Big( Z_{t+1} = \tilde{Z}_{t+1} \big| \eta_{t+1}, \tilde{\eta}_{t+1}, Z_{t}, \tilde{Z}_{t} \Big)
\Big| Z_{t}, \tilde{Z}_{t} \Big] \\
& = \mathbb{E} \Big[ \mathbb{P} \Big( Z_{t+1} = \tilde{Z}_{t+1} \big| \eta_{t+1}, \tilde{\eta}_{t+1} \Big)
\Big| Z_{t}, \tilde{Z}_{t} \Big] \\
& = \mathbb{E} \Big[ \mathrm{1}_{ \{ \eta_{t+1} = \tilde{\eta}_{t+1} \} } \Big| Z_{t}, \tilde{Z}_{t} \Big]
\text{, as } (Z_{t+1},\tilde{Z}_{t+1}) | (\eta_{t+1}, \tilde{\eta}_{t+1}) \text{ follows common random numbers coupling} \\
& = \mathbb{P} \big( \eta_{t+1} = \tilde{\eta}_{t+1} \Big| Z_{t}, \tilde{Z}_{t} \big) \\
& = \mathbb{P} \big( \eta_{1,t+1} = \tilde{\eta}_{1,t+1} \Big| Z_{t}, \tilde{Z}_{t} \big) 
\prod_{j=2}^p \mathbb{P} \big( \eta_{t+1,j} = \tilde{\eta}_{t+1,j} \Big| 
\underset{1 \leq i \leq j-1}{ \bigcap } \{ \eta_{i,t+1} = \tilde{\eta}_{i,t+1} \}, Z_{t}, \tilde{Z}_{t} \big) \\
& \qquad \qquad \qquad \qquad \qquad \qquad \qquad \qquad \qquad 
\qquad \qquad \qquad \qquad \text{ under the switch-to-CRN coupling} \\
& \geq \prod_{j=1}^p \frac{U \big( \frac{1+\nu}{2}, 1, \frac{R^{1/d}}{\nu}\big)}{U \big( \frac{1+\nu}{2}, 1, \frac{R^{-1/c}}{\nu}\big)}
\text{ by Lemma \ref{lemma:eta_coupling_bound}, as in the proof of Proposition \ref{prop:one_scale_meeting} for } Z_t, \tilde{Z}_t \in S(R),\\
&= \Big( \frac{U \big( \frac{1+\nu}{2}, 1, \frac{R^{1/d}}{\nu}\big)}{U \big( \frac{1+\nu}{2}, 1, \frac{R^{-1/c}}{\nu}\big)} \Big)^p =: \epsilon,
\end{flalign}
where $\epsilon \in (0,1)$ since $R > 1$.
\end{proof}

\section{Additional Algorithms} \label{appendices:algos}

\begin{algorithm}%[!htb]
\DontPrintSemicolon
\KwIn{$ C_t := (\beta_t, \eta_t, \sigma_t^2, \xi_t) \in \mathbb{R}^p \times \mathbb{R}^p_{> 0} \times \mathbb{R}_{> 0} \times \mathbb{R}_{> 0}$. } 
\begin{enumerate}
\item Sample $\eta_{t+1} | \beta_{t}, \sigma^2_{t}, \xi_{t}$ component-wise 
independently using Algorithm \ref{algo:slice_sampling_1}, targeting 
\begin{equation} \label{eq:eta_full_conditionals}
\pi(\eta_{t+1} | \beta_{t} , \sigma^2_{t}, \xi_{t}) \propto \prod_{j=1}^p 
\pi\big( \eta_{t+1,j} | m_{t,j} := \frac{\xi_{t} \beta_{t,j}^2}{2  \sigma^2_{t}} \big) \propto \prod_{j=1}^p \frac{e^{-m_{t,j} \eta_{t+1,j}}}{\eta_{t+1,j}^{\frac{1-\nu}{2}}(1+\nu \eta_{t+1,j})^{\frac{\nu + 1}{2}}} .
\end{equation}
\item Sample $\xi_{t+1},\sigma^2_{t+1}, \beta_{t+1}$ given $\eta_{t+1}$ as follows:
\begin{enumerate}
\item Sample $\xi_{t+1}$ using Metropolis--Rosenbluth--Teller--Hastings with step-size $\sigma_{\text{MRTH}}$: \;
Given $\xi_{t}$, propose $\log(\xi^*) \sim \mathcal{N}(\log(\xi_{t}), \sigma^2_{\text{MRTH}})$. \;
Calculate acceptance probability
\[q =  \frac{L(y | \xi_*, \eta_{t+1} ) \pi_{\xi}(\xi_*)  }{L(y | \xi_{t}, \eta_{t+1} ) \pi_{\xi}(\xi_{t})} \frac{\xi^*}{\xi_{t}},\]
%\begin{equation} 
%\log(q) = \log \Big( \frac{L(y | \xi_*, \eta_{t+1} ) \pi_{\xi}(\xi_*)  }{L(y | \xi_{t}, \eta_{t+1} ) \pi_{\xi}(\xi_{t})} \frac{\xi^*}{\xi_{t}} \Big),
%\end{equation}
where $\pi_{\xi}(\cdot) $ is the prior for $\xi$, 
$ M_{\xi,\eta} := I_n + \xi^{-1} X\, \text{Diag}(\eta^{-1})\, X^T $ and
\begin{equation} \log(L(y | \xi, \eta )) =  - \frac{1}{2} \log (|M_{\xi,\eta}|) - \frac{a_0 + n}{2} \log( b_0 + y^T M_{\xi,\eta}^{-1} y ). \label{eq:Half_t_Exact_algo_ll}
\end{equation}
Set $\xi_{t+1} := \xi^*$ with probability $\min(1,q)$, otherwise set $\xi_{t+1} := \xi_{t}$.
\item Sample $\sigma^2_{t+1}$ given $\xi_{t+1},\eta_{t+1}$:
\begin{equation}
  \sigma^2_{t+1} | \xi_{t+1}, \eta_{t+1} \sim \InvGamma\bigg(\frac{a_0+n}{2}, \frac{y^T M_{\xi_{t+1},\eta_{t+1}}^{-1} y + b_0}{2}\bigg).
\end{equation}
\item Sample $\beta_{t+1}$ given $\sigma^2_{t+1}, \xi_{t+1}, \eta_{t+1}$ using Algorithm \ref{algo:fast_mvn_bhattacharya_1} in Appendix \ref{appendices:algos}, targeting
\begin{equation}
\beta_{t+1} | \sigma^2_{t+1}, \xi_{t+1}, \eta_{t+1} \sim \mathcal{N} \big( \Sigma^{-1} X^T y,  \sigma^2_{t+1} \Sigma^{-1} \big)
\text{ for } \Sigma = X^TX + \xi_{t+1} \text{Diag}(\eta_{t+1} ).
\end{equation}
\end{enumerate}
\end{enumerate}
 \Return $ C_{t+1} := (\beta_{t+1}, \eta_{t+1}, \sigma^2_{t+1}, \xi_{t+1})$.
 \caption{MCMC for Bayesian shrinkage regression with Half-$t(\nu)$ priors.}
 \label{algo:Half_t_Exact}
\end{algorithm}

\begin{algorithm}[!htb]
\DontPrintSemicolon
\SetAlgoLined
\KwResult{Slice sampler targeting $p(\eta_j | m_j) \propto (\eta_j^{\frac{1-\nu}{2}}(1+\nu \eta_j)^{\frac{\nu + 1}{2}})^{-1} e^{-m_j \eta_j}$ on $(0, \infty)$.}
\KwIn{$m_j > 0$, state $\eta_{t,j} > 0$.}
\begin{enumerate}
\item Sample $U_{t,j} | \eta_{t,j} \sim \Uniform(0, (1+\nu \eta_{t,j})^{-\frac{\nu + 1}{2}})$.
\item Sample $\eta_{t+1,j} | U_{t,j} \sim P_j$, where $P_j$ 
has unnormalized density $\eta \mapsto \eta^{s-1} e^{ - m_j \eta}$ on 
$(0, T_{t,j})$, with $T_{t,j}=(U_{t,j}^{-2/(1+\nu)}-1)/\nu$ and $s=(1+\nu)/2$.
We can sample $\eta_{t+1,j}$ perfectly by setting 
\begin{equation}
\eta_{t+1,j} = \frac{1}{m_j} \gamma_{s}^{-1} \Big( \gamma_{s}(m_jT_{t,j}) U^*
\Big) \quad\text{where}\quad U^* \sim \Uniform(0,1),
\end{equation}
and where $\gamma_{s}(x) := \Gamma(s)^{-1} \int_0^x t^{s-1} e^{-t} dt \in [0,1] $ is 
the cumulative distribution function of a $\mathrm{Gamma}(s,1)$ distribution.
%the regularized incomplete lower Gamma function.
\end{enumerate}
\Return $\eta_{t+1,j}$.
\caption{Slice Sampling for Half-$t(\nu)$ prior.}
\label{algo:slice_sampling_1}
\end{algorithm}

\begin{algorithm}[H]
\DontPrintSemicolon
\caption{Maximal coupling with independent residuals.}
\KwIn{Distributions $P$ and $Q$ with respective densities $p$ and $q$}
Sample \(X \sim P\), and \(W \sim \Uniform(0,1)\). \;
\lIf{ \( p(X) W \leq q(X) \) }{
   set \(Y=X\) and \Return \( (X,Y) \).
   } \lElse {
   sample \( \tilde{Y} \sim Q\) and \(\tilde{W} \sim \Uniform(0,1) \) until  \( q(\tilde{Y}) \tilde{W} > p(\tilde{Y}) \). Set \(Y=\tilde{Y}\) and \Return \( (X,Y) \).
   }
\label{algo:max_coupling}
\end{algorithm}

\begin{algorithm}
\DontPrintSemicolon
\SetAlgoLined
\KwResult{Sample from $\mathcal{N}( (X^T X + \xi \text{Diag}(\eta))^{-1}X^Ty, \sigma^2(X^T X + \xi \text{Diag}(\eta))^{-1} )$}
 \begin{enumerate}
\item Sample $r \sim \mathcal{N}(0, I_p)$, $\delta \sim \mathcal{N}(0,I_n)$
\item Let $u = \frac{1}{\sqrt{\xi \eta}}r$, $v = Xu + \delta$. Calculate $v^* = M^{-1}(\frac{y}{\sigma} -v )$ for $ M = I_n + (\xi)^{-1} X \text{Diag}(\eta^{-1}) X^T $. Let $U = (\xi \eta )^{-1} X^T$. Set $\beta = \sigma ( u + Uv^* ).$
 \caption{Bhattacharya's algorithm with common random numbers.}
\end{enumerate}
\caption{Fast Sampling of Normals for Gaussian scale mixture priors \citep{bhattacharya2016fastBIOMETRIKA}.}
\label{algo:fast_mvn_bhattacharya_1}
\end{algorithm}

\begin{algorithm}
\DontPrintSemicolon
\SetAlgoLined
\KwIn{$\eta_{t+1},\xi_{t+1}, \sigma_{t+1}^2, \tilde{\eta}_{t+1}, \tilde{\xi}_{t+1}, \tilde{\sigma}_{t+1}^2 $.}
\begin{enumerate}
\item Sample $r \sim \mathcal{N}(0, I_p)$, $\delta \sim \mathcal{N}(0,I_n)$.
\item Let $u = r \big/ {\sqrt{\xi_{t+1} \eta_{t+1}}}$ and $\tilde{u} = r \big/ {\sqrt{\tilde{\xi}_{t+1} \tilde{\eta}_{t+1}}}$, and
$v = Xu + \delta$ and $ \tilde{v} = X \tilde{u} + \delta$. 
\item Calculate $v^{*} = M^{-1} \big( \frac{y}{\sigma_{t+1}} - v \big)$
and $\tilde{v}^{*} = \tilde{M}^{-1} \big( \frac{y}{ \tilde{\sigma}_{t+1}} - \tilde{v} \big)$
where \\
$ M = I_n + \xi_{t+1}^{-1} X \text{Diag}(\eta_{t+1}))^{-1} X^T $ and
$ \tilde{M} = I_n + \tilde{\xi}_{t+1}^{-1} X \text{Diag}( \tilde{\eta}_{t+1}))^{-1} X^T $.
 \item Let $U = (\xi_{t+1} \eta_{t+1} )^{-1} X^T$ and $\tilde{U} = (\tilde{\xi}_{t+1} \tilde{\eta}_{t+1} )^{-1} X^T$, and 
 set $\beta_{t+1} = \sigma_{t+1} ( u + Uv^{*} )$ and 
 $\tilde{\beta}_{t+1} = \tilde{\sigma}_{t+1} ( \tilde{u} + \tilde{U} \tilde{v}^{*} )$.
\end{enumerate}
\Return $( \beta_{t+1}, \tilde{\beta}_{t+1} )$
\caption{Coupled Normals for Gaussian scale mixture priors with common random numbers.}
\label{algo:fast_mvn_bhattacharya_crn_coupling}
\end{algorithm}

\begin{algorithm}
\SetAlgoLined
\DontPrintSemicolon
\KwIn{$\xi_{t}, \tilde{\xi}_{t} > 0$ and $\eta_{t+1}, \tilde{\eta}_{t+1} \in \mathbb{R}_{>0}^p$.}
%\KwResult{Sample $(\xi_{t+1}^{(1)}, \xi_{t+1}^{(2)}) | \eta^{(1)}, \tilde{\eta}, \xi_{t}, \tilde{\xi}_{t}$ jointly such that each $\xi_{t+1}^{(i)} | \eta^{(i)}, \xi_{t}^{(i)}$ are marginally distributed from the Metropolis--Rosenbluth--Teller--Hastings kernel in Algorithm \ref{algo:slice_sampling_1}, and $\mathbb{P}(\xi_{t+1}^{(1)} = \xi_{t+1}^{(2)} | \eta^{(1)}, \eta^{(2)}, \xi_{t}, \tilde{\xi}_{t}) > 0$.}
Sample joint proposal
\begin{equation}
\Big( \log(\xi^{*}), \log(\tilde{\xi}^{*}) \Big) \Big| \xi_{t} , \tilde{\xi}_t \sim \gamma_{max} \Big(\mathcal{N}(\log(\xi_{t}), \sigma_{\text{MRTH}}^2), \mathcal{N}(\log(\tilde{\xi}_{t}), \sigma_{\text{MRTH}}^2) \Big).
\end{equation}
where $\sigma_{\text{MRTH}}$ is the Metropolis--Rosenbluth--Teller--Hastings step-size, and $\gamma_{max}$
is a maximal coupling with independent residuals (Algorithm \ref{algo:max_coupling}). \;
Calculate log-acceptance probabilities
\begin{equation} 
\log(q) = \log \Big( \frac{L(y | \xi^{*}, \eta_{t+1} ) \pi_{\xi}(\xi^{*})  }{L(y | \xi_{t}, \eta_{t+1} ) \pi_{\xi}(\xi_{t})} \frac{\xi^{*}}{\xi_{t}} \Big) \ \text{and} \
\log(\tilde{q}) = \log \Big( \frac{L(y | \tilde{\xi}^{*}, \tilde{\eta}_{t+1} ) \pi_{\xi}( \tilde{\xi}^{*})  }{L(y | \tilde{\xi}_{t}, \tilde{\eta}_{t+1} ) \pi_{\xi}( \tilde{\xi}_{t})} \frac{ \tilde{\xi}^{*}}{\tilde{\xi}_{t}} \Big).
\end{equation}
with log-likelihoods $\log( L(y | \xi, \eta_{t+1}) )$ and $\log( L(y | \tilde{\xi}, \tilde{\eta}_{t+1}) )$ as 
defined in Algorithm \ref{algo:Half_t_Exact}, Equation \eqref{eq:Half_t_Exact_algo_ll}. \;
Sample $U^* \sim \Uniform(0,1)$. \;
\leIf{$U^* \leq q$}{set $\xi_{t+1} := \xi^{*}$}{set $\xi_{t+1} := \xi_{t}$.}
\leIf{$U^* \leq \tilde{q}$}{set $\tilde{\xi}_{t+1} := \tilde{\xi}^{*}$}{set $\tilde{\xi}_{t+1}:= \tilde{\xi}_{t}$.} 
\Return $( \xi_{t+1}, \tilde{\xi}_{t+1} )$.
\caption{Coupled Metropolis--Rosenbluth--Teller--Hastings for $\xi | \eta$ updates with exact meetings.} 
\label{algo:xi_MH_exact_meeting_coupling} 
\end{algorithm}

\begin{algorithm}
\DontPrintSemicolon
\SetAlgoLined
\KwIn{$\nu \in \mathbb{R}^p_{>0}$, discretized grid $G$ on the compact support of $\pi_\xi(\cdot)$}
Perform eigenvalue decomposition $X \text{Diag}(\eta)^{-1} X^T=Q \Lambda Q^T$ for 
orthogonal matrix $Q$ and diagonal matrix $\Lambda$. \;
Calculate unnormalized probability mass function
\begin{flalign} \label{eta_full_conditional_formula}
  \widetilde{pmf}(\xi | \eta) := \Big( \prod_{i=1}^n ( 1 + \xi^{-1} \Lambda_{i,i} ) \Big)^{-1/2} \Big( b_0 + \sum_{i=1}^n [Qy]^2_i \frac{1}{1 + \xi^{-1} \Lambda_{i,i} } \Big)^{-\frac{a_0+n}{2}} \pi_{\xi}(\xi)
\end{flalign}
for each $\xi \in G$, where $[Qy]_i$ is the $i^{th}$ entry of $Qy \in \mathbb{R}^n$. \;
Normalize to obtain probability mass function: $pmf(\xi | \eta)  = \frac{\widetilde{pmf}(\xi | \eta)}{\sum_{\xi \in G} \tilde{pmf}(\xi | \eta)}  $. \;
Sample $\xi \sim pmf( \cdot | \eta)$ on $G$ using probability inverse transform. \;
\Return $\xi .$
\caption{Perfect sampling of $\xi | \eta$.}
\label{algo:xi_perfect_sampling}
\end{algorithm}

\begin{algorithm}
\DontPrintSemicolon
\SetAlgoLined
\KwIn{$\nu, m_j := \xi \beta_j^2 / (2 \sigma^2) > 0$.}
Sample $W^* \sim \Uniform(0,1)$ \;
Set 
\begin{equation} \label{eq:eta_perfect_sample_pit}
\eta_j = \frac{1}{\nu} U_{\frac{\nu+1}{2},1,\frac{m_j}{\nu}}^{-1} \Big( W^*  U_{\frac{\nu+1}{2},1,\frac{m_j}{\nu}}(\infty) \Big)
\end{equation}
where $U_{a,b,z}(t) := \frac{1}{\Gamma(a)} \int_0^t x^{a-1}(1+x)^{b-a-1} e^{-zx} dx$ is defined as the \textit{lower incomplete} confluent hypergeometric function of the second kind, such that $U_{a,b,z}(\infty) = U(a,b,z)$ gives the confluent hypergeometric function of the second kind. \;
\Return $ \eta_j.$
\caption{Perfect sampling of $\eta_j | \beta_j, \sigma^2, \xi$.}
\label{algo:eta_perfect_sampling}
\end{algorithm}

%\newpage

\section{Algorithm Derivations} \label{appendices:algo_derivations}
\paragraph{Gibbs sampler for Half-$t(\nu)$ priors (Algorithm \ref{algo:Half_t_Exact}).}
Derivation of the blocked Gibbs sampling algorithm (Algorithm \ref{algo:Half_t_Exact})
is given in \citet{johndrow2020scalableJMLR} for the Horseshoe
prior. For Half-$t(\nu)$ priors, it remains to check the validity of the 
slice sampler in Algorithm \ref{algo:slice_sampling_1}. Working 
component-wise, let $p(\eta | m)$ denote the conditional density of 
$\eta $ given $m>0$. Then,
\begin{flalign}
p(\eta | m) &\propto \eta^{\frac{\nu-1}{2}} (1+\nu \eta)^{-\frac{1+\nu}{2}} e^{-m \eta} \mathbbm{1}_{(0, \infty)}(\eta) = \int_{u=0}^\infty \eta^{\frac{\nu-1}{2}} e^{-m \eta} \mathbbm{1}_{(0, \infty)}(\eta) \mathbbm{1}_{ \big(0, (1+\nu \eta)^{-\frac{1+\nu}{2}} \big)}(u) du. %\\
%&\propto \int_{u=0}^\infty \bar{p}(\eta, u | m) du
\end{flalign}
Let $\bar{p}(\eta, u | m) \propto \eta^{\frac{\nu-1}{2}} e^{-m \eta} 
\mathbbm{1}_{(0, \infty)}(\eta) \mathbbm{1}_{ \big(0, (1+\nu \eta)^{-\frac{1+\nu}{2}} \big)}(u)$
be the conditional density of the augmented random 
variable $(\eta, u)$ given $m$, such that $p(\eta | m) = \int_{u=0}^\infty\bar{p}(\eta, u | m) du $.
Then, the slice sampler in Algorithm \ref{algo:slice_sampling_1} corresponds to 
a Gibbs sampler targeting $\bar{p}(\eta, u | m)$:
\begin{flalign}
  \bar{p}(u| \eta, m) &\sim \Uniform\big( 0, (1+\nu \eta)^{-\frac{1+\nu}{2}}  \big), \\
\bar{p}(\eta| u, m) &\propto \eta^{\frac{\nu-1}{2}} e^{-m \eta} \mathbbm{1}_{ \big(0, \frac{u^{-2/(1+\nu)}-1}{\nu} \big) }(\eta).
\end{flalign}
Note that $\int_0^x \eta^{s-1} e^{-m \eta} d \eta = m^{-s} \Gamma(s) \gamma_{s}(mx) $, 
where $\gamma_{s}(x) := \frac{1}{\Gamma(s)} \int_0^x t^{s-1} e^{-t} dt \in [0,1] $
is the regularized incomplete lower Gamma function. For 
$T = (u^{-2/(1+\nu)}-1)/\nu, s= \frac{1+\nu}{2} $, 
we can now obtain the cumulative density function of $\bar{p}(\cdot | u, m)$, given
by $\int_0^x \bar{p}(\eta| u, m) = \gamma_{s}(mx) / \gamma_{s}(mT)$
for $0\leq x \leq T$. By probability inverse transform, we can now sample perfectly 
from $\eta| u, m$, such that 
\begin{equation}
\eta := \frac{1}{m} \gamma_{s}^{-1} \big( \gamma_{s}(mT) U^* \big) \text{ for } U^* \sim \Uniform(0,1)
\end{equation}
as in Step $(2)$ of Algorithm \ref{algo:slice_sampling_1}. 

\paragraph{Perfect sampling of $\xi | \eta$ (Algorithm \ref{algo:xi_perfect_sampling}).}
In Algorithm \ref{algo:Half_t_Exact}, we have 
\begin{flalign}
p(\xi | \eta ) \propto |M_{\xi, \eta}|^{-1/2} \big(b_0 + y^T M_{\xi, \eta}^{-1} y \big)^{-\frac{a_0+n}{2}} \pi_\xi(\eta)
\end{flalign}
for $M_{\xi, \eta} = I_n + \xi^{-1} X \text{Diag}( \eta )^{-1} X^T$ and
$\pi_\xi(\cdot)$ the prior density on $\xi$. Consider the
eigenvalue decomposition of $X \text{Diag}(\eta)^{-1} X^T$, such that
$X \text{Diag}(\eta)^{-1} X^T=Q \Lambda Q^T$ 
for $Q$ orthogonal and $\Lambda$ diagonal. Then, 
\begin{flalign}
|M_{\xi, \eta}| &= | Q (I_n + \xi^{-1} \Lambda) Q^T | =  | I_n + \xi^{-1} \Lambda| = \prod_{i=1}^n \big( 1 + \xi^{-1} \Lambda_{i,i} \big) \\
y^T M_{\xi, \eta}^{-1} y &= y^T Q (I_n + \xi^{-1} \Lambda)^{-1} Q^T y = \sum_{i=1}^n [Q^Ty]_i^2 \big( 1 + \xi^{-1} \Lambda_{i,i} \big)^{-1}
\end{flalign}
This allows the discretized perfect sampling of $\xi $ given $\eta$
using probability inverse transform as in Algorithm \ref{algo:xi_perfect_sampling}, 
with $\mathcal{O}(n^3)$ computational cost arising from
eigenvalue decomposition of $X \text{Diag}(\eta)^{-1} X^T$.

\paragraph{Perfect sampling of $\eta_j | \beta_j, \sigma^2, \xi$ (Algorithm \ref{algo:eta_perfect_sampling}).}
Working component-wise, let 
\begin{flalign}
p(\eta | m) \propto \eta^{\frac{\nu-1}{2}} (1+\nu \eta)^{-\frac{1+\nu}{2}} e^{-m \eta}
\end{flalign}
denote the conditional density of $\eta \in (0,\infty)$ 
given $m>0$. We can calculate the corresponding
 cumulative density function (see proof of Proposition \ref{prop:slice_sampler_tv}), given by
\begin{flalign}
\int_0^K p(\eta | m) d\eta = \frac{U_{\frac{\nu+1}{2}, 1, \frac{m}{\nu} }(\nu K)}{ U(\frac{\nu+1}{2}, 1, \frac{m}{\nu}) }.
\end{flalign}
where $U_{a,b,z}(t) := \frac{1}{\Gamma(a)} \int_0^t x^{a-1}(1+x)^{b-a-1} e^{-zx} dx$ 
is defined as the \textit{lower incomplete} confluent hypergeometric function of the second kind.
Algorithm \ref{algo:eta_perfect_sampling} and \eqref{eq:eta_perfect_sample_pit}
directly follow from probability inverse transform. 

\end{document}